\makeatletter\hbadness100\newif\ifworkhard\workhardtrue

\documentclass[reqno]{amsart}

\usepackage{mathtools,amssymb,yhmath}\allowdisplaybreaks

\usepackage[utf8]{inputenc}
	\makeatletter\let\DUC\DeclareUnicodeCharacter\DUC{8A54}%詔
	{\bgroup\def\UTFviii@defined##1{\expandafter\EDICTUM\string##1;}}
	\def\EDICTUM#1:#2;{\egroup\DUC{\UTFviii@hexnumber{\decode@UTFviii#2\relax}}}
	
	詔 α\alpha		詔 β\beta		詔 γ\gamma		詔 δ\delta		
	詔 ϵ\epsilon		詔 ζ\zeta		詔 η\eta			詔 θ\theta		
	詔 ι\iota		詔 κ\kappa		詔 λ\lambda		詔 μ\mu			
	詔 ν\nu			詔 ξ\xi			詔 π\pi			詔 ρ\rho			
	詔 σ\sigma		詔 τ\tau			詔 υ\upsilon		詔 ϕ\phi			
	詔 χ\chi			詔 ψ\psi			詔 ω\omega		詔 Γ\Gamma		
	詔 Δ\Delta		詔 ε\varepsilon	詔 Θ\Theta		詔 ϑ\vartheta	
	詔 Λ\Lambda		詔 Ξ\Xi			詔 Π\Pi			詔 ϖ\varpi		
	詔 ϱ\varrho		詔 Σ\Sigma		詔 ς\varsigma	詔 Υ\Upsilon		
	詔 Φ\Phi			詔 φ\varphi		詔 Ψ\Psi			詔 Ω\Omega		
	詔ℱ{\mathcal F}	詔ℐ{\mathcal I}	詔𝒳{\mathcal X}	詔𝒴{\mathcal Y}	
	詔 ℂ{\mathbb C}	詔 𝔻{\mathbb D}	詔 𝔼{\mathbb E}	詔 𝔽{\mathbb F}	
	詔 𝔾{\mathbb G}	詔 𝕀{\mathbb I}	詔 𝕂{\mathbb K}	詔 ℙ{\mathbb P}	
	詔 𝕏{\mathbb X}	詔 ℤ{\mathbb Z}	
	\DeclareMathAlphabet\mathsi{T1}\sfdefault\mddefault\sldefault
	詔 𝘈{\mathsi A}	詔 𝘉{\mathsi B}	詔 𝘊{\mathsi C}	詔 𝘋{\mathsi D}	
	詔 𝘌{\mathsi E}	詔 𝘏{\mathsi H}	詔 𝘒{\mathsi K}	詔 𝘗{\mathsi P}	
	詔 𝘚{\mathsi S}	詔 𝘝{\mathsi V}	詔 𝘞{\mathsi W}	詔 𝘡{\mathsi Z}	
	詔 𝘢{\mathsi a}	詔 𝘣{\mathsi b}	詔 𝘤{\mathsi c}	詔 𝘦{\mathsi e}	
	詔 𝘧{\mathsi f}	詔 𝘨{\mathsi g}	
	詔 √\sqrt		
	\def\({\bigl(\color{red!50!black}}\def\){\color{teal!50!black}\bigr)}
	詔 （{\Bigl(}	詔 ）{\Bigr)}	詔 ［{\bigl[}	詔 ］{\bigr]}	
	詔 「{\Bigl[}	詔 」{\Bigr]}	詔 ｛{\bigl\{}	詔 ｝{\bigr\}}	
	詔 『{\Bigl\{}	詔 』{\Bigr\}}	
	詔 ⌈\lceil		詔 ⌉\rceil		詔 ⌊\lfloor		詔 ⌋\rfloor		
	\def\|{\mathrel\Vert}			詔 ‖{\mathrel\Big\Vert}			
	詔 ｜{\mid\nobreak}				詔 ；{\mathrel;\nobreak}			
	詔 ＼{\mathord\setminus}			詔 、{\setminus\nobreak}			
	詔 ℓ\ell			詔 ∂\partial		
	詔 ˆ\hat			詔 ¯\bar			
	詔 ∏\prod		詔 ∑\sum			詔 ∫\int			
	詔 ±\pm			詔 ·\cdot		詔 ×\times		詔 ÷\frac		
	詔 •\bullet		詔 ∧\wedge		詔 ∨\vee			詔 ∩\cap			
	詔 ∪\cup			詔 ⊕\oplus		詔 ⊗\otimes		
	詔 ¬\neg			詔 ∞\infty		詔 ⊤\top			詔 ♭\flat		
	詔 ♮\natural		詔 ♯\sharp		
	詔 ←\gets		詔 →\to			詔 ↦\mapsto		
	詔 ↾{\mathord\upharpoonright}	
	詔 ∈\in			詔 ∉\notin		詔 ∋\ni			詔 ≈\approx		
	詔 ≔\coloneqq	詔 ≠\neq			詔 ≡\equiv		詔 ≤\leqslant	
	詔 ≥\geqslant	詔 ⊂\subset		詔 ⟶\longrightarrow			
	詔 ⟼\longmapsto				
	詔 …{,\allowbreak\dotsc,\allowbreak}
	詔 †\inlinetext \def\inlinetext#1†{\text{#1}}
	詔 ©\inlinecite \def\inlinecite#1©{\cite{#1}}
	
	\def\tiltbox#1#2#3{\mkern-#1mu\hbox{\pdfsave
		\pdfsetmatrix{1 0 .2 1}\rlap{$#2$}\pdfrestore}\mkern#3mu}
	\def\Wtout{\tiltbox{2.2}W{21.9}^t_\mathrm{\mkern-5muout}}
	\def\Wnot{\tiltbox{2.2}W{21.9}}	\def\Wt{\tiltbox{2.2}W{21.9}^t}
	\def\Xt{\tiltbox0X{18.8}^t}		\def\Yt{\tiltbox{1}Y{16.9}^t}
	\def\g#1{g^{#1}}				\def\w#1{^{(#1)}}
	\def\M#1{_{(#1)}}				\def\wg#1{_𝔾^{(#1)}}
	\def\inp{_\mathrm{in}}			\def\out{_\mathrm{out}}
	\def\Er{E_\mathrm r}			
	\def\Ec{E_\mathrm c}			\def\Vc{V_\mathrm c}
	\def\P{P_\mathrm e}				\def\Z{Z_{\operatorname{mad}}}
	\def\S{S_{\max} }				\def\Git{G^{-\!⊤}}
	\def\fz#1{f_{#1Z}}				\def\fs#1{f_{#1S}}
	\def\lol{[\mkern1mu{}^1_1{}^0_1]}
	\def\dt{\frac d{dt}}			\def\dtdt{\frac{d^2}{dt^2}}
	\DeclareMathOperator\wt{wt}		\DeclareMathOperator\tr{tr}
		\DeclareMathOperator\GL{GL}
	\DeclarePairedDelimiter\abs\lvert\rvert
	\DeclarePairedDelimiter\norm\lVert\rVert
	
	\renewcommand\[{\@ifstar{\begin{equation*}}{\begin{equation}}}
	\renewcommand\]{\@ifstar{\end  {equation*}}{\end  {equation}}}
	\def\steplabel#1{\stepcounter{equation}\tag\theequation\label{#1}}
	\def\repeattag#1{\tag{\ref{#1}'s copy}}

\usepackage{tikz,tikz-cd,pgfplotstable,booktabs,colortbl}
	\usetikzlibrary{calc,graphs}
	\definecolor{recyco}{HTML}{167629}	\definecolor{xbec}{HTML}{BBEECC}
	\definecolor{xdecode}{HTML}{DEC0DE}	\definecolor{xflatten}{HTML}{F1A710}
	\let\PMS\pgfmathsetmacro			\let\PMP\pgfmathparse
	\let\PMDF\pgfmathdeclarefunction
	\tikzset{
		every picture/.style={cap=round,join=round},
		CH/.pic={\tiny\draw[nodes={draw,inner sep=1pt}]
				node[left](X){(#1)}(X.west)--+(-.1,0)node[left]{$W$};},
		FIH/.pic={\tiny\node[right,draw,inner sep=1pt]{(#1)};}
	}
	\pgfplotsset{compat/show suggested version=false,compat=1.13} % arxiv = 1.13

\usepackage[unicode,pdfusetitle,colorlinks]{hyperref}
	\hypersetup{
		linkcolor=red!50!black,citecolor=green!50!black,
		pdfsubject=Information Theory (cs.IT),
		pdfcreator=LaTeX with redstone circuits,
	}
	\def\TP#1$#2${\texorpdfstring{$#2$}{#1}}
	\def\U#1+{\unichar{"#1}}

\usepackage[capitalize]{cleveref}
	
	\newtheorem{thm}{Theorem}\Crefname{thm}{Theorem}{Theorems}
	\def\newthm#1#2#3{\newtheorem{#1}[thm]{#2}\Crefname{#1}{#2}{#3}}
	\newthm{lem}{Lemma}{Lemmas}	\newthm{cor}{Corollary}{Corollaries}
	\newthm{cla}{Claim}{Claims}	\newthm{pro}{Proposition}{Propositions}

\begin{document}

{\catcode`\ =12 \catcode`\^^M=13\def^^M{^^J}\catcode`\\=12\message{
	       __                      __                                
	      / /_  __  ______  ____  / /____  ____  __  __________      
	     / __ \/ / / / __ \/ __ \/ __/ _ \/ __ \/ / / / ___/ _ \     
	    / / / / /_/ / /_/ / /_/ / /_/  __/ / / / /_/ (__  )  __/     
	   /_/ /_/\__, / .___/\____/\__/\___/_/ /_/\__,_/____/\___/      
	         /____/_/                                                
}}

\title{%
				Polar Codes' Simplicity, Random Codes' Durability%
}
\author{%
				Hsin-Po Wang and Iwan Duursma%
}
\thanks{%
				University of Illinois at Urbana--Champaign;
				hpwang2, duursma @illinois.edu%
}

\begin{abstract}
	Over any discrete memoryless channel, we build codes such that:
	for one, their block error probabilities and code rates
	scale like random codes';
	and for two, their encoding and decoding complexities
	scale like polar codes'.
	Quantitatively, for any constants $π,ρ>0$ such that $π+2ρ<1$,
	we construct a sequence of error correction codes with
	block length $N$ approaching infinity,
	block error probability $\exp(-N^π)$,
	code rate $N^{-ρ}$ less than the Shannon capacity,
	and encoding and decoding complexity $O(N\log N)$ per code block.
	The putative codes take uniform $ς$-ary messages
	for sender's choice of prime $ς$.
	The putative codes are optimal in the following manner:
	Should $π+2ρ>1$, no such codes exist for generic channels
	regardless of alphabet and complexity.
\end{abstract}

\advance\baselineskip\glueexpr0ptplus1ptminus1pt/32
\advance    \lineskip\glueexpr0ptplus1ptminus1pt/32
\advance     \parskip\glueexpr0ptplus1ptminus1pt/16
\advance    \floatsep\glueexpr0ptplus1ptminus1pt*0
\advance\textfloatsep\glueexpr0ptplus1ptminus1pt*0

\maketitle

%\vskip\glueexpr0ptplus1ptminus1pt/2

\section{Introduction}

	Richard~W.~Hamming is one of the first few people who had the idea that
	by grouping information in blocks with redundancies,
	a calculating machine can correct errors by its own
	and proceed to the next command instead of halting.
	Their solution, now called Hamming codes, is found in \cite{Ham50}.
	Claude~E.~Shannon, a colleague of Hamming in Bell labs,
	theorized the \emph{communication channels} and showed that
	a channel associates to a number called \emph{capacity}, which represents
	the ultimate limit of the efficiency of communications over that channel.
	
	To brief the rest of the history, we follow the analogy \cite{AW14} used.
	Shannon's eternal result, noisy channel coding theorem \cite{Shannon48},
	is considered the analog of the law of large numbers (LLN).
	The theorem implies that there exists a sequence
	of longer and longer block codes whose
	block error probabilities approach $0$ and code rates approach the capacity,
	which is analogous to that the empirical average of random variables
	is close to the mean with high probability.
	Robert~G.~Gallager, Shannon, Robert~M.~Fano, and followers
	extended the LLN result by looking at how
	the block error probability $\P$ scales when the code rate $R$ is fixed.
	They showed that the error probability $\P$ scales like $\exp(-\Er(R)N)$.
	Here $N$ is the block length, $\Er(R)$ is a constant depending on $R$.
	This paradigm is considered the analog
	of the large deviations principle (LDP).
	See \cite{Fano61,Gallager65,SGB67,Gallager68,Gallager73,Blahut74,
		BF02,FLM11,DZF16}.
	Meanwhile, a series of works fix the error probability $\P$
	and looked at how the code rate $R$ scales
	\cite{Wolfowitz57,Weiss60,Dobrushin61,Strassen62,BKB04,Hayashi09,PPV10}.
	They showed that the code rate $R$ scales like $I-Q^{-1}(\P)√{V/N}$
	for $I$ the capacity, $Q^{-1}$ the inverse of the standard $Q$-function,
	and $V$ an intrinsic parameter of the channel.
	The parameter $V$ is called
	the \emph{dispersion} or \emph{varentropy} by different authors.
	It is the ``variance'' of the channel
	while $I$ is the ``mean'' of the channel.
	This turns out to be more than an analog---%
	the random variable $\log(W(Y｜X)/W\out(Y))$ called
	\emph{information density} or \emph{information spectrum}
	has mean $I$ and variance $V$.
	This paradigm is considered the analog of the central limit theorem (CLT).
	Later, Altuğ--Wagner, Polyanskiy--Verdú, and followers
	considered the joint behavior when both $\P$ and $R$ vary
	\cite{AW10,AW14,PV10,Arikan15p,HT15}.
	They showed that the quantity $N(I-R)^2/\abs{\log\P}$ converges to $2V$,
	twice the very dispersion appearing in the CLT paradigm.
	This paradigm is considered the analog
	of the moderate deviations principle (MDP).
	
	\begin{table}
		\def\Pd{Paradigm}
		\pgfplotstableread[col sep=ampersand]{
			\Pd &Codes' behavior            &Random      &Polar coding ref.   
			LLN &$R→I$ and $\P→0$           &©Shannon48© &©Arikan09©          
			LDP &$\P≈\exp(N)$               &©BF02©      &©AT09,KSU10,MT14©   
			CLT &$R≈I-1/√N$                 &©PPV10©     &©MHU16,FHMV17,GRY19©
			MDP &$N≈\abs{\log(\P)}/(I-R)^2$ &©AW14©      &©GX13,MHU16,BGS18©  
		}\nlab
		\def\arraystretch{1.44}
		$$\pgfplotstabletypeset[
			every head row/.style={before row=\toprule,after row=\midrule},
			every last row/.style={after row=\bottomrule},string type
		]\nlab$$
		\caption{
			The analogy among probability theory,
			random coding theory, and polar coding theory.
		} \label{tab:analog}
	\end{table}
	
	On a parallel track,
	the engineering aspects of the communication theory thrive.
	Codes with excellent practicality are proposed.
	To name a few,
	Reed--Muller (1964),
	trellis modulation (1970s),
	turbo (1990s),
	low-density parity-check (1963, 1996),
	Repeat-accumulate (1998),
	Fountain (1998),
	and polar (2009).
	
	Among the long list of inventions, only trellis modulation,
	low-density parity-check, and polar achieve the LLN paradigm
	over nontrivial channels---they are \emph{capacity-achieving}.
	Among these three, polar stands out as
	the only code that achieves the CLT paradigm (optimally),
	the only code that achieves the LDP paradigm (optimally), and
	the only code that achieves the MDP paradigm (suboptimally).
	If only polar code achieves the optimal MDP paradigm.
	We brief the history of polar codes below.
	Unless stated otherwise,
	$I$ means the \emph{symmetric capacity} in the next three paragraphs.
	
	Erdal~Arıkan's original works on channel polarization
	\cite{Arikan08,Arikan09} established the foundation of polar codes,
	placing polar codes in the LLN paradigm on day one.
	Arıkan and Telatar \cite{AT09} characterized
	the LDP behavior of polar codes,
	showing that $\P$ scales like $\exp(-√N)$ when an $R<I$ is fixed.
	Later, Korada--Şaşoğlu--Urbanke \cite{KSU10} generalized polar codes
	from Arikan's kernel $\lol$ to any invertible $ℓ$-by-$ℓ$ matrix $G$,
	granted that $ℓ≥2$ and $G$ is not
	column-equivalent to a lower triangular matrix.
	And then they showed that the LDP behavior is $\P≈\exp(-N^{\Ec(G)})$
	where $\Ec(G)$ is a constant depending on the kernel matrix $G$.
	The notation $\Ec(G)$ is meant to resemble
	Gallager's error exponent $\Er(R)$
	but the former is at this level $\exp(-N^•)$
	while the latter is at this level $\exp(-{•}N)$.
	The LDP behavior of polar codes is then refined in \cite{HMTU13}.
	Therein, $\P$ is approximated by $\exp(-ℓ^E)$ where
	$E=\Ec(G)n-√{\Vc(G)n}Q^{-1}(R/I)+o(√n)$ is a more accurate exponent,
	$ℓ$ is the matrix dimension, $n$ is the depth of the code,
	and $\Vc(G)$ is another constant depending on $G$.
	The notation $\Vc(G)$ is meant to resemble the channel dispersion $V$.
	Appearing to be a CLT behavior, this result lies in
	the corner of the LDP paradigm that touches the MDP paradigm.
	Finally, Mori--Tanaka \cite{MT14} generalized everything above
	to channels of prime power input size.
	Over arbitrary input alphabets,
	\cite{STA09i,Sasoglu11} showed the equivalence of \cite{Arikan09,AT09}.
	Over binary but asymmetric channels,
	\cite{SRDR12,HY13} showed the counterpart of \cite{Arikan09,AT09}
	with $I$ being the Shannon capacity.
	No further result on the LDP side,
	e.g.\ over non-binary asymmetric channels, is known.
	The present work fills the gap.
	
	\begin{figure}
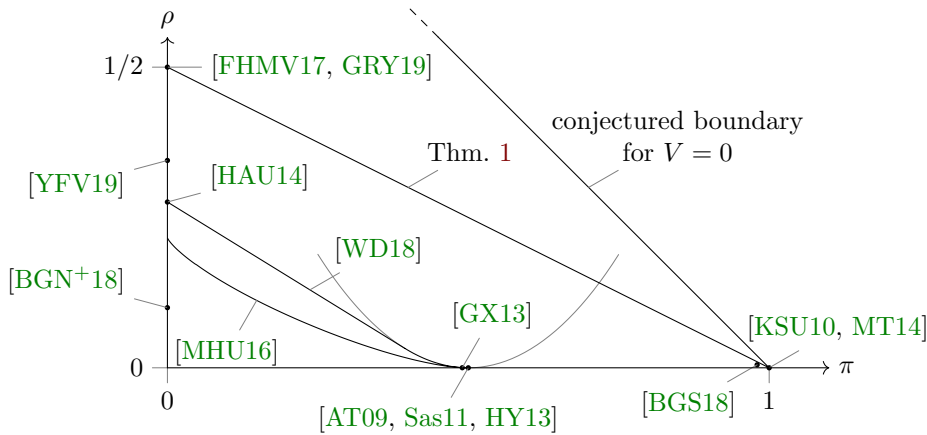
%\workhardtrue
		\tikzset{
			every pin/.style={anchor=180+\tikz@label@angle,anchor/.code=},
			dot/.pic={\fill circle(1pt);}}
		\PMDF{h2}1{\PMP{-#1*log2(#1)-(1-#1)*log2(1-#1)}}
		\PMDF{g2}1{\PMP{1/(4.627-3.627*h2(#1))}} % denominator = mu+1-mu*h2(#1)
		\def\cp#1:#2:{coordinate[pin=#1:{#2}](X)}
		$$\tikz[scale=8,nodes={black,align=center}]{
			\ifworkhard\draw[help lines]
				(-.025,0)coordinate(X)(0,-.025)coordinate(Y)
				(0,.5)--+(X)node[left]{1/2}(0,0)--+(X)node[left]{0}
				(0,0)--+(Y)node[below]{0}(1,0)--+(Y)node[below]{1}
				plot[domain=.25:.3947](\x,{1-h2(\x)})
				plot[domain=.5:.75](\x,{1-h2(\x)})
			;\fi
			\ifworkhard\draw
				(.45,.55)edge[dashed](.4,.6)--(1,0)
				(.7,.3)\cp30:conjectured boundary\\for $V=0$:
				(0,.5)pic{dot}\cp0:©FHMV17,GRY19©:
				(0,.5)--(1,0)(.4,.3)\cp30:Thm.~\ref{thm:hypotenuse}:
				(0,1/2.9)pic{dot}\cp195:©YFV19©:
				(0,1/3.627)pic{dot}\cp15:©HAU14©:
				(0,1/3.627)--node[pos=.6](X){}(.3947,.03223)
				(X)\cp30:©LargeDeviations18©:
				plot[domain=.3947:.5](\x,{1-h2(\x)})
				plot[domain=.001:.5,samples=99]({g2(\x)*\x},{(1-g2(\x))/3.627})
					({g2(.25)*.25},{(1-g2(.25))/3.627})\cp240:©MHU16©:
				(1,0)pic{dot}\cp30:©KSU10,MT14©:
				(.49,.0001)pic{dot}\cp60:©GX13©:
				(0,.1)pic{dot}\cp165:©BGNRS18©:
				(0.98,.005)pic{dot}\cp210:©BGS18©:
			;\fi
			\draw(.5,0)pic{dot}\cp240:©AT09,Sasoglu11,HY13©:;
			\draw[->](0,0)--(1.1,0)node[right]{$π$};
			\draw[->](0,0)--(0,.55)node[above]{$ρ$};
		}$$
		\caption{
			Recent works on polar coding arranged on a $ρ$-$π$ plot.
			Note that results utilizing different kernels
			over various channels are mixed.
			The higher $ρ,π$, the better performance.
			The curve part of \cite{LargeDeviations18} is $ρ=1-h_2(π)$.
		} \label{fig:triangle}
	\end{figure}
	
	The CLT behavior of polar codes turns out to be difficult to characterize.
	It was Korada--Montanari--Telatar--Urbanke \cite{KMTU10} who
	came up with the idea that approximating an \emph{eigenfunction}
	tightly bounds the \emph{eigenvalue} $ℓ^{-ρ}$.
	Here $ρ>0$ is a number such that
	$R$ scales like $I-N^{-ρ}$ with a fixed $\P$.
	They had $0.2669≤ρ≤0.2841$ over binary erasure channels (BECs).
	The upper bound was brought down to $3.553ρ≥1$ \cite{GHU12}.
	Hassini--Alishahi--Urbanke \cite{HAU14} lifted the lower bound
	to $3.627ρ≤1$ over BECs and proved a lower bound $6ρ≤1$
	over binary-input discrete-output memoryless channels (BDMCs).
	The latter is suboptimal so \cite{GB14,MHU16} improved the bound
	to $5.702ρ≤1$ and to $4.714ρ≤1$.
	Additive white Gaussian noise channles (AWGNCs) have continuous
	output alphabet, but \cite{FT17} show that they have $4.714ρ≤1$ too.
	Over BECs particularly,
	\cite{FV14,YFV19} examined a series of larger kernels;
	the current record is a $64$-by-$64$ kernel believed to have $2.9ρ≤1$.
	Near the end of the road to $2ρ<1$,
	\cite{PU16} showed that by allowing $q→∞$,
	Reed--Solomon kernels achieve $2ρ<1$ over $q$-ary channels.
	This does not really prove that polar codes achieve $2ρ<1$
	over any specific channel, but gave hopes.
	Fazeli--Hassani--Mondelli--Vardy \cite{FHMV17,FHMV18}, eventually, showed
	that large random kernels achieve $2ρ<1$ over BECs, breaking the barrier.
	Guruswami--Riazanov--Ye \cite{GRY19} extended their result
	to all BDMCs utilizing the dynamic kernel technique.
	Over the remaining channels, the present work fills the gap.
	
	Between LDP and CLT is polar codes' MDP behavior.
	Guruswami--Xia \cite{GX13} showed that there exists $ρ>0$
	such that $\P$ scales like $\exp(-N^{0.49})$
	while $R$ scales like $I-N^{-ρ}$ over BDMCs.
	This raised a question about what are the possible pairs $(π,ρ)$
	such that $(\P,R)$ scales like $(\exp(-N^π),I-N^{-ρ})$.
	Mondelli--Hassani--Urbanke \cite{MHU16} answered this,
	partially, in the same paper they bounded $ρ$.
	They showed that under a certain curve connecting
	$(0,1/5.714)$ and $(1/2,0)$ all $(π,ρ)$ are achievable over BDMCs.
	For BECs the upper left corner is $(0,1/4.627)$.
	A straightforward generalization to AWGNCs was also given in \cite{FT17}.
	We in \cite{LargeDeviations18} improved their result, suggesting that
	via a combinatorial trick the upper left corner of the curve is $(0,ρ)$
	for any $ρ$ that is valid in the CLT regime.
	The same trick also implicated that over BECs all $(π,ρ)$
	such that $π+2ρ<1$ are achievable, which is mainly owing to
	\cite{FHMV17}'s result that $2ρ<1$ over BECs is achievable.
	Meanwhile, \cite{BGNRS18} made the first step to investigate
	the general kernel matrices over general prime-ary channels.
	They showed that it is possible to achieve $ρ>0$ with $\P≈N^{-Ω(1)}$.
	This is, strictly speaking, ``only'' a CLT behavior as
	the desired error probability in the MDP world is $\exp(-N^π)$.
	Later, Błasiok--Guruswami--Sudan \cite{BGS18} were able to show that
	for all $π<\Ec(G)$ there exists $ρ>0$ such that $(π,ρ)$ is achievable.
	This makes it a direct generalization of \cite{GX13}
	to all polarizing kernel matrices $G$ over all prime-ary channels.
	Over the remaining channels, the present work fills the gap.
	
	The following works, though not counting as predecessors of ours,
	have impact on us through their insights
	on the essence of the channel polarization:
	\cite{Korada09,HKU09,KU10,Arikan10,CK10,Mori10,SP11,Sasoglu12c,Sasoglu12d,
		TV13,Hassani13,PB13,TV15,Mondelli16,Nasser16,Nasser17}.
	
	$I$ resumes to be the Shannon capacity.
	Readers are now prepared to be presented the main theorem.
	
	\begin{thm}[the main theorem---%
		polar codes' simplicity, random codes' durability]
			\label{thm:hypotenuse}
		Let $W$ be any discrete memoryless channel.
		Fix a prime $ς≥2$.
		Fix constants $π,ρ>0$ such that $π+2ρ<1$.
		There exists a sequence of block codes with
		encoding and decoding algorithms such that:
		(cs)	the codes accept uniform $ς$-ary messages.
		(cn)	the block length $N$ approaches infinity;
		(cp)	the block error probability falls below $\exp(-N^π)$;
		(cr)	the code rate exceeds $I-N^{-ρ}$; and
		(cc)	the encoding and decoding complexity
				is $O(N\log N)$ per code block.
	\end{thm}
	
	The proof of the main theorem spans over
	\Cref{sec:channel,sec:transform,sec:parameter,%
		sec:process,sec:globalMDP,pf:LDP,pf:CLT},
	lemmas continuing in \Cref{pf:extoll,pf:triforce,pf:highrule}.
	The entry points are Sections  \ref{sec:alphabet} for (cs),
	\ref{sec:length} for (cn), \ref{sec:complexity} for (cc),
		\ref{sec:error} for (cp), and \ref{pf:trichotomy} for (cr).
	The main theorem is optimal in the following manner.
	
	\begin{pro}[optimality] \label{pro:optimality}
		Fix $π,ρ>0$ such that $π+2ρ>1$.
		Assume $V>0$.
		Conditions (cn), (cp), and (cr) cannot hold simultaneously.
	\end{pro}
	\begin{proof}
		If so, $N(I-R)^2/\abs{\log\P}≤-NN^{-2ρ}/N^π=N^{1-2ρ-π}→0$ as $N→∞$.
		This contradicts $\liminf_{N→∞}N(I-R)^2/\abs{\log\P}≥2V>0$
		\cite[Theorem~2.2]{AW10} \cite[Theorem~6]{PV10} \cite[Theorem~2]{AW14}.
		Remark:
		For $V=0$ channels, the \emph{correct} threshold seems to be $π+ρ=1$
		\cite[Inequality~(3.354)]{Polyanskiy2010} \cite[Remark~1]{AW14}.
	\end{proof}
	
	For the rest of the section,
	we outline the ideas to prove \Cref{thm:hypotenuse}.
	The proof is a straightforward remix of polar coding techniques and
	random coding techniques if it were not for a few hurdles.
	
	Hurdle of input alphabet size:
	The majority of the polar coding theory assumes that
	the input alphabet of the underlying channel is
	binary, of prime size, or, less likely, of prime power size.
	But the main theorem aims for arbitrary finite alphabets.
	Finite alphabets do possess polarization behavior but the speed
	of polarization has room for improvement \cite[Theorem 3.5]{Sasoglu11}.
	We will overcome this by adding ``dummy symbols''
	into the input alphabet to make it a prime power.
	
	Hurdle of asymmetric channel:
	Although asymmetric channels do polarize, the input distributions
	do not automatically become the uniform distribution.
	Pre-composing a source coding machinery helps generate
	the desired distribution and has been proposed before
	\cite[Section~III.D]{STA09i} \cite[Section~IV]{Arikan10}.
	On the other hand Honda--Yamanoto \cite{HY13} showed that \emph{one}
	polar code can do both source coding and noisy channel coding at once.
	We borrow their idea.
	
	Hurdle of kernel selection:
	Judging and identifying the best-behaved kernel gets harder as
	we need finer descriptions of the performance of the code.
	The good result for the BEC case depends heavily on the erasure nature
	of the channels (that they are \emph{ordered} by their capacities).
	Other general results are not strong enough to meet our goal.
	To overcome, we borrow a technique called
	\emph{dynamic kernels} from \cite{YB15}.
	The idea is to prepare more than one polarizing kernel
	and apply the proper one on a channel-by-channel basis.
	This makes a paradigm shift from \emph{one kernel fits all channels}
	to \emph{every channel deserves a tailor-made kernel}.
	We will, once per channel, apply the random coding theory
	to show the existence of a proper kernel.
	
	Hurdle of output alphabet size:
	Even with the great freedom to choose one kernel for each and every channel,
	there lies the difficulty that some performance bounds are proven
	with one fixed channel in mind to favor the big-$O$ notations.
	Those bounds are prone to depend on the size of the output alphabet,
	which grows to infinity as the channel transformations take place.
	Meanwhile, some \emph{universal bounds} are proven
	that depend only on the size of the input alphabet,
	which is invariant under channel transformations.
	We will borrow a bound derived in \cite{CS07,DCS14}.

\subsection{Organization}

	\Cref{sec:channel} reviews channels and entropy notations;
	\Cref{sec:alphabet} explains how to overcome
		the hurdle of arbitrary input alphabet size.
	\Cref{sec:transform} reviews the channel transformations;
	\Cref{sec:decoder} designs the decoder;
	\Cref{sec:complexity} analyzes its complexity;
	\Cref{sec:encoder} designs the encoder,
		overcoming the hurdle of asymmetric channel.
	\Cref{sec:parameter} reviews the channel parameters
		such as the Bhattacharyya parameter;
	\Cref{sec:error} shows how to control the block error probability.
	\Cref{sec:process} reviews the channel processes;
	\Cref{sketch:hypotenuse} argues that the global MDP behaviors
		of $H(𝘞_n)$ and $H(𝘝_n)$ imply the main theorem.
		The main theorem is thus reduced to
		the behavior of certain channel processes.
	\Cref{sec:globalMDP} proves that the global MDP behavior we want
		holds granted that the local LDP and CLT behaviors hold,
		effectively boiling the main theorem down to the local behaviors.
	\Cref{sec:chimera} introduces the random kernel trick and
	\Cref{pf:trichotomy} introduces the dynamic kernel trick
		to overcome the hurdle of kernel selection.
	\Cref{pf:LDP} confirms the local LDP behavior.
		The proof distills properties of
		the weight distribution of random codes.
	\Cref{pf:FTPCZ,pf:FTPCS} proves
		the two fundamental theorems of polar coding.
	\Cref{pf:CLT} confirms the local CLT behavior.
		Contributions from Gallager and Hayashi are utilized.
	\Cref{sec:ChangSahai} invokes Chang--Sahai's universal bound,
		overcoming the hurdle of output alphabet size.

\subsection{Three families of randomnesses}

	The randomnesses from the sender's message, the channel,
	and the randomized rounding constitute the first family.
	Typeset in Roman font are random variables ($U,X,Y,\dotsc$),
	probability measures ($P,Q,W$), entropies ($H,I$),
	and other parameters ($\P,Z,T,S\dotsc$) in this family.
	The randomness from the channel process,
	one main technique in the polar coding literature, is the second family.
	Typeset in sans serif font are
	stochastic processes ($𝘒_n,𝘞_n,𝘏_n,𝘡_n,\dotsc$),
	probability measure ($𝘗$), and expectation ($𝘌$) in this family.
	The randomness from random kernel ensembles,
	the main technique in the random coding literature, is the third family.
	Typeset in blackboard bold font are random variables ($𝔾,𝕏,𝕂$),
	probability measure ($ℙ$, with exceptions), expectation ($𝔼$),
	and Kullback--Leibler divergence ($𝔻$, with exceptions) in this family.

\section{Channel and Entropy Preliminaries} \label{sec:channel}

	A \emph{discrete memoryless channel} is a Markov chain $W:𝒳→𝒴$.
	Here $𝒳$ is a finite set of input alphabet;
	$𝒴$ is a finite set of output alphabet;
	and $W$ is an array of transition probabilities
	$W(y｜x)∈[0,1]$ for all $x∈𝒳$ and $y∈𝒴$.
	The numbers satisfy $∑_{y∈𝒴}W(y｜x)=1$ for all $x∈𝒳$, which represents
	the fact that each $x$ must be transitioned to some unique $y$.
	When $𝒳$ and $𝒴$ are clear from the context, we call $W$ a \emph{channel}.
	Although the input distribution is not part of the channel data,
	we write $W\inp(x)$ to denote the input distribution.
	When $W\inp(x)$ is understood from the context,
	we write $W(x,y)$ to denote the joint distribution $W(y｜x)W\inp(x)$,
	write $W\out(y)$ to denote the output distribution, and
	write $W(x｜y)$ to denote the a posteriori probability $W(x,y)/W\out(y)$.
	(Thus the interpretation of $W(•｜•)$
		depends on the arguments and the context.)
	A tuple of inputs $(x_i,x_{i+1}…x_j)$ is abbreviated as $x_i^j$.
	Same for $y_i^j$ for tuple of outputs,
	and for $u_i^j$ for general variables.
	We assume memoryless channels, and write $W^ℓ(y_1^ℓ｜x_1^ℓ)$ to denote
	the product measure $∏_{i=1}^ℓW(y_i｜x_i)$ for consecutive usages.
	We write $W^ℓ\inp(x_1^ℓ)$, $W^ℓ\out(y_1^ℓ)$,
	$W^ℓ(x_1^ℓ,y_1^ℓ)$, and $W^ℓ(x_1^ℓ｜y_1^ℓ)$ to denote
	the input, output, joint, and a posteriori probabilities.
	
	Let $X,Y$ be two r.v.s (random variables).
	Let $H(X)$, $H(X｜Y)$, and $I(X；Y)$ be the standard
	entropy, conditional entropy, and mutual information.
	The base of the logarithm will be assigned later.
	When $X$ is the input fed into some channel $W:𝒳→𝒴$
	and $Y$ is the corresponding output,
	we say $H(W)$ and $I(W)$ to mean $H(X｜Y)$ and $I(X；Y)$.
	When the distribution of $X$ (the input distribution)
	is chosen to maximize $I(W)$,
	it is called the \emph{capacity-achieving input distribution} and
	$I(W)$ is called the \emph{(Shannon) capacity} of the channel $W:𝒳→𝒴$.
	Unless stated otherwise, the input distributions will be capacity-achieving.

\subsection{Reduce input size to prime power} \label{sec:alphabet}

	Immediately after we declared what channels are concerned
	(those with finite input and output alphabets),
	we show that it suffices to consider input alphabets of prime power size.
	Let $W:𝒳→𝒴$ be a channel.
	Let the input alphabet $𝒳$ be of size $s$.
	Let $q$ be any prime power greater than or equal to $s$.
	Degrade the channel $W$ as follows:
	Let symbols in $𝒳$ be $ξ_1,ξ_2…ξ_s$.
	Let $ξ_{s+1},ξ_{s+2}…ξ_q$ be $q-s$ extra symbols.
	Let $𝒳^♯$ be $𝒳∪\{ξ_{s+1},ξ_{s+2}…ξ_q\}$;
	this is the extended alphabet.
	Define a dummy channel $♮:𝒳^♯→𝒳$ by
	letting $♮(ξ_{\min(i,s)}｜ξ_i)$ be $1$ for all $i=1,2…,q$.
	That is, all extra symbols collapse to $ξ_s$ while the old symbols remain.
	The composition of the two channels
	\[*W\circ ♮:𝒳^♯\stackrel{♮}{⟶}𝒳\stackrel W{⟶}𝒴\pagebreak\]*
	%%% arxiv has different page-breaking point
	forms a degraded channel with prime power input size.
	By the data processing inequality, the Shannon capacity of
	the degraded channel $W\circ ♮$ is no greater than $W$'s Shannon capacity.
	Meanwhile, it is clear that the degraded channel $W\circ ♮$ achieves
	$W$'s capacity by the same input distribution, ignoring extra symbols.
	In other words, $I(W\circ ♮)=I(W)$.
	This constitutes the input size reduction.
	
	Hereafter, we assume the size of the input alphabet $𝒳$ is $q$,
	where $q$ is a prime power.
	If the sender wants to send uniform binary messages,
	let $q$ be a power of $2$.
	If the sender wants to send uniform quaternary messages,
	let $q$ be a power of $3$.
	In case the sender wishes to send uniform quaternary messages
	but does not want to split an information bit over two channel symbols,
	let $q$ be a power of $4$.
	Bonus: should the sender want to send uniform senary messages,
	choose $q_2$ a power of $2$ and $q_3$ a power of $3$
	such that $q_2q_3$ is a power of $6$;
	then alternate between $q=q_2$ and $q=q_3$.
	That is, the sender breaks every senary bit into
	a binary component and a ternary component,
	sends the binary component through the $q=q_2$ code block, and
	send the ternary component through the $q=q_3$ code block.
	For other message alphabets, apply the fundamental theorem of arithmetic.
	
	Fix a $q$.
	Let $𝔽_q$ be the finite field of order $q$
	(with the addition and multiplication structure).
	Identifying $𝒳$ with $𝔽_q$, we will use them interchangeably.
	We say that \emph{$W$ is a $q$-ary channel} when
	the variables $𝒳$ and $𝒴$ are remotely relevant.
	It is worth keeping in mind that for inequalities in this work,
	$q=2$ is the most difficult case and $q≥2$ will be used silently.
	
	We clarified (cs), there are (cn), (cc), (cp), and (cr) to go.

\subsection{On the message alphabet and the block length} \label{sec:length}

	The fact that we have some freedom to choose $q$
	blurs the meaning of the block length $N$ since, say,
	a $q^2$-bit bears twice as much message as a $q$-bit does.
	Notwithstanding, we would like to remind readers
	that multiplication and division of $N$ by any constant
	do not alter the semantics of the main theorem.
	This is because $O(N\log N)$ can absorb any constant;
	$\exp(-N^π)$ and $N^{-ρ}$ too can by fluctuating $π$ and $ρ$ a bit.
	
	A more series aftereffect is caused by mixing code blocks with distinct $q$.
	When the sender attempts to send uniform $30$-ary messages,
	they choose $q_2,q_3,q_5≥s$ and switch among the three block codes.
	The $q=q_2$ blocks have their own block length $N_2$ just like
	the other blocks have $N_3$ or $N_5$ as block lengths.
	The de facto block length $N$, the minimal number of
	the channel usages before the receiver can decode everything sent so far,
	is thus three times the least common multiple of $N_2$, $N_3$, and $N_5$.
	We claim without a proof
	(but it will be clear once we prove the rest of the main theorem)
	that it is possible to make $N_2=N_3=N_5$ and consequently $N=3N_2$.
	Again, increasing $N$ by three-fold does not make any difference.
	For numbers with more prime factors, a similar reasoning applies.
	
	We recommend readers not to worry about the message alphabet as
	there exists a powerful solution---to pre-compose another code
	that re-encodes an arbitrary finite message distribution
	(not necessarily uniform) to a uniform prime power-ary input distribution.
	The existence of such code, by duality, is tightly bonded to the existence
	of a error-correction code that carries uniform prime power-ary messages
	over channels of arbitrary arity.
	The latter is exactly what the main theorem concerns.
	
	We clarified (cs) and (cn) in this section;
	there are (cc), (cp), and (cr) to go.
	We continue proving the main theorem in the next section.

\section{Channel Transformation} \label{sec:transform}

	Let $ℓ≥2$ be an integer.
	This will be the dimension of the kernel matrices.
	But for now, let us introduce a flexible framework.
	Fix a $q$-ary channel $W:𝒳→𝒴$.
	Let $U_1,U_2…U_ℓ$ be r.v.s taking values in $𝒳$.
	For $1≤i≤j≤ℓ$, let $U_i^j$ be the joint r.v.\ $U_iU_{i+1}\dotsm U_j$.
	Let $\g W:𝒳^ℓ→𝒳^ℓ$ be a bijective map;
	that is, $H(U_1^ℓ｜\g W(U_1^ℓ))=0$.
	We now feed $X_1^ℓ≔\g W(U_1^ℓ)$ into $ℓ$ i.i.d.\
	(independent and identically distributed) copies of the channel $W$.
	Let $Y_1^ℓ∈𝒴^ℓ$ be the corresponding output.
	The chain rule of conditional entropy reads
	\[H(U_1^ℓ｜Y_1^ℓ)=H(U_ℓ｜U_1^{ℓ-1}Y_1^ℓ)
		+H(U_{ℓ-1}｜U_1^{ℓ-2}Y_1^ℓ)+\dotsb+H(U_1｜Y_1^ℓ). \label{eq:chain}\]
	Interpretation: to estimate $U_1^ℓ$ given $Y_1^ℓ$,
	we first estimate $U_1$ given $Y_1^ℓ$;
	and then use the estimate $ˆU_1$ to further estimate $U_2$;
	afterward, we estimate $U_3$ given $ˆU_1$, $ˆU_2$, and $Y_1^ℓ$;
	and so on.
	To achieve $W$'s capacity, $\g W(U_1^ℓ)$ must follow
	a certain capacity-achieving distribution.
	Since $\g W$ is bijective, this induces a distribution of $U_1^ℓ$.
	(Remark: we imply nothing about whether $U_1,U_2…U_ℓ$ are i.i.d or not.)
	Fix this distribution, then
	\[*I(U_1^ℓ；Y_1^ℓ)=I(U_ℓ；Y_1^ℓ｜U_1^{ℓ-1})
		+I(U_{ℓ-1}；Y_1^ℓ｜U_1^{ℓ-2})+\dotsb+I(U_1；Y_1^ℓ).\]*
	
	These two chain rules motivate the \emph{channel transformation}:
	Let $[ℓ]$ be the set of integers $\{1,2…ℓ\}$.
	For each $i∈[ℓ]$, let $W\w i:𝒳→𝒳^{i-1}×𝒴^ℓ$ be a channel
	where $W\w i(u_1^{i-1}y_1^ℓ｜u_i)$ is the probability that
	$U_1^{i-1}Y_1^ℓ=u_1^{i-1}y_1^ℓ$ conditioned on $U_i=u_i$.
	A more lengthy but exact form reads
	\[*W\w i(u_1^{i-1}y_1^ℓ｜u_i)≔（∑_{u_{i+1}^ℓ}W^ℓ(\g W(u_1^ℓ),y_1^ℓ)）
		\div（∑_{u_1^{i-1}u_{i+1}^ℓ}W^ℓ\inp(\g W(u_1^ℓ))）.\]*
	Its input distribution $W\w i\inp(u_i)$ is determined by that of $U_i$.
	It may sound weird that $W\w i$ will tell the receiver
	the input of $W\w1,W\w2…W\w{i-1}$ for free.
	But in reality, $W\w i$ acts as an interactive device where
	the receiver (not the sender) needs to input what $U_1^{i-1}$ is
	and the device will output something that looks like $U_1^{i-1}Y_1^ℓ$;
	only when the receiver inputs the correct $U_1^{i-1}$
	does the device return the correct $U_1^{i-1}Y_1^ℓ$.
	Under this interpretation, the de facto capability of $W\w i$ is thus
	$I(U_i；Y_1^ℓ｜U_1^{i-1})$ instead of $I(U_i；U_1^{i-1}Y_1^ℓ)$,
	which justifies the chain rule of the mutual information.
	To avoid confusion, we prefer $H(W\w i)$ over $I(W\w i)$ in calculations.
	% Garbage collection: $I(W\w i)$ is never used
	
	What makes the idea of channel transformation powerful
	is that the transformations apply recursively.
	The precise formulation is as below:
	Fix any $i∈[ℓ]$.
	Let $(X\w i)_1,(X\w i)_2…(X\w i)_ℓ∈𝒳$ be
	$ℓ$ i.i.d.\ copies of the capacity-achieving input of $W\w i$;
	let $(Y\w i)_1,(Y\w i)_2…(Y\w i)_ℓ∈𝒳^{i-1}×𝒴^ℓ$
	be the corresponding outputs.
	Let $\g{W\w i}:𝒳^ℓ→𝒳^ℓ$ be a bijection.
	Define a tuple of r.v.s $(U\w i)_1^ℓ≔(\g{W\w i})^{-1}((X\w i)_1^ℓ)$;
	that is to say, $\g{W\w i}((U\w i)_1^ℓ)=(X\w i)_1^ℓ$.
	For each $j∈[ℓ]$, we define a depth-$2$ channel
	$(W\w i)\w j:𝒳→𝒳^{j-1}×(𝒳^{i-1}×𝒴^ℓ)^ℓ$ where
	$(W\w i)\w j((u\w i)_1^{i-1}(y\w i)_1^ℓ｜(u\w i)_j)$ is the probability that
	$(U\w i)_1^{j-1}(Y\w i)_1^ℓ=(u\w i)_1^{j-1}(y\w i)_1^ℓ$
	conditioned on $(U\w i)_j=(u\w i)_j$.
	To sum up, we can define $(W\w i)\w1,(W\w i)\w2…(W\w i)\w{ℓ}$ out of $W\w i$
	for any $i∈[ℓ]$ in the same way we define $W\w 1,W\w2…W\w{ℓ}$ out of $W$.
	For $i,j∈[ℓ]$, each $(W\w i)\w j$ is again a channel,
	so the transformations apply to generate depth-$3$ channels.
	In the setup of the classical polar coding, a fixed bijection $g$ is used
	to define $W\w i$, $(W\w i)\w j$, $((W\w i)\w j)\w k$, et seq.
	To reach the optimal MDP paradigm,
	we allow $\g W$ to depend on the channel $W$.
	That is to say,
	we need $ℓ$ (presumably distinct) bijections $\g{W\w i}:𝒳^ℓ→𝒳^ℓ$
	for every $i∈[ℓ]$ when we want to define $(W\w i)\w j$ out of $W\w i$.
	Similarly, we need yet another $ℓ^2$ bijections $\g{(W\w i)\w j}:𝒳^ℓ→𝒳^ℓ$
	for every $i,j∈[ℓ]$ in defining depth-$3$ channels.
	And the recursion goes on ad infinitum.
	
	Prudent readers are invited to check
	\cite[the paragraph before Section~III]{STA09i,STA09x}
	\cite{YB15,PSL16,EKMFLK17,LargeDeviations18,GRY19}
	for a list of inhomogeneous configurations of kernels.
	See \cite{STA09i,STA09x,MT10c,Sasoglu11} for how nonlinear bijections
	are similar to (or different from) linear bijections.

\subsection{Design of the decoder} \label{sec:decoder}

	\begin{figure}
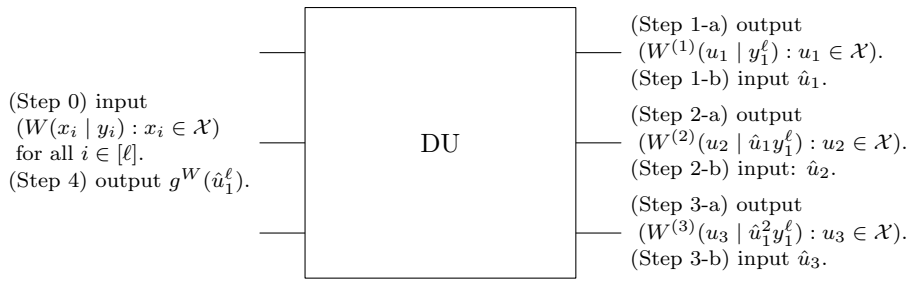
%\workhardtrue
		$$\tikz[scale=1.2]{
			\ifworkhard\draw[nodes={font=\footnotesize,align=left}]
				(0,3)--+(4,0)node[right]{
					(Step 1-a) output\\~~$(W\w1(u_1｜y_1^ℓ):u_1∈𝒳)$.\\
					(Step 1-b) input $ˆu_1$.}
				(0,2)node[left]{
					(Step 0) input\\~~$(W(x_i｜y_i):x_i∈𝒳)$\\
					~~for all $i∈[ℓ]$.\\(Step 4) output $\g W(ˆu_1^ℓ)$.}
				--+(4,0)node[right]{
					(Step 2-a) output\\~~$(W\w2(u_2｜ˆu_1y_1^ℓ):u_2∈𝒳)$.\\
					(Step 2-b) input: $ˆu_2$.}
				(0,1)--+(4,0)node[right]{
					(Step 3-a) output\\~~$(W\w3(u_3｜ˆu_1^2y_1^ℓ):u_3∈𝒳)$.
					\\(Step 3-b) input $ˆu_3$.}
			;\fi
			\draw[fill=white](.5,.5)rectanglenode{DU}+(3,3);
		}$$
		\caption{
			A DU with $ℓ=3$ and its I/Os.
		} \label{fig:DU}
	\end{figure}
	
	To implement channel transformations,
	we define a \emph{DU} (decoding unit) to be an automata as follows:
	It is a box with $ℓ$ pins on the left and $ℓ$ pins on the right.
	Each pin is connected to another DU,
	a CH, an FH, or an IH (to be defined later).
	Each pin may take inputs or output but not at the same moment.
	A DU works as follows:
	Let $W:𝒳→𝒴$ be the channel it is to transform.
	\def\s#1 {(Step~#1) }
	\s0		For all $i∈[ℓ]$, the $i$-th pin on the left takes the input $y_i$.
			The input is passed in the form of
			the a posteriori distribution $(W(x_i｜y_i):x_i∈𝒳)$.
			This is what Arıkan calls $α$-representation
			\cite[Section~II.A]{Arikan15v}.
	\s1-a	It computes the a posteriori distribution of $U_1$ given $y_1^ℓ$;
			that is, $(W\w1(u_1｜y_1^ℓ):u_1∈𝒳)$.
			And then it outputs this tuple of probabilities
			to the first pin on the right.
	\s1-b	At a later moment, it will receive an estimate $ˆu_1$ of $U_1$
			from the first pin on the right.
			Note that $ˆu_1$ is a hard symbol in $𝒳$,
			not a soft tuple of probabilities.
	\s2-a	It computes the a posteriori distribution
			of $U_2$ given $ˆu_1y_1^ℓ$;
			that is to say, it pretends that $U_1$ happens to be $ˆu_1$
			and computes $(W\w2(u_2｜ˆu_1y_1^ℓ):u_2∈𝒳)$ accordingly.
			And then it outputs this tuple of probabilities
			to the second pin on the right.
	\s2-b	At a later moment, it will receive an estimate $ˆu_2$ of $U_2$
			from the second pin on the right.
	\s$i$-a	In general,
			it computes $W\w i(u_i｜ˆu_1^{i-1}y_1^ℓ)$ for all $u_i∈𝒳$ 
			and then output the tuple to the $i$-th pin on the right.
	\s$i$-b	After a while, it will receive $ˆu_i$.
	\s$ℓ+1$	Once it receives $ˆu_ℓ$ from the last pin on the right,
			it computes $ˆy_1^ℓ≔\g W(ˆu_1^ℓ)$, and then
			output $ˆy_i$ to the $i$-th pin on the left for all $i∈[ℓ]$.
	See \Cref{fig:DU,fig:6DU,fig:12DU} for illustrations.
	
	\begin{figure}
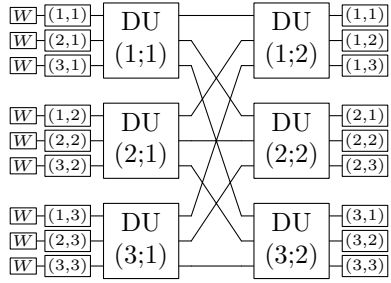
%\workhardtrue
		\tikzset{DU/.pic={\pgftransformscale{1/3}\def\c{coordinate}
			\ifworkhard\draw
				(0,3)\c(1>)--+(4,0)\c(<1)
				(0,2)\c(2>)--+(4,0)\c(<2)
				(0,1)\c(3>)--+(4,0)\c(<3)
			;\fi
			\draw[fill=white](.5,.5)rectangle+(3,3)
				+(1.5,1.5)node[align=center]{DU\\(#1)};
		}}
		$$\tikz{
			\draw[xscale=2,yscale=4/3]
				(0,3)pic(x1){DU=1;1}(1,3)pic(1x){DU=1;2}
				(0,2)pic(x2){DU=2;1}(1,2)pic(2x){DU=2;2}
				(0,1)pic(x3){DU=3;1}(1,1)pic(3x){DU=3;2}
			;
			\ifworkhard\draw foreach\i in{1,2,3}{foreach\j in{1,2,3}{
				(x\j\i>)pic{CH={\i,\j}}
				(x\j<\i)--(\i x\j>)
				(\i x<\j)pic{FIH={\i,\j}}
			}};\fi
		}$$
		\caption{
			$6$ DUs are chained together
			to implement $(W\w1)\w1…(W\w3)\w3$.
			Boxes marked ``$W$'' are channels.
			Boxes next to channels are CHs; the labels are their indexes.
			Boxes at the rightmost column are either FHs or IHs;
			the labels are their indexes.
			Note that DUs in the first column use the same $\g W$.
			DUs in the second column use $\g{W\w1}$,
			$\g{W\w2}$, and $\g{W\w3}$, respectively.
		} \label{fig:6DU}
	\end{figure}
	
	The general rule to arrange the DUs is as follows:
	For a depth-$n$ construction, put DUs in an $ℓ^{n-1}$-by-$n$ array.
	Each DU is indexed by $(k_1,k_2…k_{n-1};m)$
	where $k_1,k_2…k_{n-1}∈[ℓ]$ and $m∈[n]$.
	For all $m∈[n-1]$ and all $k_1,k_2…k_n∈[ℓ]$, connect the
	$k_m$-th pin on the right of the $(k_1…k_{m-1},k_{m+1}…k_n;m)$-th DU to the
	$k_{m+1}$-th pin on the left of the $(k_1…k_m,k_{m+2}…k_n;m+1)$-th DU.
	Here, the $(k_1,k_2…k_{n-1};m)$-th DU is to transform
	the channel $(\cdots(W\w{k_1})\cdots)\w{k_m}$
	into $(\cdots(W\w{k_1})\cdots)\w{k_{m+1}}$.
	The $k_1$-th pin on the left of the $(k_2…k_n;1)$-th DU
	connects to a \emph{CH} (channel helper) indexed by $(k_1,k_2…k_n)$.
	Each CH then connects to the output of a copy of the channel $W$.
	The $k_n$-th pin on the right of the $(k_1…k_{n-1};n)$-th DU
	connects to either an \emph{FH} (frozen bit helper)
	or an \emph{IH} (information bit helper);
	in either case, the connected helper is indexed by $(k_1,k_2…k_n)$.
	Let $ℐ⊂[ℓ]^n$ be the set of indexes $(k_1,k_2…k_n)$
	such that the $k_n$-th pin on the right of
	the $(k_1…k_{n-1};n)$-th DU connects to an IH.
	Then $[ℓ]^n、ℐ$ is the set of indexes
	where the pin connects to an FH.
	
	On the left hand side of the DU array,
	the task of the $(k_1,k_2…k_n)$-th CH is to receive
	the channel output $Y\M{k_1,k_2…k_n}∈𝒴$
	and then forward the a posteriori distribution
	$(W(x\M{k_1,k_2…k_n}｜Y\M{k_1,k_2…k_n}):x\M{k_1,k_2…k_n}∈𝒳)$
	to the DU array.
	On the right hand side,
	FHs correspond to what Arıkan called \emph{frozen bits}---%
	bits that do not carry information and the receiver
	knows their values as part of the communication protocol.
	The task of the $(k_1,k_2…k_n)$-th FH is to receive
	the a posteriori distribution of the $(k_1,k_2…k_n)$-th frozen bit and then
	return the correct symbol $U\M{k_1,k_2…k_n}∈𝒳$ back to the DU array.
	IHs correspond to information bits that carry the sender's messages.
	The task of the $(k_1,k_2…k_n)$-th IH is to receive the a posteriori
	distribution of the $(k_1,k_2…k_n)$-th information bit and then return
	the most probable symbol $ˆU\M{k_1,k_2…k_n}∈𝒳$ back to the DU array.
	When all IHs are activated once, a code block completes.
	The most probable symbols they returned to the DU array
	form the decoded message $ˆU_ℐ$,
	meaning the tuple $(ˆU\M{k_1,k_2…k_n}:(k_1,k_2…k_n)∈ℐ)$.
	
	\begin{figure}
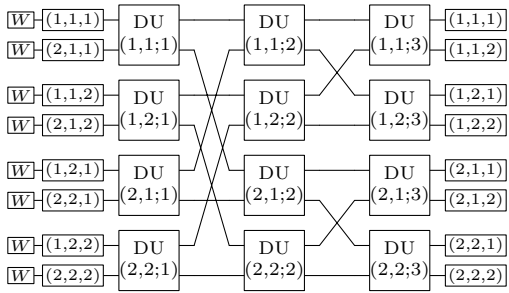
%\workhardtrue
		\tikzset{DU/.pic={\pgftransformscale{2/5}\def\c{coordinate}\scriptsize
			\ifworkhard\draw
				(0,2)\c(1>)--+(3,0)\c(<1)(0,1)\c(2>)--+(3,0)\c(<2)
			;\fi
			\draw[fill=white](.5,.5)rectangle+(2,2)
				+(1,1)node[align=center]{DU\\(#1)};
		}}
		\def\DU(#1,#2;#3)(#4){(#3,-#1*2-#2)pic(#4){DU={#1,#2;#3}}}
		$$\tikz{
			\draw[xscale=5/3,yscale=1]
				\DU(1,1;1)(x11)\DU(1,1;2)(1x1)\DU(1,1;3)(11x)
				\DU(1,2;1)(x12)\DU(1,2;2)(1x2)\DU(1,2;3)(12x)
				\DU(2,1;1)(x21)\DU(2,1;2)(2x1)\DU(2,1;3)(21x)
				\DU(2,2;1)(x22)\DU(2,2;2)(2x2)\DU(2,2;3)(22x)
			;
			\ifworkhard
				\draw foreach\i in{1,2}{foreach\j in{1,2}{foreach\k in{1,2}{
					(x\j\k\i>)pic{CH={\i,\j,\k}}
					(x\j\k<\i)--(\i x\k\j>)(\i x\k<\j)--(\i\j x\k>)
					(\i\j x<\k)pic{FIH={\i,\j,\k}}
				}}}
			\fi
		}$$
		\caption{ 
			$12$ DUs (with $ℓ=2$) are chained together to implement
			$((W\w1)\w1)\w1…((W\w2)\w2)\w2$.
			DUs in the first column use $\g W$;
			DUs in the second column use $\g{W\w1}$ and $\g{W\w2}$;
			DUs in the third column use $\g{(W\w1)\w1}$,
			$\g{(W\w1)\w2}$, $\g{(W\w2)\w1}$, and $\g{(W\w2)\w2}$.
		} \label{fig:12DU}
	\end{figure}
	
	What we just established is the \emph{successive cancellation decoder}
	of polar codes that could be found in most works that implement polar codes.
	For instance, \cite[Section~VIII]{Arikan09}, \cite[Section~3.2]{Korada09},
	and \cite[Section~III]{HY13}, and \cite[Section~Vi.B]{EKMFLK17}.
	See especially \cite[Section~9]{GRY19} for
	an almost identical construction albeit they had $q=2$ in mind.
	We replicate the whole story to demonstrate that each DU may use
	a unique bijection ``$g$'' without changing the overall structure too much.
	Whether or not this construction
	can transmit information reliably is discussed in \Cref{sec:error}.
	There, we will also clarify how to arrange FHs and IHs.
	The complexity can be estimated prior to further specification.

\subsection{Complexity of the decoder} \label{sec:complexity}

	There are various models that measure the complexity of a structure.
	The polar coding community uses a variant of the circuit complexity
	where the arithmetic of real numbers costs $O(1)$
	and passing probabilities between DUs costs $O(1)$.
	The complexity of the DU array is thus
	the number of the DUs multiplied by the complexity of a single DU.
	The number of DUs is $ℓ^{n-1}n$.
	The complexity of a DU depends on
	how a DU computes the a posteriori probabilities
	$W\w j(u_i｜u_1^{i-1}y_1^ℓ)$ out of $W(x_i｜y_i)$.
	The naïve approach is to exhaust all possible inputs $u_1^ℓ∈𝒳^ℓ$
	and compute the a posteriori probabilities using Bayesian formulas.
	This costs $O(ℓ^{10}q^{ℓ+10})$ (here $10$ is an overestimate).
	Hence the overall complexity is $O(ℓ^{n-1}nℓ^{10}q^{ℓ+10})$.
	In our setup, however, $q$ is fixed, $ℓ$ will be chosen upon knowing $π,ρ$,
	and $n$ goes to infinity afterwards.
	So we advertise that the complexity is $O(ℓ^nn)$, or $O(N\log N)$.
	Here $N≔ℓ^n$ is the block length,
	equal to the number of copies of the channel $W$ attached to the DU array.
	The complexities of the CHs, FHs, and IHs can be computed similarly.
	They are all bounded by $O(ℓ^{n+10}q^{10})$.
	Thus the decoder as a whole costs $O(N\log N)$.
	
	We claim that the encoder has the same complexity $O(N\log N)$
	although we have not defined the encoder yet.
	The encoder is essentially a special decoder
	and is the subject of the next subsection.

\subsection{Design of the encoder} \label{sec:encoder}

	The encoder will be an exact copy of the decoder
	except that CHs and IHs will behave differently.
	In greater detail:
	Let there be an $ℓ^{n-1}$-by-$n$ array of DUs indexed and connected
	in the same way described in \Cref{sec:decoder}.
	Each DU executes the exact same task described in \Cref{sec:decoder}.
	The left pins of the DUs in the first column each connect to a CH.
	The right pins of the DUs in the last column each connect to
	the same type of device (an IH or an FH) as its twin-DU in the decoder does.
	Here, as part of the encoder, a CH will output
	the capacity-achieving input distribution
	($W\inp(x)$ for all $x∈𝒳$) into the DU array.
	For each $(k_1,k_2…k_n)∈ℐ$,
	the $(k_1,k_2…k_n)$-th IH will receive a recommended distribution of 
	the $(k_1,k_2…k_n)$-th information bit and then return the message symbol
	$U\M{k_1,k_2…k_n}∈𝒳$ the sender wants to send back to the DU array.
	For each $(k_1,k_2…k_n)∈[ℓ]^n、ℐ$,
	the $(k_1,k_2…k_n)$-th FH will receive a recommended distribution of
	the $(k_1,k_2…k_n)$-th frozen bit and then return a r.v.\
	$U\M{k_1,k_2…k_n}∈𝒳$ that follows that distribution back to the DU array.
	This r.v.\ is simulated by a pseudo random number generator
	shared between the encoder and the decoder.
	The twin-FH in the decoder, regardless what distribution it receives,
	will return the exact same symbol $U\M{k_1,k_2…k_n}$ back to the DU array.
	This step is called \emph{randomized rounding}
	and is found in \cite[Section~3.3]{Korada09}, \cite[Section~III]{KU10},
	\cite[Section~II]{KT10}, and \cite[Section~III.A]{HY13}.
	After all IHs return the sender's messages and all FHs returns
	randomly rounded bits to the DU array, the CHs will each get
	a codeword symbol $X\M{k_1,k_2…k_n}∈𝒳$ from the DU array.
	And then each CH will forward that symbol
	to an i.i.d.\ copy of the channel $W$.
	
	This design is a copy of \cite{HY13}'s encoder explained in our terminology.
	It is clear that the encoding complexity will be $O(N\log N)$, too.
	Alongside the decoder, the encoder creates its own channel transformations.
	Let $W:𝒳→𝒴$ be a $q$-ary channel and $X$ be a capacity-achieving input.
	Define a flattening channel $W_♭:𝒳→\{η\}$ that erases all information.
	Then the encoder is effectively synthesizing
	depth-$1$ channels $W_♭\w i:𝒳→𝒳^{i-1}×\{η\}^ℓ$ for each $i∈[ℓ]$,
	depth-$2$ channels
	$(W_♭\w i)\w j:𝒳→𝒳^{j-1}×(𝒳^{i-1}×\{η\}^ℓ)^ℓ$ for each $j∈[ℓ]$,
	depth-$3$ channels
	$((W_♭\w i)\w j)\w k:𝒳→𝒳^{k-1}×(𝒳^{j-1}×(𝒳^{i-1}×\{η\}^ℓ)^ℓ)^ℓ$
	for each $k∈[ℓ]$, et seq.\ 
	utilizing the same input distributions and series of bijections.
	For instance, $W_♭\w i(u_1^{i-1}y_1^ℓ｜u_i)$ is the probability that
	$U_1^{i-1}=u_1^{i-1}$ conditioned on $U_i=u_i$, or equally
	\[*（∑_{u_{i+1}^ℓ}W^ℓ\inp(\g W(u_1^ℓ))）\div
		（∑_{u_1^{i-1}u_{i+1}^ℓ}W^ℓ\inp(\g W(u_1^ℓ))）.\]*
	Moreover, $H(W_♭)=H(X)$ and $H(W_♭\w i)=H(U_i｜U_1^{i-1})$.
	No ``$Y$'' plays any role here since they are constant.
	The fact that a channel as boring as $W_♭$ is helpful
	to our main theorem will be covered later, in \Cref{sec:error}.
	
	We clarified (cs), (cn), and (cc) up to this section;
	there are (cp) and (cr) to go.

\section{Channel Parameters} \label{sec:parameter}

	Let $W:𝒳→𝒴$ be a $q$-ary channel.
	Let $X$ be a capacity-achieving input and $Y$ be the corresponding output.
	Besides $H$ and $I$, there are several channel parameters
	that capture the qualities of channels.
	Here is a list of parameters extracted from
	the work \cite{MT14} of Mori and Tanaka.
	
	\par	Both $H(X｜Y)$ and $H(W)$ are the base-$q$ conditional entropy,
		the base chosen such that $0≤H(X｜Y)≤H(X)≤1$.
		Both $I(X｜Y)$ and $I(W)$ are the base-$q$ mutual information,
		and hence $0≤I(X｜Y)≤H(X)≤1$.
	\par	$\P(X｜Y)$ is the error probability of
		the maximum a posteriori (MAP) decoder.
		The MAP decoder looks at an output $y∈𝒴$ and
		chooses a symbol $ˆx∈𝒳$ that maximizes $W(ˆx｜y)$.
		When the output is $Y=y$, the probability that the MAP decoder
		does not choose $X$ as $ˆx$ is $1-\max_{x∈𝒳}W(x｜y)$.
		Therefore, $\P(X｜Y)=∑_{y∈𝒴}W\out(y)(1-\max_{x∈𝒳}W(x｜y))$.
		In a channel-centric narrative, we also write $\P(W)$ for $\P(X｜Y)$.
	\par	$Z(X｜Y)$ is the rescaled sum of Bhattacharyya coefficients
		of the transition distribution $W(y｜x)$ for the uniform input.
		For non-uniform inputs, a modification is made to generalize
		the definition and the properties that used to hold.
		Intuitively speaking, a MAP decoder seeing $y$ is ``confident''
		if $W(x｜y)$ is small for all but one $x$, or equivalently,
		if the product $W(x,y)W(x',y)$ is small for all distinct $x,x'∈𝒳$.
		The \emph{Bhattacharyya parameter} measures the ``confidence'' by
		\[*Z(X｜Y)≔÷1{q-1}∑_{\substack{x,x'∈𝔽_q\\x≠x'}}
			∑_{y∈𝒴}√{W(x,y)W(x',y)}.\]*
		In addition, define
		\[*\Z(X｜Y)≔\max_{0≠d∈𝔽_q}∑_{x∈𝔽_q}∑_{y∈𝒴}√{W(x,y)W(x+d,y)}.\]*
		We also write $Z(W)$, and $\Z(W)$ for these quantities.
		Remarks:
		The rescaling is such that $0≤Z≤\Z≤(q-1)Z≤q-1$.
		Our definition of $\Z$ is different from the $Z_{\max}$ in
		\cite{MT14}, but rather a mixture of $Z_{\max}$ and $Z_d$ therein.
		That said, the definitions of other parameters---%
		$H$, $I$, $\P$, $Z$, $T$, $S$, and $\S$---match \cite{MT14}'s.
		Cf.\ \cite[Section~3.C]{Sasoglu11}.
	\par	$T(X｜Y)$ is the weighted average of the total variation distances
		from the a posteriori distributions $(W(x｜y):x∈𝒳)$
		to the uniform noise $(1/q,1/q…1/q)$.
		More formally, it is defined to be
		$∑_{y∈𝒴}W\out(y)∑_{x∈𝒳}\abs{W(x｜y)-1/q}$.
		We also write $T(W)$ for this quantity.
	\par	$S(X｜Y)$ is the weighted average of the $L^1$-norms of
		the Fourier coefficients of the a posteriori distributions.
		The formal definition is as follows.
		Let $\tr:𝔽_q→𝔽_p$ be the field trace,
		where $𝔽_q=𝒳$ and $𝔽_p$ is the prime subfield.
		Let $χ:𝔽_q→ℂ$ be an additive character
		defined as $χ(x)≔\exp(2πi\tr(x)/p)$,
		where $2πi$ is temporarily the period of $\exp$.
		Define the Fourier coefficient
		\[*M(w｜y)≔∑_{z∈𝔽_q}W(z｜y)χ(wz).\]*
		Define the $S$-parameters
		\begin{gather*}
			S(X｜Y)≔÷1{q-1}∑_{0≠w∈𝔽_q}∑_{y∈𝒴}W\out(y)·\abs[\Big]{M(w｜y)},
				\rlap{\quad and} \\
			\S(X｜Y)≔\max_{0≠w∈𝔽_q}∑_{y∈𝒴}W\out(y)·\abs[\Big]{M(w｜y)}.
		\end{gather*}
		We also write $S(W)$ and $\S(W)$ for these quantities.
		Remarks:
		The rescaling is such that $0≤S≤\S≤(q-1)S≤q-1$.
		An interpretation is as follows:
		Fix a $y$.
		When $W(x｜y)$ is roughly equal to $1/q$ for all $x∈𝔽_q$,
		the Fourier coefficient $M(w｜y)=∑_{z∈𝔽_q}W(z｜y)χ(wz)$
		should be roughly $∑_{z∈𝔽_q}χ(wz)/q=0$.
		The $S$-parameter measures how far those coefficients are from zero.

\subsection{Relations among channel parameters} \label{sec:tolls}

	The following is a series of lemmas we extract from existing works.
	They characterize the relations among $H$, $I$, $\P$, $Z$, $T$, and $S$.
	
	\begin{lem} \label{lem:pz}
		\cite[Lemma~22 with $k=1$]{MT14}
		For any $q$-ary channel $W$,
		\[*÷{q-1}{q^2}（√{1+(q-1)Z(W)}-√{1-Z(W)}）^2≤\P(W)≤÷{q-1}2Z(W).\]*
	\end{lem}
	
	\begin{lem} \label{lem:pt}
		\cite[Lemma~23 with $k=q-1$]{MT14}
		For any $q$-ary channel $W$,
		\[*÷{q-1}q-\P(W)≤÷{T(W)}2≤÷{q-1}q-÷1q（(q-1)q\P(W)-(q-1)(q-2)）.\]*
	\end{lem}
	
	\begin{lem} \label{lem:ps}
		\cite[Lemma~26 with $k=q-1$]{MT14}
		For any $q$-ary channel $W$,
		\[*1-÷q{q-1}\P(W)≤S(W)
			≤(q-1)q（÷{q-1}q-\P(W)）√{1-÷q{q-1}÷{q-2}{q-1}}.\]*
	\end{lem}
	
	\begin{lem} \label{lem:ph}
		\cite[Theorem~1]{FM94}
		For any $q$-ary channel $W$,
		\begin{gather*}
			h_2(\P(W))+\P(W)\log_2(q-1)≥H(W)\log_2q≥2\P(W)\rlap{\quad  and} \\
			H(W)\log_2q≥(q-1)q\log_2÷q{q-1}（\P(W)-÷{q-2}{q-1}）+\log_2(q-1).
		\end{gather*}
		Here, $h_2$ is the binary entropy function; $h_2(1/2)=1$.
		The upper bound is Fano's inequality.
		The first lower bound fits when $H(W)$ and $\P(W)$ are small;
		the second lower bound fits when $H(W)$ and $\P(W)$ are close to $1$.
	\end{lem}

	The above lemmas inspire the following characterization:
	Let $A$ and $B$ be two channel parameters,
	we say $A,B$ are \emph{bi-Hölder at $(a,b)$} if there exists
	$c,d>0$ such that $\abs{A(W)-a}<c\abs{B(W)-b}^d$ and
	$\abs{B(W)-b}<c\abs{A(W)-a}^d$ for all $q$-ary channels $W$.
	The notion of bi-Hölder is an equivalence relation.
	In particular, if $A,B$ are bi-Hölder at $(a,b)$ and $(B,C)$ are
	bi-Hölder at $(b,c)$, then $(A,C)$ are bi-Hölder at $(a,c)$.
	In this case, it makes sense to say $A,B,C$ are bi-Hölder at $(a,b,c)$.
	This notion generalizes to tuples of more parameters.
	Now we can summarize \Cref{lem:pz,lem:pt,lem:ps,lem:ph}
	in a more concise statement.
	
	\begin{lem}[implicit bi-Hölder tolls] \label{lem:imtoll}
		Parameters $H,\P,Z,\Z$ are bi-Hölder at $(0,0,0,0)$.
		Parameters $H,\P,T,S,\S$ are bi-Hölder at $(1,1-1/q,0,0,0)$.
	\end{lem}
	\begin{proof}
		$Z,\Z$ are bi-Hölder at $(0,0)$ since $Z≤\Z≤(q-1)Z$.
		\Cref{lem:pz} implies that $\P,Z$ are bi-Hölder at $(0,0)$.
		\Cref{lem:ph} (with the first lower bound)
		implies that $\P,H$ are bi-Hölder at $(0,0)$.
		Now apply the transitivity to conclude the first statement.
		For the second statement,
		$S,\S$ are bi-Hölder at $(0,0)$ since $S≤\S≤(q-1)S$.
		\Cref{lem:ps} implies that $\P,S$ are bi-Hölder at $(1-1/q,0)$.
		\Cref{lem:pt} implies that $\P,T$ are bi-Hölder at $(1-1/q,0)$.
		\Cref{lem:ph} (with the second lower bound)
		implies that $\P,H$ are bi-Hölder at $(1-1/q,1)$.
		Now apply the transitivity to conclude.
	\end{proof}
	
	See also \cite[Corollary~28]{MT14} for what inspired us.
	They use notation $A\stackrel{\mathrm e}\sim B$ to mean $A,B$
	are bi-Hölder at $(0,0)$ and at $(1,1)$.
	For some very technical details on the way toward the main theorem,
	we need explicit Hölder relations among $H$, $\Z$, and $\S$.
	We claim them here.
	The proof is nothing but looking closer into \Cref{lem:pz,lem:ps,lem:ph}.
	A written-out proof is in \Cref{pf:extoll}.
	
	\begin{lem}[explicit Hölder tolls] \label{lem:extoll}
		$\log$ is natural.
		For all $q$-ary channels $W$, the following hold:
		\begin{gather}
			\Z(W)≤q√{H(W)\log_4q}, \label{eq:Z<H}\\
			H(W)≤√{e(q-1)\Z(W)/2}, \label{eq:H<Z}\\
			\S(W)≤(q-1)q√{(1-H(W))\log(q)/2},\rlap{\quad and} \label{eq:S<1-H}\\
			1-H(W)≤(q-1)\S(W)/\log q\vphantom{√)}. \label{eq:1-H<S}
		\end{gather}
	\end{lem}

\subsection{Control of the block error probability} \label{sec:error}

	Let $W$ be the channel we want to communicate over;
	and let $X$ be any input.
	In the classical theory of polar coding,
	the second last step of the construction of the block code
	is to determine a subset $ℐ⊂[ℓ]^n$ of indexes that points to
	the depth-$n$ channels that transmit information bits.
	When decoding this code, a block error happens if
	the successive cancellation decoder fails to decode any information bit.
	Let $E\M{k_1,k_2…k_n}$ be the event that the first error occurs when
	the decoder is solving for the input to $(\cdots(W\w{k_1})\cdots)\w{k_n}$,
	i.e., when $ˆU\M{k_1,k_2…k_n}≠U\M{k_1,k_2…k_n}$
	and the equality holds for lexicographically earlier indexes.
	Then the event's probability measure $P(E\M{k_1,k_2…k_n})$ is no more than
	the bit error probability $\P\((\cdots(W\w{k_1})\cdots)\w{k_n}\)$.
	By the union bound,
	the block error probability of the decoder is bounded from above by a sum
	\[*P\{ˆU_ℐ≠U_ℐ\}
		≤∑_{(k_1,k_2…k_n)∈ℐ}\P（(\cdots(W\w{k_1})\cdots)\w{k_n}）.\]*
	With this observation, we may define $ℐ$ to be the set of indexes
	$(k_1,k_2…k_n)∈[ℓ]^n$ such that $H\((\cdots(W\w{k_1})\cdots)\w{k_n}\)<θ_n$
	for some clever choice of the threshold $θ_n>0$.
	This immediately implies $\P\((\cdots(W\w{k_1})\cdots)\w{k_n}\)<cθ_n^d$
	for some $c,d>0$ by \Cref{lem:imtoll}.
	Let $θ_n$ be $\exp(-ℓ^{πn}n)$.
	The sum of $\P$ is less than $ℓ^ncθ_n^d<\exp(-ℓ^{πn})$ for
	sufficiently large $n$, which is the block error probability we claimed.
	Remark:
	Arıkan used a different criterion $Z<θ_n$.
	It still implies $\P<cθ_n^d$ and that the sum of $\P$
	is less than $ℓ^ncθ_n^d<\exp(ℓ^{πn})$ for large $n$.
	The benefit of controlling $\P$ using other parameters
	is that some parameters are easier to control
	(because \Cref{thm:FTPCZ,thm:FTPCS} exist).
	
	For the main theorem where the channel $W$ is asymmetric,
	we want to control both the decoder block error and the encoder block error.
	Here, the encoder block error is not the encoder's failure
	to encode a message, but rather its failure to generate
	the capacity-achieving input distribution of $W$.
	To penalize, imagine that we employ an oracle that \emph{claims}
	an encoder block error whenever the generated codeword
	should have been another word to fit the ideal distribution.
	That way, the actual block error probability will not exceed
	the sum of the encoder and decoder block error probabilities.
	More rigorously, let $P$ be the probability measure assuming
	the ideal distribution of $U_ℐ$ and $Q$ be the probability measure
	assuming the actual $U_ℐ$ generated by the encoder.
	Then the overall block error probability can be bounded by
	\[*Q\{ˆU_ℐ≠U_ℐ\}≤P\{ˆU_ℐ≠U_ℐ\}+\norm{P-Q}.\]*
	$P\{ˆU_ℐ≠U_ℐ\}$ as the decoder block error probability is bounded before.
	The encoder block error probability is represented by $\norm{P-Q}$,
	the total variation distance from $P$ to $Q$.
	There is a telescoping argument similar to how we control the decoder error%
	---classifying events by the first input bit
	where the oracle disagrees with the encoder
	\cite[Lemma~3.5]{Korada09} \cite[Lemma~4]{KU10}
	\cite[Lemma~2]{KT10} \cite[Lemma~1]{HY13}.
	It yields that the encoder block error probability
	is bounded from above by the sum
	\[*\norm{P-Q}≤∑_{(k_1,k_2…k_n)∈ℐ}T（(\cdots(W_♭\w{k_1})\cdots)\w{k_n}）.\]*
	In controlling the encoder bit error probability,
	we strengthen the policy of collecting indexes
	$(k_1,k_2…k_n)∈[ℓ]^n$ for $ℐ$ by asking for
	$H\((\cdots(W_♭\w{k_1})\cdots)\w{k_n}\)>1-θ_n$.
	The latter immediately implies
	$T\((\cdots(W_♭\w{k_1})\cdots)\w{k_n}\)<cθ_n^d$ by \Cref{lem:imtoll}.
	As a consequence, the overall block error probability is controlled by
	$Q\{ˆU_ℐ≠U_ℐ\}≤P\{ˆU_ℐ≠U_ℐ\}+\norm{P-Q}<2ℓ^ncθ_n^d<\exp(-ℓ^{πn})$
	for $n$ large.
	
	The preceding argument is a paraphrase
	of the proof of \cite[Theorem~13]{HY13};
	Inequalities (59) and~(57) therein are the keys.
	So far the block length, the complexity,
	and the error aspects of the main theorem are covered,
	it remains to control the code rate $\abs{ℐ}/ℓ^n$.
	In other words, we are to compute the cardinality of $ℐ$
	given that $ℐ⊂[ℓ]^n$ is the set of indexes such that
	$H\((\cdots(W\w{k_1})\cdots)\w{k_n}\)<θ_n$ and
	$1-H\((\cdots(W_♭\w{k_1})\cdots)\w{k_n}\)<θ_n$, where $θ_n≔\exp(-ℓ^{πn}n)$.
	% Garbage collection: $T$

\subsection{Before and after channel transformations} \label{sec:FTPC}

	Alongside the relations among different parameters applied to
	the same channel, there are also relations between the same parameter
	applied to the original and the transformed channels.
	That $∑_{i=1}^ℓH(W\w i)=ℓH(W)$ is one.
	There are two more that are pivotal in the theory
	of polar coding but require more prerequisites.
	Assume that $\g W:𝒳^ℓ→𝒳^ℓ$ is a linear isomorphism given by
	the multiplication of an invertible matrix $G$ from the right%
	---$\g W(u_1^ℓ)≔u_1^ℓG$.
	The following framework extends to nonlinear bijections
	but we do not need that much.
	(There is also the paradigm that random linear codes perform better
	than random codebooks for that a bad linear code tends to hoard a lot of
	short codewords at once, effectively removing them from the ensemble pool.
	So there is a good reason to stick to the linear case.)
	Let $0_1^{i-1}1_iu_{i+1}^ℓ∈𝔽_q^ℓ$ be a tuple of
	$i-1$ many $0$ followed by a $1$ and $ℓ-i$ arbitrary symbols.
	A \emph{coset code} is a subset of codewords
	of the form $\{0_1^{i-1}1_iu_{i+1}^ℓG:u_{i+1}^ℓ∈𝔽_q^{ℓ-i}\}⊂𝔽_q^ℓ$.
	The coset codes have weight distributions just like every other code does.
	Let $\wt(x_1^ℓ)$ be the hamming weight of $x_1^ℓ$.
	The weight enumerator of the $i$-th coset code is
	defined to be a one-variable polynomial over the integers
	\[*\fz G\w i(z)≔∑_{u_{i+1}^ℓ}z^{\wt(0_1^{i-1}1_iu_{i+1}^ℓG)}∈ℤ[z].\]*
	We can now state the second relation.
	This is considered the main cause of why polar coding ever exists/works.
	
	\begin{thm}[fundamental theorem of polar coding---the $Z$-end, FTPC$Z$]
		\label{thm:FTPCZ} \cite[Proposition~5]{Arikan09}
		\cite[Lemma~10]{KSU10} \cite[Lemma~3.5]{Sasoglu11}
		\cite[Section~4.1]{FHMV17} \cite[Lemma~33]{MT14}
		\[*\Z(W\w i)≤\fz G\w i(\Z(W)).\]*
	\end{thm}
	
	The proof is postponed until \Cref{pf:FTPCZ}.
	The fundamental theorems come as a pair.
	Let $u_1^{i-1}1_i0_{i+1}^ℓ∈𝔽_q^ℓ$ be a tuple of
	$i-1$ arbitrary symbols followed by a $1$ and $ℓ-i$ many $0$.
	Let $\Git$ be the inverse transpose of $G$.
	The weight enumerator of the $i$-th dual coset code is
	defined to be this one-variable polynomial over the integers
	\[*\fs G\w i(s)≔∑_{u_1^{i-1}}s^{\wt(u_1^{i-1}1_i0_{i+1}^ℓ\Git)}∈ℤ[s].\]*
	We can now state the third relation, the dual of the second.
	The proof is postponed until \Cref{pf:FTPCS}.
	
	\begin{thm}[fundamental theorem of polar coding---the $S$-end, FTPC$S$]
		\label{thm:FTPCS}
		\cite[Lemma~5.7]{Korada09} \cite[Theorem~19]{KU10} \cite[Lemma~6]{KT10}
		\cite[Lemma~34]{MT14} \cite[Inequalities (74) and~(75)]{GRY19}
		\[*\S(W\w i)≤\fs G\w i(\S(W)).\]*
	\end{thm}
	
	Remark:
	These two bounds are not tight---the equality does not hold for BECs.
	In detail, Arıkan's original bound reads $\Z(W\w 1)≤2\Z(W)-\Z(W)^2$
	while our bound turns into $\Z(W\w 1)≤2\Z(W)$, the subtraction term missing.
	We are simply not able to prove a version that degenerates to an equality
	over erasure channels, nor does any prior work seem to.
	This causes a serious aftermath that $\Z(𝘞_n)$ (to be defined later)
	is no longer a supermartingale.
	Nonetheless, this bound is strong enough
	to collaborate with the random coding theory.
	See, for example, how we compensate in \Cref{pf:super}.
	
	We clarified (cs), (cn), (cc) and (cp) up to this section;
	there is (cr) to go.

\section{Channel Processes} \label{sec:process}

	Let $𝘒_1,𝘒_2,𝘒_3,\dotsc$ be i.i.d.\ uniform r.v.s on $[ℓ]$,
	where $[ℓ]$ is the set of integers $\{1,2…ℓ\}$.
	Let $W$ be the $q$-ary channel we want to communicate over.
	Let $𝘞_0$ be $W$.
	For each nonnegative integer $n$, let $𝘞_{n+1}$ be $(𝘞_n)\w{𝘒_{n+1}}$,
	which means $(\cdots(W\w{𝘒_1})\cdots)\w{𝘒_{n+1}}$ in full.
	Recall that at the end of \Cref{sec:encoder} we defined
	$W_♭$, $W_♭\w i$, $(W_♭\w i)\w j$, $((W_♭\w i)\w j)\w k$, et seq.\ 
	all with the same series of input distributions and bijections.
	Let $𝘝_0$ be $W_♭$;
	let $𝘝_{n+1}$ be $(𝘝_n)\w{𝘒_{n+1}}$,
	which means $(\cdots(W_♭\w{𝘒_1})\cdots)\w{𝘒_{n+1}}$ in full.
	These r.v.s provide a new family of randomness that
	does not appear in the encoding and decoding algorithms,
	but they help us understand the code rate $\abs{ℐ}/ℓ^n$ in this manner:
	Counting how many indexes are in $ℐ$ is nothing more than
	measuring the probability $𝘗\{(𝘒_1,𝘒_2…𝘒_n)∈ℐ\}$.
	With the processes $𝘞_n$ and $𝘝_n$ thus defined,
	it is further equivalent to measuring the probability
	$𝘗\{H(𝘞_n)<θ_n† and †1-H(𝘝_n)<θ_n\}$, where $θ_n≔\exp(-ℓ^{πn}n)$.
	Moreover, it suffices to know how $H(𝘞_n)$ and $H(𝘝_n)$
	behave as stochastic processes taking values in $[0,1]$
	without comprehending $𝘞_n$ and $𝘝_n$ themselves.
	The general fact is that $H(𝘞_n)$ is either very small
	(channel is reliable) or very close to $1$ (channel is noisy).
	Arıkan called this phenomenon \emph{channel polarization}.
	The following claim generalizes channel polarization
	and implies the main theorem.
	
	\begin{cla} \label{cla:trichotomy}
		Fix any $π,ρ>0$ such that $π+2ρ<1$.
		We will choose an $ℓ$ and a series of bijections of $𝔽_q^ℓ$---%
		namely, $\g W$, $\g{W\w i}$, $\g{(W\w i)\w j}$,
		$\g{((W\w i)\w j)\w k}$, et seq.---such that
		\begin{gather*}
			𝘗\{H(𝘞_n)<\exp(-ℓ^{πn}n)\}>1-H(W)-ℓ^{-ρn+o(n)}, \\
			𝘗\{1-H(𝘞_n)<\exp(-ℓ^{πn}n)\}>H(W)-ℓ^{-ρn+o(n)}, \\
			𝘗\{H(𝘝_n)<\exp(-ℓ^{πn}n)\}>1-H(W_♭)-ℓ^{-ρn+o(n)},\rlap{\quad and} \\
			𝘗\{1-H(𝘝_n)<\exp(-ℓ^{πn}n)\}>H(W_♭)-ℓ^{-ρn+o(n)}.
		\end{gather*}
		Here, $o(n)$ is the little-$o$ function in $n$;
		it is such that $o(n)/n→0$ as $n→∞$.
	\end{cla}
	
	For polar codes over symmetric channels,
	the first inequality in \Cref{cla:trichotomy} alone implies
	that the code rate is $1-H(W)-ℓ^{-ρn+o(n)}=I(W)-N^{-ρ+o(1)}$.
	The first two inequalities imply the polarization behavior that 
	channels become either satisfactorily reliable (low $H(𝘞_n)$)
	or desperately noisy (high $H(𝘞_n)$).
	For asymmetric channels, however,
	we need to characterize $H(𝘝_n)$ alongside $H(𝘞_n)$.
	The last two inequalities in \Cref{cla:trichotomy} show that the same
	series of bijections polarize $W_♭$ at the same time they polarize $W$.
	While $W_♭$ contains no randomness form the channel $W$,
	what is polarized is that each input bit $U\M{k_1,k_2…k_n}$ either
	depends heavily on lexicographically earlier input bits (low $H(𝘝_n)$)
	or behaves like a free r.v.\ conditioned on earlier bits (high $H(𝘝_n)$).
	We then categorize the fate of indexes in $[ℓ]^n$
	into the following three types.
	(A)	Free and reliable:
		These are indexes that will be in $ℐ$;
		they point to channels that transmit information bits.
	(B)	Free but noisy:
		The sender \emph{can} feed information into these channels
		only to find that the decoder will almost always make some mistakes.
		The sender should, instead,
		feed some pseudo random numbers shared with the receiver.
	(C)	Dependent and reliable.
		The input of these channels depends on previous inputs.
		Their main purpose is to \emph{shape}
		the capacity-achieving input distribution.
	(D)	Dependent but noisy is not possible because $H(𝘝_n)≥H(𝘞_n)$.
		This is the key to \cite[Theorem~1]{HY13}.
		We reproduce their proof in the next subsection.
	
	\begin{figure}
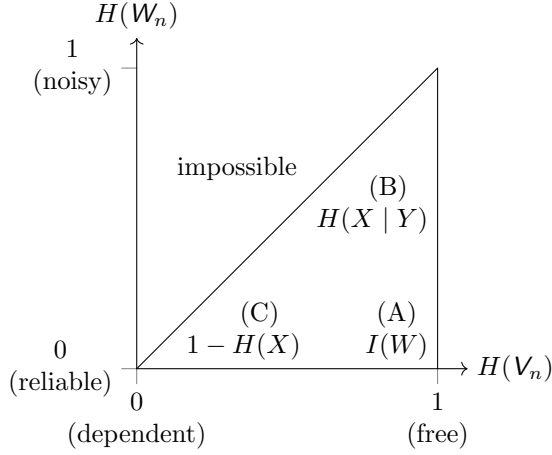
%\workhardtrue
		$$\tikz[scale=4,nodes={black,align=center}]{
			\PMS\wHXY{width("$H(X｜Y)$")*1pt+2em/3-1.2em}\PMS\hHX{1.2em+1em/3}
			\ifworkhard\draw
				(1/3,2/3)node{impossible}
				(1,1cm-\wHXY pt/4)node[below left]{\hskip1.2em(B)\\$H(X｜Y)$}
				(\hHX pt/4,0)node[above right]{\hskip1.2em(C)\\$1-H(X)$}
				(1,0)node[above left]{(A)\\$I(W)$}
			;\fi
			\draw[help lines]
				(-.05,0)coordinate(X)(0,-.05)coordinate(Y)
				(0,1)--+(X)node[left]{1\\(noisy)}
				(0,0)--+(X)node[left]{0\\(reliable)}
				(0,0)--+(Y)node[below]{0\\(dependent)}
				(1,0)--+(Y)node[below]{1\\(free)}
			;
			\draw(0,0)--(1,1)--(1,0);
			\draw[->](0,0)--(1.1,0)node[right]{$H(𝘝_n)$};
			\draw[->](0,0)--(0,1.1)node[above]{$H(𝘞_n)$};
		}$$
		\caption{
			The trichotomy of the fates of synthetic channels.
			Label (A) marks the corner of the free and reliable channels.
			Label $I(W)$ beneath (A) is the limit of the probability measure
			$𝘗(𝘈_n)=𝘗\{𝘞_n\text{ is free and reliable}\}$ as $n→∞$.
			Labels (B) and (C) and the numbers beneath marks
			the corresponding fates and probability measures.
		} \label{fig:trichotomy}
	\end{figure}

\subsection{Claim \ref{cla:trichotomy} implies the main theorem}
	\label{sketch:hypotenuse}
	
	As mentioned, $H(𝘝_n)≥H(𝘞_n)$ so (D) dependent but noisy is not possible.
	Let $𝘈_n$ be the intersection event of free $\{1-H(𝘝_n)<\exp(-ℓ^{πn}n)\}$
	and reliable $\{H(𝘞_n)<\exp(-ℓ^{πn}n)\}$.
	Let $𝘉_n$ be the intersection event of free
	and noisy $\{1-H(𝘞_n)<\exp(-ℓ^{πn}n)\}$.
	Let $𝘊_n$ be the intersection event of dependent
	$\{H(𝘝_n)<\exp(-ℓ^{πn}n)\}$ and reliable.
	Since noisy implies free, $𝘗(𝘉_n)>H(W)-ℓ^{-ρn+o(n)}$
	follows the second inequality in \Cref{cla:trichotomy}.
	Also since dependent implies reliable, $𝘗(𝘊_n)>1-H(W_♭)-ℓ^{-ρn+o(n)}$
	follows the third inequality in \Cref{cla:trichotomy}.
	Note that $𝘈_n$ or $𝘉_n$ implies
	free but not ``neither reliable nor noisy'';
	that is, $(†free†∧†reliable†)∨(†free†∧†noisy†)
		→†free†∧¬(¬†reliable†∨¬†noisy†)$.
	We deduce that $𝘗(𝘈_n)+𝘗(𝘉_n)>H(W_♭)-ℓ^{-ρn+o(n)}-2ℓ^{-ρn+o(n)}$.
	Similarly, since $𝘈_n$ or $𝘊_n$ implies
	reliable but not ``neither free nor dependent,''
	we deduce that $𝘗(𝘈_n)+𝘗(𝘊_n)>1-H(W)-ℓ^{-ρn+o(n)}-2ℓ^{-ρn+o(n)}$.
	In summary, we derive that
	\begin{align*}
		𝘗(𝘈_n)
		&≥	(𝘗(𝘈_n)+𝘗(𝘉_n))+(𝘗(𝘈_n)+𝘗(𝘊_n))-1 \\
		&>	(H(W_♭)-3ℓ^{-ρn+o(n)})+(1-H(W)-3ℓ^{-ρn+o(n)})-1 \\
		&=	H(X)-H(X｜Y)-6ℓ^{-ρn+o(n)} = I(W)-ℓ^{-ρn+o(n)}
	\end{align*}
	Finally, recall that $ℐ$ collects free and reliable indexes,
	so the code rate is $\abs{ℐ}/ℓ^n=𝘗(𝘈_n)>I(W)-ℓ^{-ρn+o(n)}$.
	We almost finish the proof of the main theorem except that
	we claimed $I(W)-N^{-ρ}=I(W)-ℓ^{-ρn}$, without the little-$o$ term.
	It can be fixed by finding a slightly larger $ϱ>ρ$ such that $π+2ϱ<1$
	still holds, and then rerunning the whole argument again with the new $ϱ$.
	The conclusion becomes that the code rate is at least $I(W)-ℓ^{-ϱn+o(n)}$.
	Since $-ϱn+o(n)<-ρn$ for sufficiently large $n$,
	this completes the proof of the main theorem.
	It remains to show that \Cref{cla:trichotomy} can be achieved.

\section{Global MDP Behavior Modulo Local Behaviors} \label{sec:globalMDP}

	In this section, we put constraints on an abstract process $\{𝘏_n\}$
	and show that they imply inequalities of the form
	\[*𝘗\{𝘏_n<†threshold†\}>†limit measure†-†decaying gap†\]*
	as those in \Cref{cla:trichotomy}.
	Let $ℱ_n$ be the sigma-algebra generated by $𝘒_1,𝘒_2…𝘒_n$ for each $n$.
	Then $ℱ_0⊂ℱ_1⊂ℱ_2⊂\dotsb$ form  a filtration of sigma-algebras.
	Let $\{𝘏_n\}$, $\{𝘡_n\}$, and $\{𝘚_n\}$ be three stochastic processes
	adapted to $\{ℱ_n\}$ (meaning $𝘒_1,𝘒_2…𝘒_n$ determine $𝘏_n,𝘡_n,𝘚_n$).
	The following assumptions are easy to verify
	when we reveal what those processes are:
	(cb)	$0≤𝘏_n,𝘡_n,𝘚_n$ and $𝘏_n≤1$;
	(cm)	$\{𝘏_n\}$ is a martingale, i.e., $𝘌[𝘏_{n+1}｜ℱ_n]=𝘏_n$;
	(ct)	$𝘏_n≤q^3√{𝘡_n}$ and $𝘡_n≤q^3√{𝘏_n}$ along with
			$1-𝘏_n≤q^3√{𝘚_n}$ as well as $𝘚_n≤q^3√{1-𝘏_n}$ for all $n$.
	Furthermore, assume large kernels:
	(cl)	$ℓ≥\max(e^4,q^5,3^q)$.
	Let $α≔\log(\logℓ)/\logℓ$ be a small number shrinking as $ℓ$ increases.
	Define the potential function
	$h_α:[0,1]→[0,1]$ to be $h_α(z)≔\min(z,1-z)^α$.
	(Remark: $h_2$ is not a special case of $h_α$ for $α=2$;
	we expect $α\ll1$ in practice.)
	Here are the difficult but sufficient criteria for the main theorem.
	
	\begin{lem}(calculus machinery for global MDP) \label{lem:triforce}
		Assume criteria (cb), (cm), (ct), and (cl).
		Assume the \emph{local LDP behavior}:
			$𝘡_{n+1}≤ℓ\exp(q𝘡_nℓ)(q𝘡_n)^{⌈𝘒_{n+1}^2/3ℓ⌉}$
			and $𝘚_{n+1}≤ℓ\exp(q𝘚_nℓ)(q𝘚_n)^{⌈(ℓ+1-𝘒_{n+1})^2/3ℓ⌉}$.
		Assume the \emph{local CLT behavior}:
			$𝘌[h_α(𝘏_{n+1})｜ℱ_n]<4ℓ^{-1/2+α}$.
		Then, for any constants $π,ρ>0$ such that
		\[π+2ρ≤1-8α, \label{eq:finiteif}\]
		the following holds:
		\[𝘗\{𝘏_n<\exp(-ℓ^{πn}n)\}>1-𝘏_0-ℓ^{-ρn+o(n)}. \label{eq:finiteso}\]
	\end{lem}
	
	We defer the proof until \Cref{pf:triforce}.
	The term $𝘒_{n+1}^2/3ℓ$ in the lemma is to control
	the local LDP behavior of the process $\{𝘏_n\}$---%
	the behavior of $𝘏_{n+1}$ when $𝘏_n$ is close to $0$ and
	the behavior that is closely related to the LDP behavior of polar codes.
	The term is chosen in a way such that $h_2((k^2/3ℓ)/ℓ)<k/ℓ$
	and such that $∑_k(k^2/3ℓ)^t$ is easy to handle.
	In \cite[Theorem~7]{FHMV17}, a similar criterion is stated
	and is annotated as \emph{faster polarization at the tails}.
	In \cite[Definition~2.4]{BGS18}, a similar criterion is stated
	and is annotated as \emph{strong suction at the low end}.
	The \emph{eigenfunction} $h_α$ in the lemma is to control
	the local CLT behavior of the process $\{𝘏_n\}$---%
	the behavior of $𝘏_n$ when it is away from $0$ and
	the behavior that is closely related to the CLT behavior of polar codes.
	In \cite[Theorem~7]{FHMV17}, a similar criterion is
	annotated as \emph{near optimal polarization in the middle}
	with $h_†FHMV†(z)≔(z(1-z))^α$ for positive
	but small $α$ at most $\log(\logℓ)/\logℓ$.
	In \cite[Definition~2.3]{BGS18}, a similar criterion is
	annotated as \emph{variance in the middle} with $h_†BGS†(z)≔√{\min(z,1-z)}$.
	Note how our choice of $h_α(z)≔\min(z,1-z)^α$ resembles theirs.
	In both cases, the criteria are \emph{local} because
	they refer to a small slice of the process, focusing on
	how $𝘏_{n+1}$ (or $𝘡_{n+1}$) behaves in terms of $𝘏_n$ (or $𝘡_n$).
	This perspective frees \cite{FHMV17,BGS18} from
	considering the (\emph{global}) process $\{𝘏_n\}$
	as a whole and simplifies the analysis.
	We specifically benefiti from the fact that we can choose the bijection
	$\g{(\cdots(W\w{k_1})\cdots)\w{k_n}}$ solely according to the channels
	$(\cdots(W\w{k_1})\cdots)\w{k_n}$ and $(\cdots(W_♭\w{k_1})\cdots)\w{k_n}$
	instead of the complete channel family-tree.
	This is also the approach taken in \cite{GRY19}.

\subsection{Lemma \ref{lem:triforce} helps achieve
	Claim \ref{cla:trichotomy}} \label{sketch:trichotomy}

	The formulation and the choice of the variables make it clear how
	\Cref{lem:triforce} will be applied to support \Cref{cla:trichotomy}.
	For instance, if we let $\{𝘏_n\}$, $\{𝘡_n\}$ and $\{𝘚_n\}$ be
	$\{H(𝘞_n)\}$, $\{\Z(𝘞_n)\}$, and $\{\S(𝘞_n)\}$,
	respectively, then \Cref{lem:triforce} supports
	the first of the four inequalities in \Cref{cla:trichotomy}.
	Moreover, if we let $\{𝘏_n\}$, $\{𝘡_n\}$, and $\{𝘚_n\}$ be
	$\{1-H(𝘞_n)\}$, $\{\S(𝘞_n)\}$, and $\{\Z(𝘞_n)\}$, then
	\Cref{lem:triforce} supports the second inequality of \Cref{cla:trichotomy}.
	If $\{𝘏_n\}$, $\{𝘡_n\}$, and $\{𝘚_n\}$ are let to be $\{H(𝘝_n)\}$,
	$\{\Z(𝘝_n)\}$, and $\{\S(𝘝_n)\}$, that supports the third inequality.
	If $\{𝘏_n\}$, $\{𝘡_n\}$, and $\{𝘚_n\}$ are let to be $\{1-H(𝘝_n)\}$,
	$\{\S(𝘝_n)\}$, and $\{\Z(𝘝_n)\}$, that supports the forth inequality.
	The criteria (cb), (cm), (ct), and (cl)
	listed above \Cref{lem:triforce} are easy to verify;
	for instance, \Cref{lem:extoll} implies (ct) for all four cases.
	It remains to show that for each of the four triples of processes,
	the local LDP behavior and the local CLT behavior hold.
	
	To do so, one advantage is that the two desired behaviors are local.
	They only involve how $𝘏_{n+1}$, $𝘡_{n+1}$ and $𝘚_{n+1}$
	behave conditioned on the history $ℱ_n$.
	A potential tedious aspect is that for each candidate of the bijection
	$\g{(\cdots(W\w{k_1})\cdots)\w{k_n}}$, we have to verify the two behaviors
	four times, once for each of the four triples of channel parameters.
	Luckily, within random coding theory, we are in the situation
	that to choose an object that satisfies multiple criteria,
	it suffices to choose the object from an ensemble and
	compute the probabilities that each criterion fails;
	as long as the sum of failing probabilities is small, most objects satisfy.
	Even more luckily, when we choose a bijection
	$\g{(\cdots(W\w{k_1})\cdots)\w{k_n}}$ from some ensemble,
	we only have to compute the probability that the local CLT or LDP behavior
	fails for $\{𝘏_n\}$, $\{𝘡_n\}$, and $\{𝘚_n\}$ being $\{H(𝘞_n)\}$,
	$\{\Z(𝘞_n)\}$, and $\{\S(𝘞_n)\}$ but not the other three triples.
	This is because other three triples are the special case
	and/or the \emph{dual} of this triple.
	Elaboration:
	Since $𝘝_n$ are $q$-ary channels just like $𝘞_n$ are,
	inequalities hold true for arbitrary $𝘞_n$ should hold true for any $𝘝_n$.
	Also, since $\Z$ and $\S$ are in duality,
	inequalities hold true for $H,\Z,\S$ hold true for $1-H,\S,\Z$.
	The duality is due to the duality between FTPC$Z$ and FTPC$S$,
	within the explicit Hölder tolls, and within the ensemble of bijections
	we are to choose $\g{(\cdots(W\w{k_1})\cdots)\w{k_n}}$ from.

\subsection{Random linear isomorphisms as bijections} \label{sec:chimera}

	Fix $q$ and $ℓ$.
	Let $\GL(ℓ,q)$ be the group of $ℓ$-by-$ℓ$ invertible matrices
	over $𝔽_q$ together with the ordinary matrix multiplication.
	Select an element $𝔾∈\GL(ℓ,q)$ uniformly at random.
	Let $\g W:𝔽_q^ℓ→𝔽_q^ℓ$ be the multiplication of $𝔾$ from the right,
	namely $\g W(u_1^ℓ)≔u_1^ℓ𝔾$.
	This map is bijective since $𝔾$ is invertible.
	Let $W$ be a $q$-ary channel.
	Recall that we defined $q$-ary channels
	$W\w1,W\w2…W\w{ℓ}$ in \Cref{sec:transform}.
	To emphasis that these imaginary channels
	depend heavily on the randomness source $𝔾$,
	we call them $W\wg 1,W\wg 2…W\wg{ℓ}$ instead.
	The following two lemmas help verify
	the two criteria in \Cref{lem:triforce}.
	Proofs are given in upcoming sections, \ref{pf:LDP} and \ref{pf:CLT}.
	
	\begin{lem}[local LDP behavior] \label{lem:LDP}
		Fix an $ℓ≥30$.
		Let $𝔾$ vary;
		with probability less than $3q^{-√ℓ/13}$,
		each of the following fails for each $i∈[ℓ]$:
		\begin{gather}
			\Z(W\wg i)≤ℓ\exp(q\Z(W)ℓ)(q\Z(W))^{⌈i^2/3ℓ⌉},
				\rlap{\quad and} \label{eq:LDPZ}\\
			\S(W\wg i )≤ℓ\exp(q\S(W)ℓ)(q\S(W))^{⌈(ℓ+1-i)^2/3ℓ⌉}. \label{eq:LDPS}
		\end{gather}
	\end{lem}
	
	\begin{lem}[local CLT behavior] \label{lem:CLT}
		Fix an $ℓ≥20$.
		Recall $α≔\log(\logℓ)/\logℓ$ and $h_α(z)≔\min(z,1-z)^α$.
		Let $𝔾$ vary;
		with probability less than $2ℓ^{-\log(ℓ)/20}$, this fails:
		\[÷1{ℓ}∑_{i=1}^ℓh_α(H(W\wg i))<4ℓ^{-1/2+α}. \label{eq:CLT}\]
	\end{lem}

\subsection{Local behaviors imply Claim \ref{cla:trichotomy}
	(and hence the main theorem)} \label{pf:trichotomy}

	We now can see how \Cref{lem:triforce,lem:LDP,lem:CLT} imply
	that \Cref{cla:trichotomy} is achievable for the right choice
	of $ℓ$ and bijections $\g W$, $\g{W\w i}$, et seq.:
	For any given $q$-ary channel $W$, let $ℓ$ be $\max(e^4,q^5,3^q)$.
	For any given $π,ρ>0$ such that $π+2ρ<1$, enlarge $ℓ$ such that
	Inequality~\eqref{eq:finiteif} holds, given $α≔\log(\logℓ)/\logℓ$.
	Consider a random kernel $𝔾$ as a candidate of the bijection $\g W$.
	Increase $ℓ$ further so that the failing probabilities---%
	\Cref{lem:LDP}'s $3q^{-√ℓ/13}$ and \Cref{lem:CLT}'s $2ℓ^{-\log(ℓ)/20}$%
	---amount to $1/3$ or less.
	% per our calculation using further optimization, ℓ ≥ 8400 works.
	Recall the flattening channel $W_♭$.
	The probability that any of the inequalities in \Cref{lem:LDP,lem:CLT}
	fails for $W_♭$ is less than $1/3$, too.
	Invoke the union bound; $1/3+1/3<1$.
	Hence there exists a solid choice of $\g W$ as
	the multiplication of some proper instance of $𝔾$ from the right.
	With this $\g W$ determined,
	we define $W\w i$ and $W_♭\w i$ for all $i∈[ℓ]$.
	Consider first $i=1$,
	anything that has been done to $W$ now applies to $W\w i$.
	That is, let $\g{W\w i}$ be the multiplication
	of a random kernel $𝔾$ from the right.
	With $W$, $i$, and $W\wg i$ replaced by $W\w i$, $j$, and $(W\w i)\wg j$,
	the probabilities that inequalities in
	\Cref{lem:LDP,lem:CLT} fail add up to $1/3$ or less.
	So is the flattening ($♭$) counterpart.
	Hence there is a solid choice of $\g{W\w i}$.
	Repeat this for every other $i=2,3…ℓ$.
	Once finished, proceed to choosing $\g{(W\w i)\w j}$ for all $i,j∈[ℓ]$.
	And so on and so forth for cases beyond depth-$2$.
	Notice that we always make a solid choice of a bijection before
	we proceed to the next level of channels, hence the failing probabilities
	of \Cref{lem:LDP,lem:CLT} do not accumulate as the depth increases.
	By how we select bijections, the criteria in \Cref{lem:triforce}
	hold for $(\{𝘏_n\},\{𝘡_n\},\{𝘚_n\})$ being these four triples:
	\let\AB\allowbreak
	$(\{H(𝘞_n)\},\AB\{\Z(𝘞_n)\},\AB\{S(𝘞_n)\})$ and
	$(\{1-H(𝘞_n)\},\AB\{\S(𝘞_n)\},\AB\{\Z(𝘞_n)\})$ along with
	$(\{H(𝘝_n)\},\AB\{\Z(𝘝_n)\},\AB\{\S(𝘝_n)\})$ as well as
	$(\{1-H(𝘝_n)\},\AB\{\S(𝘝_n)\},\AB\{\Z(𝘝_n)\})$
	Hence the process $\{𝘏_n\}$ satisfies
	Inequality~\eqref{eq:finiteso} for the four processes: $\{H(𝘞_n)\}$
	and $\{1-H(𝘞_n)\}$ along with $\{H(𝘝_n)\}$ as well as $\{1-H(𝘝_n)\}$.
	This results in the four inequalities in \Cref{cla:trichotomy}.
	And we are done.
	It remains to prove \Cref{lem:triforce,lem:LDP,lem:CLT}
	in order to prove the main theorem.

\section{Local LDP Behavior (Proof of Lemma~\ref{lem:LDP})} \label{pf:LDP}

	In this section, we will first prove the two
	fundamental theorems of polar coding in
	\Cref{pf:FTPCZ} (for the $Z$-end) and in \Cref{pf:FTPCS} (for the $S$-end).
	And then we will target that
	the following inequalities hold with high probability:
	\begin{gather*}
		\Z(W\wg i)≤ℓ\exp(q\Z(W)ℓ)(q\Z(W))^{⌈i^2/3ℓ⌉},
			\rlap{\quad and} \repeattag{eq:LDPZ}\\
		\S(W\wg{ℓ+1-i})≤ℓ\exp(q\S(W)ℓ)(q\S(W))^{⌈i^2/3ℓ⌉}. \repeattag{eq:LDPS}
	\end{gather*}
	By the duality between the two fundamental theorems and between
	the two targeted inequalities, it is not hard to see that it suffices
	to prove the $\Z$-case and the $\S$-case follows immediately.
	We will prove that the first targeted inequality, for each $i∈[ℓ]$,
	holds with probability $1-3q^{-√ℓ/13}$ in \Cref{sec:weight},
	closing this section.

\subsection{Proof of FTPC\TP\U1D44D+Z$Z$ (Theorem \ref{thm:FTPCZ})}
	\label{pf:FTPCZ}

	As is promised in \Cref{sec:FTPC},
	we prove the two fundamental theorems of polar coding.
	We first go for the $Z$-end.
	
	Recall that $\fz G\w i(z)≔∑_{u_{i+1}^ℓ}z^{\wt(0_1^{i-1}1_iu_{i+1}^ℓG)}$
	is the weight enumerator of the $i$-th coset code.
	\Cref{thm:FTPCZ} claims that $\Z(W\w i)≤\fz G\w i(\Z(W))$.
	By the definition of $W\w i$ and
	the definition of the Bhattacharyya parameter, $\Z(W\w i)$ is
	\[*\max_{0≠d_i∈𝔽_q}∑_{u_i∈𝔽_q}∑_{u_1^{i-1}y_1^ℓ∈𝔽_q^i×𝒴^ℓ}
		√{W\w i(u_i,u_1^{i-1}y_1^ℓ)W\w i(u_i+d_i,u_1^{i-1}y_1^ℓ)}.\]*
	By the nature of $\max_{0≠d_i∈𝔽_q}$, it suffices to show that
	the double sum is at most $\fz G\w i(\Z(W))$ for arbitrary nonzero $d_i$.
	In the upcoming argument, tuple concatenation takes precedence over
	vector-matrix multiplication and vector addition.
	Fix a $d_i$, we argue that
	\begin{align*}
		&{}	∑_{u_i∈𝔽_q}∑_{u_1^{i-1}y_1^ℓ∈𝔽_q^i×𝒴^ℓ}
			√{W\w i(u_i,u_1^{i-1}y_1^ℓ)W\w i(u_i+d_i,u_1^{i-1}y_1^ℓ)} \\
		&=	∑_{u_1^iy_1^ℓ}
			√{W\w i(u_i,u_1^{i-1}y_1^ℓ)W\w i(u_i+d_i,u_1^{i-1}y_1^ℓ)} \\
		&=	∑_{u_1^iy_1^ℓ}\vphantom{∑_{𝔽_q^ℓ}}√{\vphantom{∑_u}\smash[b]{
			∑_{u_{i+1}^ℓ∈𝔽_q^{ℓ-i}}W^ℓ(u_1^iu_{i+1}^ℓG,y_1^ℓ)
			∑_{v_{i+1}^ℓ∈𝔽_q^{ℓ-i}}W^ℓ(u_1^{i-1}(u_i+d_i)v_{i+1}^ℓG,y_1^ℓ)}} \\
		&≤	∑_{u_1^iy_1^ℓ}∑_{u_{i+1}^ℓ}∑_{v_{i+1}^ℓ}√{W^ℓ(u_1^iu_{i+1}^ℓG,y_1^ℓ)
			W^ℓ(u_1^{i-1}(u_i+d_i)v_{i+1}^ℓG,y_1^ℓ)} \\
		&=	∑_{y_1^ℓ}∑_{u_1^ℓ}∑_{d_{i+1}^ℓ∈𝔽_q^{ℓ-i}}
			√{W^ℓ(u_1^ℓG,y_1^ℓ)W^ℓ(u_1^{i-1}(u_i^ℓ+d_i^ℓ)G,y_1^ℓ)} \\
		&=	∑_{y_1^ℓ}∑_{x_1^ℓ∈𝔽_q^ℓ}∑_{d_{i+1}^ℓ}
			√{W^ℓ(x_1^ℓ,y_1^ℓ)W^ℓ(x_1^ℓ+0_1^{i-1}d_i^ℓG,y_1^ℓ)} \\
		&=	∑_{d_{i+1}^ℓ}∑_{y_1^ℓ}∑_{x_1^ℓ}
			√{W^ℓ(x_1^ℓ,y_1^ℓ)W^ℓ(x_1^ℓ+e_1^ℓ,y_1^ℓ)} \\
		&=	∑_{d_{i+1}^ℓ}∑_{y_1^ℓ}∑_{x_1^ℓ}
			∏_{j∈[ℓ]}√{W(x_j,y_j)W(x_j+e_j,y_j)} \\
		&=	∑_{d_{i+1}^ℓ}∑_{y_1^ℓ}∑_{x_1^ℓ}
			∏_{j∈J}√{W(x_j,y_j)W(x_j+e_j,y_j)}∏_{k∉J}W(x_k,y_k) \\
		&=	∑_{d_{i+1}^ℓ}∏_{j∈J}（∑_{x_jy_j}√{W(x_j,y_j)W(x_j+e_j,y_j)}）
			∏_{k∉J}（∑_{x_ky_k}W(x_k,y_k)） \\
		&=	∑_{d_{i+1}^ℓ}∏_{j∈J}（∑_{x_jy_j}√{W(x_j,y_j)W(x_j+e_j,y_j)}） \\
		&≤	∑_{d_{i+1}^ℓ}∏_{j∈J}\max_{0≠e_j∈𝔽_q}
			（∑_{x_jy_j}√{W(x_j,y_j)W(x_j+e_j,y_j)}） \\
		&=	∑_{d_{i+1}^ℓ}∏_{j∈J}\Z(W) = ∑_{d_{i+1}^ℓ}\Z(W)^{\abs J}
		=	∑_{d_{i+1}^ℓ}\Z(W)^{\wt(0_1^{i-1}d_id_{i+1}^ℓG)} \\
		&=	∑_{d_{i+1}^ℓ}\Z(W)^{\wt(0_1^{i-1}1_id_{i+1}^ℓG)} = \fz G\w i(\Z(W)).
	\end{align*}
	The first equality abbreviates the summation.
	The next equality expands $W\w i$ by the very definition,
	where $u_{i+1}^ℓ$ and $v_{i+1}^ℓ$ are free variables in $𝔽_q$.
	The next inequality is by the sub-additivity of the square root.
	In the next equality we define $d_{i+1}^ℓ≔v_{i+1}^ℓ-u_{i+1}^ℓ$;
	so summing over $v_{i+1}^ℓ$ is equivalent to summing over $d_{i+1}^ℓ$.
	In the next equality we define $x_1^ℓ≔u_1^ℓG$;
	so summing over $u_1^ℓ$ is equivalent to summing over $x_1^ℓ$
	as $G$ is invertible.
	In the next equality we substitute $e_1^ℓ≔0_1^{i-1}d_i^ℓG$
	and reorder the summation.
	The next equality expands the product of the memoryless channels.
	The next equality classifies indexes into two classes---%
	$j∈J$ are those such that $e_j≠0$ and $k∉J$ are such that $e_k=0$.
	The next equality is the distributive law $ax+ay+bx+by=(a+b)(x+y)$.
	The next equality uses the fact that the $W(x,y)$ sum to $1$.
	In the next inequality we replace $e_j$ by a nonzero element
	that maximizes the sum in the parentheses.
	In the next equality we realize that
	the maximum is the Bhattacharyya parameter (surprisingly).
	The second last equality uses the fact that
	multiplying a vector by a scalar preserves its hamming weight.
	And quod erat demonstrandum.
	Experienced readers may find that all but the last inequality
	follows the proof strategy \cite[Lemma~10]{KSU10}.

\subsection{Proof of FTPC\TP\U1D446+$S$ (Theorem \ref{thm:FTPCS})}
	\label{pf:FTPCS}

	We now go for the $S$-end of the fundamental theorem of polar coding.
	Recall the character $χ(x)≔\exp(2πi\tr(x)/p)$.
	We need the following properties:
	(pa)	$χ(0)=1$;
	(pb)	$\abs{χ(x)}=1$ for all $x∈𝔽_q$;
	(pc)	$χ(x)χ(z)=χ(x+z)$ for all $x,z∈𝔽_q$;
	(pd)	$∑_{x∈𝔽_q}χ(x)=0$.
	See also \cite[Definition~24]{MT14} or a dedicated book \cite{Terras99}.
	To prove the theorem,
	we first verify that Fourier coefficients recover the origin:
	Let $M(w,y)≔W\out(y)M(w｜y)=∑_{z∈𝔽_q}W(z,y)χ(wz)$, then
	\begin{multline*}
		∑_{w∈𝔽_q}M(w,y)χ(-xw) = ∑_{w∈𝔽_q}∑_{z∈𝔽_q}W(z,y)χ(wz)χ(-xw) \\
		=	∑_{z∈𝔽_q}W(z,y)∑_{w∈𝔽_q}χ(w(z-x)) = ∑_{z∈𝔽_q}W(z,y)q𝕀\{z-x=0\}
		=	qW(x,y).
	\end{multline*}
	The first equality expands $M(w,y)$ by the definition.
	The next equality uses that $χ$ is an additive character (pc),
	and reorders the summation.
	The next equality uses $∑_{w∈𝔽_q}χ(w)=0$ (pd) and $∑_{w∈𝔽_q}χ(0)=q$ (pa);
	and $𝕀$ is the indicator function.
	
	Knowing that $W(x_j,y_j)=q^{-1}∑_{w_j∈𝔽_q}M(w_j,y_j)χ(-x_jw_j)$,
	we proceed to
	\begin{align*}
		&{}	W\w i(u_i,u_1^{i-1}y_1^ℓ) = ∑_{u_{i+1}^ℓ}W^ℓ(u_1^ℓG,y_1^ℓ)
		=	∑_{u_{i+1}^ℓ∈𝔽_q^{ℓ-i}}W^ℓ(x_1^ℓ,y_1^ℓ)
		=	∑_{u_{i+1}^ℓ}∏_{j∈[ℓ]}W(x_j,y_j) \\
		&=	∑_{u_{i+1}^ℓ}∏_{j∈[ℓ]}（q^{-1}∑_{w_j∈𝔽_q}M(w_j,y_j)χ(-x_jw_j)）
		=	q^{-ℓ}∑_{u_{i+1}^ℓ}∑_{w_1^ℓ}∏_{j∈[ℓ]}M(w_j,y_j)χ(-x_jw_j) \\
		&=	q^{-ℓ}∑_{u_{i+1}^ℓ}∑_{w_1^ℓ}χ(-x_1^ℓ(w_1^ℓ)^⊤)∏_{j∈[ℓ]}M(w_j,y_j)
		=	q^{-ℓ}∑_{u_{i+1}^ℓ}∑_{w_1^ℓ}χ(-x_1^ℓ(w_1^ℓ)^⊤)M^ℓ(w_1^ℓ,y_1^ℓ) \\
		&=	q^{-ℓ}∑_{\!u_{i+1}^ℓ\!}∑_{w_1^ℓ}
			χ(-u_1^ℓG(w_1^ℓ)^⊤)M^ℓ(w_1^ℓ,y_1^ℓ)
		=	q^{-ℓ}∑_{\!u_{i+1}^ℓ\!}∑_{w_1^ℓ}
			χ(-u_1^ℓ(w_1^ℓG^⊤)^⊤)M^ℓ(w_1^ℓ,y_1^ℓ) \\
		&=	q^{-ℓ}∑_{u_{i+1}^ℓ}∑_{v_1^ℓ}
			χ(-u_1^ℓ(v_1^ℓ)^⊤)M^ℓ(v_1^ℓ\Git,y_1^ℓ) \\
		&=	q^{-ℓ}∑_{v_1^ℓ}χ(-u_1^i(v_1^i)^⊤)M^ℓ(v_1^ℓ\Git,y_1^ℓ)
			（∑_{u_{i+1}^ℓ}χ(-u_{i+1}^ℓ(v_{i+1}^ℓ)^⊤)） \\
		&=	q^{-ℓ}∑_{v_1^ℓ}χ(-u_1^i(v_1^i)^⊤)M^ℓ(v_1^ℓ\Git,y_1^ℓ)
			q^{ℓ-i}𝕀\{v_{i+1}^ℓ=0\} \\
		&=	q^{-i}∑_{v_1^i}χ(-u_1^i(v_1^i)^⊤)M^ℓ(v_1^i0_{i+1}^ℓ\Git,y_1^ℓ).
	\end{align*}
	The first equality expands the definition of $W\w i$.
	In the next equality, we substitute $x_1^ℓ≔u_1^ℓG$.
	The next equality expand the definition of $W^ℓ$ down to $W$.
	The next two equalities Fourier expand $W$ and reorder the operators.
	The next equality merges all $χ(-x_jw_j)$ into one term by additivity (pc).
	In the next equality we define $M^ℓ(w_1^ℓ,y_1^ℓ)$
	to be the product of all $M(w_j,y_j)$.
	The next two equalities use
	$x_1^ℓ(w_1^ℓ)^⊤=u_1^ℓG(w_1^ℓ)^⊤=u_1^ℓ(w_1^ℓG^⊤)^⊤$.
	In the next equality we define $v_1^ℓ≔w_1^ℓG^⊤$;
	so summing over $w_1^ℓ$ is equivalent to summing over $v_1^ℓ$.
	(Recall that $\Git$ is the notation of the inverse transpose of $G$.)
	The next three equalities sum over $u_{i+1}^ℓ$ to force $v_{i+1}^ℓ=0$.
	
	Having $W\w i(u_i,u_1^{i-1}y_1^ℓ)=
		q^{-i}∑_{v_1^i}χ(-u_1^i(v_1^i)^⊤)M^ℓ(v_1^i0_{i+1}^ℓ\Git,y_1^ℓ)$
	in mind, we move on to
	\begin{align*}
		&{}	M\w i(ω_i,u_1^{i-1}y_1^ℓ)
		≔	∑_{z_i∈𝔽_q}W\w i(z_i,u_1^{i-1}y_1^ℓ)χ(ω_iz_i) \\
		&=	∑_{z_i∈𝔽_q}q^{-i}∑_{v_1^i}χ(-u_1^{i-1}z_i(v_1^i)^⊤)
			M^ℓ(v_1^i0_{i+1}^ℓ\Git,y_1^ℓ)χ(ω_iz_i) \\
		&=	q^{-i}∑_{v_1^i}χ(-u_1^{i-1}(v_1^{i-1})^⊤)
			M^ℓ(v_1^i0_{i+1}^ℓ\Git,y_1^ℓ)（∑_{z_i∈𝔽_q}χ(z_i(ω_i-v_i))） \\
		&=	q^{-i}∑_{v_1^i}χ(-u_1^{i-1}(v_1^{i-1})^⊤)
			M^ℓ(v_1^i0_{i+1}^ℓ\Git,y_1^ℓ)q𝕀\{ω_i=v_i\} \\
		&=	q^{1-i}∑_{v_1^{i-1}}χ(-u_1^{i-1}(v_1^{i-1})^⊤)
			M^ℓ(v_1^{i-1}ω_i0_{i+1}^ℓ\Git,y_1^ℓ).
	\end{align*}
	In the first line we let $M\w i$ be the Fourier coefficient of $W\w i$.
	The next equality plugs in what we have about $W\w i$ in mind.
	The next three equalities sum over $z_i$ to force $v_i=ω_i$.
	
	With $M\w i(ω_i,u_1^{i-1}y_1^ℓ)=q^{1-i}∑_{v_1^{i-1}}
		χ(-u_1^{i-1}(v_1^{i-1})^⊤)M^ℓ(v_1^{i-1}ω_i0_{i+1}^ℓ\Git,y_1^ℓ)$
	in place, we obtain that with arbitrary $0≠ω_i∈𝔽_q$,
	\begin{align*}
		&{}	∑_{u_1^{i-1}y_1^ℓ∈𝔽^{i-1}×𝒴^ℓ}
			\abs{M\w i(ω_i,u_1^{i-1}y_1^ℓ)} \steplabel{eq:Somegai}\\
		&=	∑_{u_1^{i-1}y_1^ℓ}q^{1-i}\abs[\Big]{∑_{v_1^{i-1}}
			χ(-u_1^{i-1}(v_1^{i-1})^⊤)M^ℓ(v_1^{i-1}ω_i0_{i+1}^ℓ\Git,y_1^ℓ)} \\
		&≤	∑_{u_1^{i-1}y_1^ℓ}q^{1-i}∑_{v_1^{i-1}}
			\abs{M^ℓ(v_1^{i-1}ω_i0_{i+1}^ℓ\Git,y_1^ℓ)}
		=	∑_{y_1^ℓ}∑_{v_1^{i-1}}\abs{M^ℓ(v_1^{i-1}ω_i0_{i+1}^ℓ\Git,y_1^ℓ)} \\
		&=	∑_{y_1^ℓ}∑_{v_1^{i-1}}∏_{j∈[ℓ]}\abs{M(w_j,y_j)}
		=	∑_{y_1^ℓ}∑_{v_1^{i-1}}∏_{j∈J}\abs{M(w_j,y_j)}
			∏_{k∉J}\abs{M(w_k,y_k)} \\
		&=	∑_{v_1^{i-1}}∏_{j∈J}（∑_{y_j}\abs{M(w_j,y_j)}）
			∏_{k∉J}（∑_{y_k}\abs{M(w_k,y_k)}）
		=	∑_{v_1^{i-1}}∏_{j∈J}（∑_{y_j}\abs{M(w_j,y_j)}） \\
		&≤	∑_{v_1^{i-1}}∏_{j∈J}\S(W) = ∑_{v_1^{i-1}}\S(W)^{\abs J}
		=	∑_{v_1^{i-1}}\S(W)^{\wt(v_1^{i-1}ω_i0_{i+1}^ℓ\Git)} \\
		&=	∑_{v_1^{i-1}}\S(W)^{\wt(v_1^{i-1}1_i0_{i+1}^ℓ\Git)}
		=	\fs G\w i(\S(W)).
	\end{align*}
	The first inequality expands the Fourier coefficient.
	The next inequality is triangle plus (pb).
	The next equality cancels the summation over $u_1^{i-1}$ with $q^{1-i}$.
	In the next equality we substitute $w_1^ℓ≔v_1^{i-1}ω_i0_{i+1}^ℓ\Git$;
	slightly different from the $w_1^ℓ$ above,
	they are now restricted to a proper subspace.
	The next equality classifies indexes into two classes---%
	$j∈J$ are those such that $w_j≠0$ and $k∉J$ are such that $w_k=0$.
	The next two equalities reorder the operators and
	simplify $∑_{y_k}\abs{M(0,y_k)}=∑_{y_k}W\out(y_k)=1$.
	The next inequality replaces $w_j$ by the one
	that maximizes $∑_{y_k}\abs{M(w_j,y_k)}$.
	The rest is trivial.
	
	\Cref{thm:FTPCS} claims that $\S(W\w i)≤\fs G\w i(\S(W))$,
	where $\fs G\w i$ is the weight enumerator of the $i$-th dual coset code.
	Since $\S(W\w i)$ is just the maximum of Formula~\eqref{eq:Somegai}
	over $0≠ω_i∈𝔽_q$, we arrive at $\S(W\w i)≤\fs G\w i(\S(W))$.
	And quod erat demonstrandum.
	Experienced readers may find that all but the last inequality
	is a duplicate of \cite[Lemma~34]{MT14}.

\subsection{An upper bound on entropy functions}

	For all $z∈[0,1]$,
	\[*h_2(z)≤√{ez}.\]*
	See the left half of \Cref{fig:hqre} for evidence.
	More generally, for all prime powers $q$,
	\[*1-÷1{\log q}𝔻（z‖÷{q-1}q）=-z\log_q÷z{q-1}-(1-z)\log_q(1-z)≤√{ez}\]*
	for all $z∈[0,1]$, where $𝔻$ is the Kullback--Leibler divergence.
	This falls back to the $h_2$ case when $q=2$.
	See the right half of \Cref{fig:hqre} for evidences for $q=3,4,5,7$.
	It can be observed that as $q→∞$ the function tends to a line
	connecting $(0,0)$ and $(1,1)$, hence the upper bound should hold.
	Taking derivative in $q$ shows that
	the left hand side decreases as $q$ increases and $x<1/2$.
	\begin{figure}
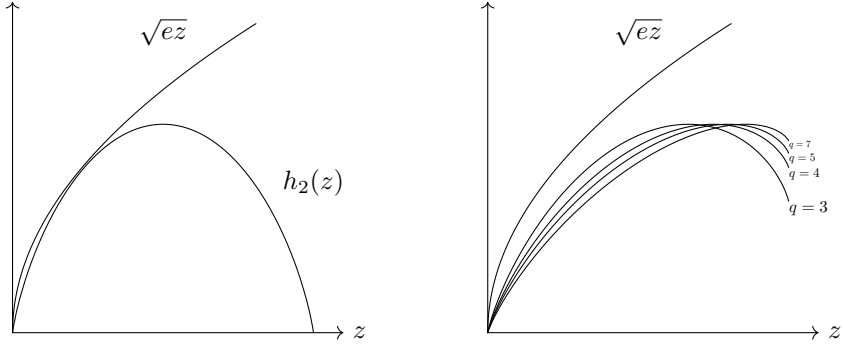
%\workhardtrue
		\edef\9{\ifworkhard99\else3\fi}
		\PMDF{h2}1{\PMP{5.54518-#1*ln(#1)-(4-#1)*ln(4-#1))}} % 8ln(2)
		$$\tikz{
			\draw[<->](0,4.4)--(0,0)--(4.4,0)node[right]{$z$};
			\draw plot[domain=0:1.8,samples=\9]
				(\x*\x,2.28561*\x)(2,4)node{$√{ez}$}; % 2ln(2)√e
			\draw(0,0)--plot[domain=.01:3.99,samples=\9]
				(\x,{h2(\x)})--(4,0)(4,2)node{$h_2(z)$};
		}\mkern72mu
		\tikz{
			\draw[<->](0,4.4)--(0,0)--(4.4,0)node[right]{$z$};
			\draw plot[domain=0:1.8,samples=\9]
				(\x*\x,2.28561*\x)(2,4)node{$√{ez}$}; % 2ln(2)√e
			\foreach\q in{3,4,5,7}{
				\PMS\lgq{log2(\q)}\PMS\lnqq{ln(\q-1)}\PMS\invq{1/log2(\q)}
				\PMDF{hq}1{\PMP{(5.54518-##1*ln(##1)+%
					##1*\lnqq-(4-##1)*ln(4-##1))/\lgq}}
				\draw(0,0)--plot[domain=.01:3.99,samples=\9]
					(\x,{hq(\x)})--(4,{4*\lnqq/\lgq})node
					[below right,inner sep=0]{\scalebox\invq{$q=\q$}};
			}
		}$$
		\caption{
			To the left:
			Binary entropy function $h_2(z)$ and an upper bound of $√{ez}$.
			To the right:
			$1-𝔻(z\|1-1/q)/\log q$ for $q=3,4,5,7$
			and an upper bound of $√{ez}$.
		} \label{fig:hqre}
	\end{figure}

\subsection{On the weight distribution of random linear codes}
	\label{sec:weight}

	This subsection contains the nontrivial part of the proof of \Cref{lem:LDP}.
	Fix any $i∈[ℓ]$.
	We want to prove that when $𝔾∈\GL(ℓ,q)$ is selected uniformly at random,
	the inequality
	\[*\Z(W\wg i)≤ℓ\exp(q\Z(W)ℓ)(q\Z(W))^{⌈i^2/3ℓ⌉} \repeattag{eq:LDPZ}\]*
	holds with probability $1-3q^{-√ℓ/13}$.
	In bounding the left hand side, the fundamental theorem of polar coding%
	---$Z$-end reads $\Z(W\wg i)≤\fz{𝔾}\w i(\Z(W))$, where $\fz{𝔾}\w i$ is
	the weight enumerator of codewords of the form $0_1^{i-1}1u_{i+1}^ℓ𝔾$.
	Thus it remains to show the inequality with the left hand side replaced
	\[*\fz{𝔾}\w i(z)≤ℓ\exp(qzℓ)(qz)^{⌈i^2/3ℓ⌉}\]*
	where $z≔\Z(W)$ for short.
	This inequality is in fact a consequence of
	\[\fz{𝔾}\w i(z)≤ℓ(1+(q-1)z)^{ℓ-⌈i^2/3ℓ⌉}((q-1)z)^{⌈i^2/3ℓ⌉}
		\label{eq:enumerator}\]
	because $(1+a)^b<\exp(ab)$.
	We will show the last inequality.
	Now divide $i$ into two cases: $1≤i≤√{3ℓ}$ and $√{3ℓ}<i≤ℓ$.
	
	For $i=1,2…√{3ℓ}$, the exponent $⌈i^2/3ℓ⌉$ is simply $1$, so
	the inequality to be proven reads $\fz{𝔾}\w i(z)≤ℓ(1+(q-1)z)^{ℓ-1}(q-1)z$.
	The right hand side over counts all nonzero codewords by
	choosing a nonzero position ($ℓ$), assigning a nonzero symbol ($(q-1)z$),
	and filling in the rest of $ℓ-1$ blanks arbitrarily ($(1+(q-1)z)^{ℓ-1}$).
	On the left hand side, $\fz{𝔾}\w i$ enumerates only codewords of the form
	$0_1^{i-1}1_iu_{i+1}^ℓ𝔾$, which are all nonzero as $𝔾$ is invertible.
	Hence Inequality~\eqref{eq:enumerator} holds for
	$i≤√{3ℓ}$ and nonnegative $z$ regardless of what kernel $𝔾$ is in effect.
	
	For $i=√{3ℓ}+1,√{3ℓ}+2…ℓ$, let $k≔ℓ-i$ and let $d≔i^2/3ℓ$.
	% Garbage collection: k not used
	These variables resemble the dimension and the minimal distance
	of a linear block code as in the notation
	\emph{an $[ℓ,k,d]$-code} in classical (algebraic) coding theory.
	To make Inequality~\eqref{eq:enumerator} hold, we execute
	a two-phase procedure to avoid all codewords of weight less than $d$
	and to eliminate kernels with poor overall scores.
	In further detail, we will reject a kernel $𝔾$ if there exists $u_{i+1}^ℓ$
	such that $\wt(0_1^{i-1}1u_{i+1}^ℓ𝔾)<d$ and call it phase~I\@.
	Afterwards, among surviving kernels with only \emph{heavy} codewords,
	we will reject a kernel if its overall score $\fz{𝔾}\w i(z)$
	is too low and call it phase~II\@.
	The failing probability $3q^{-√ℓ/13}$ is the price we pay for rejecting.
	Up to this point, two things remain to be analyzed:
	how much probability we pay for rejecting \emph{light} codewords
	in phase~I (answer: $q^{-√ℓ/13}$),
	and what is the Markov cutoff that makes Inequality~\eqref{eq:enumerator}
	in phase~II (answer: $2q^{-√ℓ/13}$).
	
	Phase~I analysis is as follows:
	Fix $u_{i+1}^ℓ$ and vary $𝔾∈\GL(ℓ,q)$;
	the codeword $𝕏_1^ℓ≔0_1^{i-1}1_iu_{i+1}^ℓ𝔾$ is a nonzero vector
	distributed uniformly on $𝔽_q^ℓ、\{0_1^ℓ\}$.
	This distribution is almost identical to
	the uniform distribution on $𝔽_q^ℓ$.
	Assume $𝕏_1^ℓ$ follows the latter;
	this makes $𝕏_1^ℓ$ lighter,
	which is compatible with the direction of the inequalities we want.
	Then the probability that $𝕏_1^ℓ$ has weight less than $d$
	is the probability that $ℓ$ Bernoulli trials%
	---$𝕏_j$ being ``zero'' with probability $1/q$
	and ``nonzero'' with probability $(q-1)/q$---%
	result in less than $d$ ``nonzero''s.
	By the large deviations theory \cite[Exercise~2.2.23(b)]{DZ10},
	$\wt(𝕏_1^ℓ)<d$ holds with probability less than
	\[*\exp（-ℓ𝔻（÷d{ℓ}‖÷12））=2^{-ℓ(1-h_2(d/ℓ))}\]*
	for the $q=2$ case, where $𝔻$ is the Kullback--Leibler divergence.
	For general $q$, similarly, $\wt(𝕏_1^ℓ)<d$ holds with probability less than
	\[*\exp（-ℓ𝔻（÷d{ℓ}‖1-÷1q））
		=\exp（-d\log÷{d/ℓ}{1-1/q}-(ℓ-d)\log÷{1-d/ℓ}{1/q}）.\]*
	This quantity is less than $q^{-ℓ(1-h_2(d/ℓ))}$ by \Cref{fig:hqre}
	(meaning that $q=2$ is the most difficult case).
	By \Cref{fig:hqre}, $h_2(d/ℓ)<√{ed/ℓ}=√{ei^2/3ℓ^2}=(√{e/3})i/ℓ<0.952i/ℓ$.
	So the rejecting probability is less than
	$q^{-ℓ(1-h_2(d/ℓ))}<q^{-ℓ+0.952i}$.
	Take into account that there are $q^{ℓ-i}$ possibilities of $u_{i+1}^ℓ$.
	The union bound yields
	$q^{ℓ-i}q^{-ℓ+0.952i}=q^{-0.048i}<q^{-0.048√{3ℓ}}<q^{-√ℓ/13}$.
	Hence the rejecting probability $q^{-√ℓ/13}$.
	Phase~I ends here.
	
	Phase~II analysis is as follows:
	After we reject some $𝔾$ in phase~I, some codewords will disappear;
	particularly, this includes all codewords of low weights.
	Therefore,
	the expectation of $\fz{𝔾}\w i(z)$ is bounded by the weight enumerator
	of all heavy codewords rescaled by the number of codewords.
	In detail, start from
	\begin{align*}
		𝔼[\fz{𝔾}\w i(z)｜𝔾† survives phase I†]
		&=	𝔼[\fz{𝔾}\w i(z)𝕀\{𝔾† survives†\}]/ℙ\{𝔾† survives†\} \\
		&≤	𝔼[\fz{𝔾}\w i(z)𝕀\{𝔾† survives†\}]/(1-q^{-√ℓ/13}).
			\steplabel{eq:phasechange}
	\end{align*}
	$𝕀$ is the indicator function.
	In the denominator, $1-q^{-√ℓ/13}>1/4$ as $ℓ≥30$.
	Put that aside and redefine $d≔⌈i^2/3ℓ⌉$.
	The expected value part is bounded from above by
	\def\bn#1#2{\hbox{\Large$\tbinom{#1}{#2}$}}
	\begin{align*}
		&{}	𝔼[\fz{𝔾}\w i(z)𝕀\{𝔾† survives†\}]
		=	𝔼「∑_{u_{i+1}^ℓ}z^{\wt(u_{i+1}^ℓ𝔾)}𝕀\{𝔾† survives†\}」 \\
		&≤	𝔼「∑_{u_{i+1}^ℓ}z^{\wt(u_{i+1}^ℓ𝔾)}𝕀\{\wt(u_{i+1}^ℓ𝔾)≥d\}」 
		=	∑_{u_{i+1}^ℓ}𝔼[z^{\wt(u_{i+1}^ℓ𝔾)}𝕀\{\wt(u_{i+1}^ℓ𝔾)≥d\}] \\
		&≤	q^{ℓ-i}𝔼[z^{\wt(𝕏_1^ℓ)}𝕀\{\wt(𝕏_1^ℓ)≥d\}]
		=	q^{ℓ-i}q^{-ℓ}∑_{x_1^ℓ}z^{\wt(x_1^ℓ)}𝕀\{\wt(x_1^ℓ)≥d\} \\
		&=	q^{-i}∑_{w≥d}\bn{ℓ}wz^w(q-1)^w
		≤	q^{-i}∑_{w≥d}\bn{ℓ}d\bn{ℓ-d}{w-d}z^w(q-1)^w \\
		&=	q^{-i}\bn{ℓ}d∑_{w≥d}\bn{ℓ-d}{w-d}z^{w-d}(q-1)^{w-d}((q-1)z)^d. \\
		&=	q^{-i}\bn{ℓ}d(1+(q-1)z)^{ℓ-d}((q-1)z)^d
		\qquad †(overestimate the scalar $q^{-i}\bn{ℓ}d$)† \\
		&≤	(q^{-√ℓ/13}ℓ/2)(1+(q-1)z)^{ℓ-d}((q-1)z)^d.
	\end{align*}
	The first equality expands the definition.
	The next inequality replaces $𝔾$ surviving phase~I by a weaker condition.
	The next equality switches $𝔼$ and $∑$.
	The next inequality replaces
	the ensemble of $u_{i+1}^ℓ𝔾$ by a uniform $𝕏_1^ℓ∈𝔽_q^ℓ$.
	The next equality expands the definition of the expectation over $𝕏_1^ℓ$.
	The next equality counts codewords.
	The next inequality selects $w$ positions
	by first selecting $d$ and then selecting $w-d$.
	The next two equalities factor and apply the binomial theorem.
	The rest is by a series of inequalities that overestimate the scalar:
	$q^{-i}\binom{ℓ}d=q^{-i}\binom{ℓ}{⌈i^2/4ℓ⌉}<q^{-i}\binom{ℓ}{i^2/4ℓ}ℓ/2
		<q^{-i}2^{ℓh_2(i^2/4ℓ^2)}ℓ/2≤q^{-i+ℓh_2(i^2/4ℓ^2)}ℓ/2$.
	Similar to the end of phase~I, the exponent part is
	$-i+ℓh_2(i^2/3ℓ^2)<-i+ℓ√{ei^2/3ℓ^2}=-i+i√{e/3}<-0.048i<-0.048√{3ℓ}<-√ℓ/13$.
	Hence the scalar part is less than $q^{-√ℓ/13}ℓ/2$.
	Put $1-q^{-√ℓ/13}>1/4$ back to the denominator
	as in Inequality~\eqref{eq:phasechange};
	$𝔼[\fz{𝔾}\w i(z)｜𝔾† survives phase I†]$ has an upper bound of
	\[*2q^{-√ℓ/13}ℓ(1+(q-1)z)^{ℓ-d}((q-1)z)^d.\]*
	By Markov's inequality,
	Inequality~\eqref{eq:enumerator} holds with probability $1-2q^{-√ℓ/13}$,
	i.e., the rejecting probability is $2q^{-√ℓ/13}$.
	Phase~II ends here.
	The sum of the two rejecting probabilities is $3q^{-√ℓ/13}$
	as claimed in \Cref{lem:LDP}, hence the lemma settled.

\subsection{Bibliographic remarks}

	Concerning the fundamental theorems:
	Nonlinear $\g W$ is not taken into consideration for that
	it is hard to imagine how MacWilliams duality works then.
	Also the $S$-parameter does not generalize to non-field input alphabet.
	Concerning random linear codes:
	\cite[Section~II.C]{BF02} portrays a clear picture of
	the weight distribution of binary random linear codes.
	\Cref{sec:weight} accommodates and extends
	their argument to general prime power $q$.
	Concerning the LDP behavior:
	\cite[Theorem~22]{KSU10} showed that $π<1$ can be arbitrary close to $1$
	over binary alphabet utilizing the Bose–Chaudhuri–Hocquenghem codes.
	Our \Cref{lem:LDP} on the other hand,
	implies that almost all kernels make $π$ close to $1$.
	
	It remains to prove \Cref{lem:triforce,lem:CLT}.

\section{Local CLT Behavior (Proof of Lemma~\ref{lem:CLT})} \label{pf:CLT}

	We are to prove that the following inequality
	holds with high probability:
	\[*∑_{i=1}^ℓh_α(H(W\wg i))<4ℓ^{1/2+α}. \repeattag{eq:CLT}\]*
	The target inequality is the sum of the following three inequalities:
	\begin{align*}
		∑_{i=⌈H(W)ℓ+ℓ^{1/2+α}⌉+1}^ℓh_α(H(W\wg i))
			&< ℓ^{1/2+α}, \steplabel{eq:bob}\\
		∑_{i=⌊H(W)ℓ-ℓ^{1/2+α}⌋+1}^{⌈H(W)ℓ+ℓ^{1/2+α}⌉}h_α(H(W\wg i))
			&< 2ℓ^{1/2+α}, \rlap{\quad and} \\
		∑_{i=1}^{⌊H(W)ℓ-ℓ^{1/2+α}⌋}h_α(H(W\wg i))
			&< ℓ^{1/2+α}. \steplabel{eq:eve}
	\end{align*}
	The second one is trivial as $h_α(z)≤(1/2)^α$.
	The first one will be proven in \Cref{sec:Gallager}
	with failing probability $ℓ^{-\log(ℓ)/20}$.
	The third one will be proven in \Cref{sec:Hayashi}
	with failing probability $ℓ^{-\log(ℓ)/20}$.
	Before the main proofs,
	we devote \Cref{sec:symmetrize} to introduce the symmetrization trick,
	which will reduce our proof to the case of symmetric $q$-ary channels.
	A channel $W$ being symmetric means that for any affine shifting $ξ∈𝔽_q$,
	there exists an permutation $σ$ on $𝒴$ such that
	$W(y｜ξ+x)=W(σ(y)｜x)$ holds for all $x∈𝔽_q$ and $y∈𝒴$.
	It also means that the uniform input achieves the Shannon capacity.
	This justifies the usage of linear codes.
	In \Cref{sec:ChangSahai}, we invoke some \emph{universal bound}
	on entropies and exponents from Chang, Draper, and Sahai's works.
	Finally, we will be abusing the theory of random linear codes
	in \Cref{sec:Gallager} for noisy channel coding and
	in \Cref{sec:Hayashi} for secrecy over wiretap channels.

\subsection{Symmetrize channel and uniformize input} \label{sec:symmetrize}

	Let $W:𝔽_q→𝒴$ be any $q$-ary channel;
	let $X$ and $Y$ be some input and the corresponding output.
	Symmetrize the channel as follows:
	Let $Ξ∈𝔽_q$ be a uniform r.v.\ independent of $X,Y$.
	Let $¯W:𝔽_q×(𝔽_q×𝒴)→[0,1]$ be the probability mass function of
	this combination of r.v.s $(Ξ+X,(X,Y))∈𝔽_q×(𝔽_q×𝒴)$.
	This $¯W$ behaves like a channel such that, quote, unquote,
	$¯W((x,y)｜z)=W(x,y)/q$ for all inputs $z∈𝔽_q$ and outputs $(x,y)∈𝔽_q×𝒴$.
	Despite that this channel might be properly simulated by
	a symmetric channel with feedback to the sender,
	all that matters is that the biased input $X$
	is neutralized by the uniform r.v.\ $Ξ$, and becomes uniform.
	Let $\g W$ be the multiplication of an invertible matrix $G$ from the right.
	Let $¯W\w i(u_i,u_1^{i-1}x_1^ℓy_1^ℓ)$ be the probability mass function of
	the tuple $(U_i,U_1^{i-1}X_1^ℓY_1^ℓ)$, where $U_1^ℓG=Ξ_1^ℓ+X_1^ℓ$.
	This definition is compatible with the channel transformation
	of $¯W$ as if $¯W$ was an actual channel in the first place.
	Let $H(¯W\w i)$ be $H(U_i｜U_1^{i-1}X_1^ℓY_1^ℓ)$;
	this is also compatible.
	The following lemma justifies why $¯W$ is useful in theory.
	
	\begin{lem}[channel symmetrization]
		$¯W$ is a symmetric $q$-ary channel, $H(¯W)=H(W)$,
		and $H(¯W\w i)=H(W\w i)$ for all $i∈[ℓ]$.
	\end{lem}
	
	This lemma is by \cite[Definition~6 and Lemmas 7 and~8]{MT14},
	plus the arguments in between.
	See also \cite[Theorem~2]{HY13} where
	they cared about whether $Z(¯W\w i)=Z(W\w i)$.
	One could also expand all definitions to verify the identities.
	
	The consequence of this lemma is that
	$¯W$ behaves like a shadow copy of $W$, but is symmetric.
	All inequalities involving entropies of $W$ and $W\w i$
	are reduced to inequalities involving entropies of $¯W$ and $¯W\w i$.
	Subsequently, passing statements to $¯W$ is effectively assuming
	that the channel $W$ is symmetric with the uniform input to begin with.
	In the upcoming subsections, we will prove that the targeted inequalities,
	\eqref{eq:bob} and~\eqref{eq:eve}, hold for any symmetric $q$-ary channel
	$W$ with the uniform input with high probability.
	We conjecture that the symmetrization technique is optional as
	it seems like a wrapper of complicated Bayesian formulas.

\subsection{Chang--Sahai's universal quadratic bound} \label{sec:ChangSahai}

	This and the next two subsections contain
	the most convoluted part of the proof of \Cref{lem:CLT}.
	This subsection prepares a universal upper bound on
	Gallager's E-null function, which ultimately evolves into
	a universal lower bound on Gallager's error exponent.
	
	Let $W:𝒳→𝒴$ be a $q$-ary channel.
	Symmetry is not required in this subsection but it is in the next two.
	Assume the uniform input distribution $W\inp(x)=1/q$ for all $x∈𝒳$.
	Define Gallager's E-null function and
	its complement \cite[Formula~(1)]{CS07}:
	\begin{gather*}
		E_0(t)≔-\log∑_{y∈𝒴}（∑_{x∈𝒳}W\inp(x)W(y｜x)^{1/(1+t)}）^{1+t},
			\rlap{\quad and} \\
		¯E_0(t)≔\log∑_{y∈𝒴}（∑_{x∈𝒳}W(x,y)^{1/(1+t)}）^{1+t}.
	\end{gather*}
	By complement we mean that under the uniform input, $¯E_0(t)$ degenerates to
	\begin{multline*}
		¯E_0(t)
		=	\log∑_{y∈𝒴}（∑_{x∈𝒳}(q^{-1}W(y｜x))^{1/(1+t)}）^{1+t} \\
		=	t\log q-\log∑_{y∈𝒴}（∑_{x∈𝒳}(q^{-1}W(y｜x))^{1/(1+t)}）^{1+t}
		=	t\log q-E_0(t).
	\end{multline*}
	Equivalently, $E_0(t)+¯E_0(t)=t\log q$.
	For non-uniform inputs, $W\inp(x)$ does not penetrate the summations.
	The E-null function and its complement
	deeply associate to the following family of measures:
	For any $t∈[-2/5,1]$, define the $t$-tilted probability mass function
	$\Wt:𝒳×𝒴→[0,1]$ as in \cite[Definition~1]{CS07}
	\[*\Wt(x,y)≔
		÷{\(∑_{ξ∈𝒳}W(ξ,y)^{1/(1+t)}\)^{1+t}}
			{∑_{η∈𝒴}\(∑_{ξ∈𝒳}W(ξ,η)^{1/(1+t)}\)^{1+t}}
		\times÷{W(x,y)^{1/(1+t)}}{∑_{ξ∈𝒳}W(ξ,y)^{1/(1+t)}}.\]*
	Do not confuse $W^ℓ$ with $\Wt$,
	the latter $\tiltbox{2.2}W{21.9}$ is tilted.
	When $t=0$,
	the tilted falls back to its italic origin $\Wnot^0(x,y)=W(x,y)$.
	These measures can be interpreted as follows:
	$\Wt$ behaves like a channel with a dedicated input distribution.
	The first fraction in the definition
	specifies the output distribution $\Wtout(y)$.
	The second fraction specifies the
	a posteriori distribution $\Wt(x｜y)$ when $y$ is known.
	As $\Wt$ is not an actual channel, it is not meaningful to
	alter the input distribution and ask for the corresponding output.
	Like the symmetrization trick, all that matters is that
	we can compute entropies, and what not, as if they were real channels.
	Quantities we are interested in are listed below:
	Let $H_e$ be the base-$e$ entropy.
	Let $H_e(\Wt)$ be $H_e(\Xt｜\Yt)$ where
	$(\Xt,\Yt)$ is a tuple r.v.\ that follows $\Wt$.
	Let $H_e(\Xt↾y)$ be the entropy of
	the a posteriori distribution of $\Xt$ given $\Yt=y$;
	to be specific, $H_e(\Xt↾y)=∑_{x∈X}\Wt(x｜y)\log \Wt(x｜y)$.
	
	\cite[Formula (13) and~(19)]{CS07}
	have that the following hold for $t∈[0,1]$:
	\begin{align*}
		\dt¯E_0(t) &= H_e(\Wt),\rlap{\quad and} \\
		\dtdt¯E_0(t) &= \dt H_e(\Wt)
		=	÷1{1+t}∑_{y∈𝒴}\Wtout(y)（∑_{x∈𝒳}\Wt(x｜y)\log(\Wt(x｜y))^2）+{} \\
		&\mskip120mu +÷t{1+t}∑_{y∈𝒴}\Wtout(y)H_e(\Xt↾y)^2-H_e(\Wt)^2.
			\steplabel{eq:holomorphic}
	\end{align*}
	Careful readers may verify them by hand or follow
	\cite[Formulas (13) to~(19)]{CS07} and \cite[Lemmas 9 and~10]{DCS14}.
	Similar computations are also carried out by \cite{AW10,AW14}.
	Notice that $¯E_0(t)$, $H_e(\Wt)$, and every other term
	in Equation~\eqref{eq:holomorphic} are all holomorphic functions in $t$
	on the half-plane $\operatorname{Re}t>-1$
	(there is a singularity at $1/(1+t)=∞$).
	By the identity theorem in complex analysis \cite[Corollary~8.16]{BMPS02},
	Equation~\eqref{eq:holomorphic} holds for all $t∈[-2/5,1]$.
	Dropping the nonpositive square,
	we deduce an upper bound for each $t∈[-2/5,1]$:
	\begin{align*}
		\dtdt ¯E_0(t)
		&≤	÷1{1+t}∑_{y∈𝒴}\Wtout(y)（∑_{x∈𝒳}\Wt(x｜y)\log(\Wt(x｜y))^2）+{} \\
		&\qquad	+÷{\max(0,t)}{1+t}∑_{y∈𝒴}\Wtout(y)H_e(\Xt↾y)^2.
			\steplabel{eq:combination}
	\end{align*}
	This upper bound on $¯E_0''(t)$ is a linear combination of
	\[*∑_{x∈𝒳}\Wt(x｜y)\log(\Wt(x｜y))^2\qquad†and†\qquad H_e(\Xt↾y)^2\]*
	parametrized by $y∈𝒴$, so it remains to bound them separately.
	For the second kind of constituents,
	the entropy cannot exceed $\log q$ so $H_e(\Xt↾y)^2≤\log(q)^2$.
	For the first kind of constituents,
	the following lemma adapted from \cite[Lemma~1]{CS07} helps.
	
	\begin{lem}[second moment] \label{lem:second}
		If $w_1,w_2,…w_q$ are positive numbers of sum $1$, then
		\[*∑_iw_i\log(w_i)^2≤\begin{Bmatrix}
			\log(q)^2	&	†for †q≥3 \\
			0.563		&	†for †q=2
		\end{Bmatrix}≤1.2\log(q)^2.\]*
	\end{lem}
	
	With the lemma,
	we do have $∑_{x∈𝒳}\Wt(x｜y)(\log \Wt(x｜y))^2≤1.2\log(q)^2$.
	Now Inequality~\eqref{eq:combination} becomes
	\begin{align*}
		\dtdt¯E_0(t)
		&≤	÷1{1+t}∑_{y∈𝒴}\Wtout(y)·1.2\log(q)^2
			+÷{\max(0,t)}{1+t}∑_{y∈𝒴}\Wtout(y)\log(q)^2 \\
		&≤	÷1{1+t}·1.2\log(q)^2+÷{\max(0,t)}{1+t}\log(q)^2 ≤ 2\log(q)^2
	\end{align*}
	for all $t∈[-2/5,1]$.
	Since $E_0(t)$ is a linear function $t\log q$ minus $¯E_0(t)$,
	their first derivatives sum to $\log q$ while
	their second derivatives are opposite.
	Hence the following lemma.
	\begin{lem}[universal quadratic bound] \label{lem:quadratic}
		\cite[Theorem~2]{CS07}. Cf.\ \cite[Theorem~5.6.3]{Gallager65}.
		Let $W$ be a $q$-ary channel.
		Assume the uniform input distribution.
		Then Gallager's E-null function satisfies
		\begin{align*}
			E_0(0) &= 0, \\
			E_0'(0) &= I(W)\log q,\rlap{\quad and} \\
			E_0''(t) &≥ -2\log(q)^2
		\end{align*}
		for all $t∈[-2/5,1]$.
		In particular, it satisfies
		\[*E_0(t)≥I(W)t\log q-t^2\log(q)^2.\]*
	\end{lem}

\subsection{Gallager's argument at Bob's end} \label{sec:Gallager}

	This subsection take advantage of the universal bound developed
	three lines ago and starts actually proving \Cref{lem:CLT}.
	This subsection deals with
	\[*∑_{i=⌈H(W)ℓ+ℓ^{1/2+α}⌉+1}^ℓh_α(H(W\wg i))<ℓ^{-1/2+α}
		\repeattag{eq:bob}\]*
	by passing it to an inequality that captures
	the performance of noisy channel coding.
	Owing to $h_α$'s concavity,
	the left hand side of Inequality~\eqref{eq:bob} is
	\[*∑_{i=j+1}^ℓh_α(H(W\wg i))
		≤(ℓ-j)h_α（÷1{ℓ-j}∑_{i=j+1}^ℓH(W\wg i)）\]*
	where $j≔⌈H(W)ℓ+ℓ^{1/2+α}⌉$ for short.
	It suffices to prove that the right hand side is less than $ℓ^{-1/2+α}$.
	In the spirit of the motivational Chain Rule~\eqref{eq:chain},
	the sum of the chain of $H(W\wg i)$
	on the right hand side is $H(U_{j+1}^ℓ｜U_1^jY_1^ℓ)$.
	In order to prove Inequality~\eqref{eq:bob}, we will show
	\[(ℓ-j)h_α（÷1{ℓ-j}H(U_{j+1}^ℓ｜U_1^jY_1^ℓ)）<ℓ^{-1/2+α}.
		\label{eq:bobblock}\]
	But what is $H(U_{j+1}^ℓ｜U_1^jY_1^ℓ)$?
	It measures the equivocation at Bob's end when $U_1^j$ is known to Bob.
	In other words, we may as well pretend that
	there is a random rectangular full-rank matrix $𝔾'$ with
	$ℓ$ columns and only $k≔ℓ-j=⌊ℓ-H(W)ℓ-ℓ^{1/2+α}⌋$ rows,
	that Alice computes and sends $X_1^ℓ≔U_1^k𝔾'$ to Bob, and that Bob
	attempts to decode $ˆU_1^k$ upon receiving $Y_1^ℓ$ using the MAP decoder.
	The equivocation is thus, by Fano's inequality,
	bounded in terms of the probability that Bob fails to decode $U_1^k$:
	\begin{multline}
		H(U_{j+1}^ℓ｜U_1^jY_1^ℓ)
		≤	-\P\log_q\P-(1-\P)\log_q(1-\P)+\P\log_q(q^k-1) \\
		≤	-\P\log_q\P+÷\P{\log q}+\P·k = \P·（÷{1-\log\P}{\log q}+k）.
			\label{eq:Fano}
	\end{multline}
	Here $\P$ is the probability that Bob fails to decode, $ˆU_1^k≠U_1^k$.
	
	The following is how to compute Bob's decoder block error probability.
	The generator matrix $𝔾'$ Alice uses is selected uniformly from
	the ensemble of full-rank $k$-by-$ℓ$ matrices.
	The difference of every pair of codewords
	distributes uniformly on $𝔽_q^ℓ、\{0_1^ℓ\}$.
	Over symmetric channels, the difference alone determines the output's joint
	distribution because $W^ℓ(y_1^ℓ｜ξ_1^ℓ+x_1^ℓ)=W^ℓ(σ_1^ℓ(y_1^ℓ)｜x_1^ℓ)$
	for some component-wise permutation $σ_1^ℓ$ on $𝒴^ℓ$ depending on $ξ_1^ℓ$.
	Gallager's bound applies.
	To elaborate, let $t∈[0,1]$.
	Bob's average error probability satisfies
	\cite[Inequalities (5.6.2) to~(5.6.14)]{Gallager68}
	\begin{align*}
		&{}	𝔼P\{†Bob fails to decode †U_1^k† given †𝔾'\} \\
		&=	𝔼∑_{u_1^k}q^{-k}∑_{y_1^ℓ}W^ℓ(y_1^ℓ｜u_1^k𝔾')
			P\{†Bob has †ˆU_1^k≠u_1^k† given †𝔾',u_1^k,y_1^ℓ\} \\
		&=	𝔼∑_{y_1^ℓ}W^ℓ(y_1^ℓ｜0_1^ℓ)
			P\{†Bob has †ˆU_1^k≠0_1^k† given †𝔾',0_1^k,y_1^ℓ\} \\
		&≤	𝔼∑_{y_1^ℓ}W^ℓ(y_1^ℓ｜0_1^ℓ)（∑_{v_1^k≠0_1^k}
			P\{{†Bob prefers †v_1^k† over †0_1^k† given †𝔾'}\}）^t \\
		&≤	𝔼∑_{y_1^ℓ}W^ℓ(y_1^ℓ｜0_1^ℓ)（∑_{v_1^k≠0_1^k}
			÷{W^ℓ(y_1^ℓ｜v_1^k𝔾')^{1/(1+t)}}{W^ℓ(y_1^ℓ｜0_1^ℓ)^{1/(1+t)}}）^t \\
		&=	𝔼∑_{y_1^ℓ}W^ℓ(y_1^ℓ｜0_1^ℓ)^{1/(1+t)}
			（∑_{v_1^k≠0_1^k}W^ℓ(y_1^ℓ｜v_1^k𝔾')^{1/(1+t)}）^t \\
		&≤	∑_{y_1^ℓ}W^ℓ(y_1^ℓ｜0_1^ℓ)^{1/(1+t)}
			（𝔼∑_{v_1^k≠0_1^k}W^ℓ(y_1^ℓ｜v_1^k𝔾')^{1/(1+t)}）^t \\
		&=	∑_{y_1^ℓ}W^ℓ(y_1^ℓ｜0_1^ℓ)^{1/(1+t)}（∑_{x_1^ℓ≠0_1^ℓ}
			÷{q^k-1}{q^ℓ-1}·W^ℓ(y_1^ℓ｜x_1^ℓ)^{1/(1+t)}）^t \\
		&≤	q^{-kt}∑_{y_1^ℓ}W^ℓ(y_1^ℓ｜0_1^ℓ)^{1/(1+t)}
			（∑_{x_1^ℓ≠0_1^ℓ}q^{-ℓ}W^ℓ(y_1^ℓ｜x_1^ℓ)^{1/(1+t)}）^t \\
		&≤	q^{-kt}∑_{y_1^ℓ}W^ℓ(y_1^ℓ｜0_1^ℓ)^{1/(1+t)}
			（∑_{x_1^ℓ}q^{-ℓ}W^ℓ(y_1^ℓ｜x_1^ℓ)^{1/(1+t)}）^t \\
		&=	q^{-kt}∑_{y_1^ℓ}（∑_{x_1^ℓ}q^{-ℓ}W^ℓ(y_1^ℓ｜x_1^ℓ)^{1/(1+t)}）
			（∑_{x_1^ℓ}q^{-ℓ}W^ℓ(y_1^ℓ｜x_1^ℓ)^{1/(1+t)}）^t \\ %% nontrivial
		&=	q^{-kt}∑_{y_1^ℓ}
			（∑_{x_1^ℓ}q^{-ℓ}W^ℓ(y_1^ℓ｜x_1^ℓ)^{1/(1+t)}）^{1+t} \\
		&=	q^{-kt}∑_{y_1^ℓ}
			（∑_{x_1^ℓ}W^ℓ\inp(x_1^ℓ)W^ℓ(y_1^ℓ｜x_1^ℓ)^{1/(1+t)}）^{1+t} \\
		&=	\exp(kt\log q-(\emph{the E-null function of }W^ℓ)(t)) \\
		&=	\exp(kt\log q-ℓE_0(t)).
	\end{align*}
	In summary, $𝔼P\{†Bob fails to decode †U_1^k† given †𝔾'\}
		≤\exp(kt\log q-ℓE_0(t))$ whenever $0≤t≤1$.
	Recall the universal quadratic bound
	$E_0(t)≥I(W)t\log q-t^2\log(q)^2$ derived in \Cref{lem:quadratic}.
	We obtain that the exponent is
	\begin{align*}
		&{}	kt\log q-ℓE_0(t) = (ℓ-H(W)ℓ-ℓ^{1/2+α})t\log q-ℓE_0(t) \\
		&≤	(ℓ-H(W)ℓ-ℓ^{1/2+α})t\log q-ℓ(I(W)t\log q-t^2\log(q)^2) \\
		&=	(ℓt\log q-ℓ^{1/2+α})t\log q\qquad
			†(redeem the supremum at $t=ℓ^{-1/2+α}/2\log q$)† \\
		&↦	(ℓℓ^{-1/2+α}/2-ℓ^{1/2+α})ℓ^{-1/2+α}/2
		=	-ℓ^{2α}/4 = -ℓ^{2\log(\logℓ)/\logℓ}/4 = -\log(ℓ)^2/4.
	\end{align*}
	So far we obtain that the average error probability is less than
	$\exp(-\log(ℓ)^2/4)=ℓ^{-\log(ℓ)/4}$.
	Run Markov's inequality with cutoff $ℓ^{-\log(ℓ)/20}$.
	That is, we reject kernels such that
	$P\{†Bob fails to decode †U_1^k† given †𝔾'\}≥ℓ^{-\log(ℓ)/5}$.
	Then the rejecting probability is $ℓ^{-\log(ℓ)/20}$ because $1/20+1/5=1/4$.
	
	An upper bound on Bob's error probability being $\P<ℓ^{-\log(ℓ)/5}$,
	an upper bound on Bob's equivocation is
	\[*ℓ^{-\log(ℓ)/5}（÷{1-\logℓ^{-\log(ℓ)/5}}{\log q}+k）
		=ℓ^{-\log(ℓ)/5}（÷{1+\log(ℓ)^2/5}{\log q}+k）\]*
	by Inequality~\eqref{eq:Fano}.
	Plugging the latter in $kh_α(†here†/k)$, we derive that
	the left hand side of Inequality~\eqref{eq:bobblock} is less than
	\begin{align*}
		&{}	kh_α（÷{ℓ^{-\log(ℓ)/5}}k（÷{1+\log(ℓ)^2/5}{\log q}+k））
		=	k（ℓ^{-\log(ℓ)/5}（÷{1+\log(ℓ)^2/5}{k\log q}+1））^α \\
		&=	ℓ^{-α\log(ℓ)/5}k（÷{1+\log(ℓ)^2/5}{k\log q}+1）^α
		<	ℓ^{-α\log(ℓ)/5}ℓ（÷{1+\log(ℓ)^2/5}{ℓ\log q}+1）^α \\
		&<	ℓ^{-α\log(ℓ)/5}·ℓ·2^α = 2^αℓ\log(ℓ)^{-\log(ℓ)/5}.
	\end{align*}
	The first inequality uses that the left hand side
	increases monotonically in $k$ and $k≔ℓ-j=⌊ℓ-H(W)ℓ-ℓ^{1/2+α}⌋<ℓ$.
	The second inequality uses the assumption $ℓ≥2$.
	In any regard, the quantity at the end of the inequalities decays to $0$
	as $ℓ→∞$, so eventually it becomes less than $ℓ^{1/2+α}$,
	the right hand side of Inequality~\eqref{eq:bobblock}.
	This proves that Inequality~\eqref{eq:bob} holds with
	failing probability $ℓ^{-\log(ℓ)/20}$ as soon as $ℓ$ is large enough. % 20
	The lower bound on $ℓ$ in the statement of \Cref{lem:CLT} is large enough,
	hence the first half of \Cref{lem:CLT} settled.

\subsection{Hayashi's argument at Eve's end} \label{sec:Hayashi}

	This subsection contains the very last ingredient
	of the proof of \Cref{lem:CLT}.
	We dealt with Inequality \eqref{eq:bob} in the last subsection.
	We now deal with
	\[*∑_{i=1}^{⌊H(W)ℓ-ℓ^{1/2+α}⌋}h_α(H(W\wg i))<ℓ^{1/2+α}.
		\repeattag{eq:eve}\]*
	Similar to how we motivated Inequality~\eqref{eq:bobblock},
	we apply Jensen's inequality and the chain rule of conditional entropy
	to simplify Inequality~\eqref{eq:eve}.
	The left hand side becomes $jh_α(H(U_1^j｜Y_1^ℓ)/j)$
	where $j≔⌊H(W)ℓ-ℓ^{1/2+α}⌋$ for short.
	(This is not the same $j$ as in the last subsection.)
	The input uniform, the argument of $h_α$ is
	$H(U_1^j｜Y_1^ℓ)/j=1-I(U_1^j；Y_1^ℓ)/j$, which can be replaced by
	$I(U_1^j｜Y_1^ℓ)/j$ thanks to the evenness $h_α(1-z)=h_α(z)$.
	We will show
	\[jh_α（÷1jI(U_1^j；Y_1^ℓ)）<ℓ^{1/2+α}. \label{eq:eveblock}\]
	But what is $I(U_1^j；Y_1^ℓ)$?
	It is the amount of information Eve learns from wiretapping $Y_1^ℓ$
	if they know that $U_{j+1}^ℓ$ are junk.
	In other words, we may pretend that
	Alice transmits $X_1^ℓ≔U_1^jV_{j+1}^ℓ𝔾$ with
	confidential bits $U_1^j$ and obfuscating bits $V_{j+1}^ℓ$,
	Bob receives $X_1^ℓ$ in full, and Eve learns $Y_1^ℓ$.
	This context falls back to (a special case of) the traditional setup
	of wiretap channels \cite{Wyner75} where various bounds are studied,
	some in terms of Gallager's E-null function.
	
	Here are some preliminaries
	to control the information leaked to Eve.
	We follow the blueprint of how Hayashi derived
	the secrecy exponent in \cite[Inequality~(21)]{Hayashi06}.
	Consider the communication protocol depicted in \Cref{fig:cryptoalice}:
	Karl fixes a kernel $𝔾∈\GL(ℓ,q)$ and everyone knows $𝔾$.
	Alice chooses the confidential message $U_1^ℓ$.
	Vincent chooses the obfuscating bits $V_{j+1}^ℓ$.
	Charlie generates $Y_1^ℓ$ by
	plugging $X_1^ℓ≔U_1^jV_{j+1}^ℓ𝔾$ into a simulator of $W^ℓ$.
	Eve learns $Y_1^ℓ$ and is interested in knowing $U_1^j$ alone.
	So the channel on topic is the composition of Vincent and Charlie.
	Notation:
	Running out of symbols, we all use $ℙ$ with proper subscriptions
	to indicate the corresponding probability measures.
	That said, indexes in the subscription will be omitted.
	As Eve is interested in the relation between $U_1^j$ and $Y_1^ℓ$,
	let $Y_1^ℓ↾Gu_1^j$ be the r.v.\ that follows
	the a posteriori distribution of $Y_1^ℓ$ given $𝔾=G$ and $U_1^j=u_1^j$.
	More formally, $ℙ_{Y↾Gu}(y_1^ℓ)
		=ℙ_{Y|𝔾U}(y_1^ℓ｜G,u_1^ℓ)=ℙ_{𝔾  UY}(G,u_1^j,y_1^ℓ)/ℙ_{𝔾U}(G,u_1^j)$.
	We could have defined $Y_1^ℓ↾G$ to be the
	a posteriori distribution of $Y_1^ℓ$ given $𝔾=G$;
	but it is simply the same distribution as $Y_1^ℓ$ since $U_1^jV_{j+1}^ℓG$
	traverses all inputs uniformly regardless of the choice of $G$.
	That is, $ℙ_{Y|𝔾}(y_1^ℓ｜G)=ℙ_Y(y_1^ℓ)$ for all $y_1^ℓ∈𝒴^ℓ$.
	
	\begin{figure}
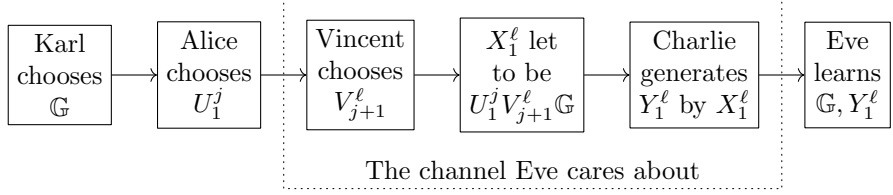
%\workhardtrue
		$$\tikz{
			\draw[nodes={right,draw,align=center}]
				(.6,0)coordinate(X)
				            node(K){Karl\\chooses\\$𝔾$}
				(K.east)+(X)node(A){Alice\\chooses\\$U_1^j$}
				(A.east)+(X)node(V){Vincent\\chooses\\$V_{j+1}^ℓ$}
				(V.east)+(X)node(N)
							{$X_1^ℓ$ let\\to be\\$U_1^jV_{j+1}^ℓ𝔾$}
				(N.east)+(X)node(C){Charlie\\generates\\$Y_1^ℓ$ by $X_1^ℓ$}
				(C.east)+(X)node(E){Eve\\learns\\$𝔾,Y_1^ℓ$}
			;
			\graph[use existing nodes]{
				K -> A -> V -> N -> C -> E
			};
			\draw[dotted]
				(V.west|-C.north)+(-.3,.3)-|($(C.south east)+(.3,-.8)$)-|
				node[pos=.25,above]{The channel Eve cares about}cycle
			;
		}$$
		\caption{
			A finer setup for Hayashi's secrecy exponent.
			Charlie generates $Y_1^ℓ$ such that
			$X_1^ℓ≔U_1^jV_{j+1}^ℓ𝔾$ and $Y_1^ℓ$ follow $W^ℓ$.
			Despite of the seemingly sequential structure,
			Karl, Alice, and Vincent work independently.
		} \label{fig:cryptoalice}
	\end{figure}
	
	Fix $G$ as an instance of $𝔾$.
	Let $I_e$ be the base-$e$ mutual information.
	The channel Eve cares about leaks information of this amount:
	\begin{align*}
		&{}	I_e(U_1^j；Y_1^ℓ｜G)
		=	∑_{u_1^jy_1^ℓ}ℙ_{UY|𝔾}(u_1^j,y_1^ℓ｜G)
			\log÷{ℙ_{Y|𝔾U}(y_1^ℓ｜G,u_1^j)}{ℙ_{Y|𝔾}(y_1^ℓ｜G)} \\
		&=	∑_{u_1^j}ℙ_U(u_1^j)∑_{y_1^ℓ}ℙ_{Y|𝔾U}(y_1^ℓ｜G,u_1^j)
			\log÷{ℙ_{Y|𝔾U}(y_1^ℓ｜G,u_1^j)}{ℙ_{Y|𝔾}(y_1^ℓ｜G)} \\
		&=	∑_{u_1^j}ℙ_U(u_1^j)∑_{y_1^ℓ}ℙ_{Y↾Gu}(y_1^ℓ)
			\log÷{ℙ_{Y↾Gu}(y_1^ℓ)}{ℙ_Y(y_1^ℓ)}
		=	∑_{u_1^j}ℙ_U(u_1^j)𝔻(Y_1^ℓ↾Gu_1^j\|Y_1). \steplabel{eq:fixkld}
	\end{align*}
	$𝔻(Y_1^ℓ↾Gu_1^j\|Y_1^ℓ)$ is the Kullback--Leibler divergence
	from the a posteriori distribution of $Y_1^ℓ$ given $G,u_1^j$
	to the coarsest distribution $Y_1^ℓ$.
	We are to take expectation over $𝔾$ to find the average information leak
	since we are interested in Markov's inequality.
	Equality~\eqref{eq:fixkld} yields
	\[𝔼I_e(U_1^j；Y_1^ℓ｜𝔾)=∑_Gℙ_𝔾(G)∑_{u_1^j}ℙ_U(u_1^j)
		𝔻(Y_1^ℓ↾Gu_1^j\|Y_1^ℓ). \label{eq:expkld}\]
	We now discover that there are redundancies
	in traversing all $G$ and $u_1^ℓ$:
	After all, $X_1^j$ is $u_1^jV_{j+1}^ℓG=u_1^j0_{j+1}^ℓG+0_1^jV_{j+1}^ℓG$,
	which is a fixed linear combination of the first $j$ rows
	plus a random vector from the span of the bottom $ℓ-j$ rows.
	When $V_1^ℓ$ varies, the track of $X_1^ℓ$ forms an affine subspace of
	$𝔽_q^ℓ$, a \emph{coset code} as in the context of the fundamental theorems.
	So what matters is the distribution of this coset code.
	In this regard, we replace the uniform ensemble of $(𝔾,U_1^j)$ by
	the uniform ensemble of $𝕂$ a rank-$(ℓ-j)$ affine subspace of $𝔽_q^ℓ$,
	where $j≔⌊H(W)ℓ-ℓ^{1/2+α}⌋$.
	Karl and Alice together choose $𝕂$ uniformly.
	Vincent chooses $X_1^ℓ∈𝕂$ uniformly.
	Charlie generates $Y_1^ℓ$ by throwing $X_1^ℓ$ into a simulator of $W^ℓ$.
	See \Cref{fig:cryptokarl} for the depiction of the new scheme.
	Hence Equality~\eqref{eq:expkld} becomes
	\[*𝔼I_e(U_1^j；Y_1^ℓ｜𝔾)=∑_Kℙ_𝕂(K)𝔻(Y_1^ℓ↾K\|Y_1^ℓ)\]*
	where $Y_1^ℓ↾K$ is the a posteriori distribution of $Y_1^ℓ$ given $𝕂=K$.
	Suddenly, the quantity $𝔼I_e(U_1^j；Y_1^ℓ｜𝔾)$ we are interested in
	turns into the mutual information $I_e(𝕂；Y_1^ℓ)$ between $𝕂$ and $Y_1^ℓ$
	as $𝕂$ replaces the role of $U_1^j$ in Formula~\eqref{eq:fixkld}.
	Recall that in \Cref{lem:quadratic} the mutual information
	is the derivative of Gallager's E-null function.
	We exploit this.
	Define the double-stroke E-null function for $(𝕂,Y_1^ℓ)$ as follows
	\[*𝔼_0(t)≔-\log∑_{y_1^ℓ}
		（∑_Kℙ_𝕂(K)ℙ_{Y|𝕂}(y_1^ℓ｜K)^{1/(1+t)}）^{1+t}.\]*
	Then $𝔼_0'(0)=I_e(𝕂；Y_1^ℓ)=𝔼I_e(U_1^j；Y_1^ℓ｜𝔾)$.
	Owing to the concavity of the E-null function,
	$𝔼_0'(0)≤𝔼_0(t)/t$ whenever $-2/5≤t<0$.
	Recap:
	To bound the average leaked information $𝔼I_e(U_1^j；Y_1^ℓ｜𝔾)$
	it suffices to bound $I_e(𝕂；Y_1^ℓ)$, which is then morphing
	to bounding $𝔼_0'(0)$ from above and to bounding $𝔼_0(t)$ from below.
	
	\begin{figure}
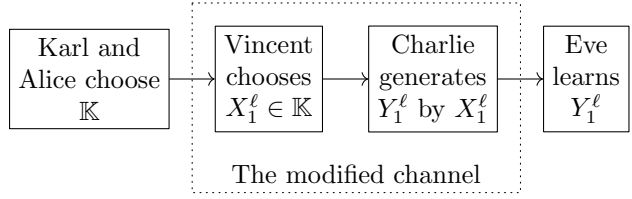
%\workhardtrue
		$$\tikz{
			\draw[nodes={right,draw,align=center}]
				(.6,0)coordinate(X)
				            node(K){Karl and \\Alice choose\\$𝕂$}
				(K.east)+(X)node(V){Vincent\\chooses\\$X_1^ℓ∈𝕂$}
				(V.east)+(X)node(C){Charlie\\generates\\$Y_1^ℓ$ by $X_1^ℓ$}
				(C.east)+(X)node(E){Eve\\learns\\$Y_1^ℓ$}
			;
			\graph[use existing nodes]{
				K -> V -> C -> E
			};
			\draw[dotted]
				(V.west|-C.north)+(-.3,.3)-|($(C.south east)+(.3,-.8)$)-|
				node[pos=.25,above]{The modified channel}cycle
			;
		}$$
		\caption{
			A simplified setup for Hayashi's secrecy exponent.
			Charlie generates $Y_1^ℓ$ such that
			$X_1^ℓ$ and $Y_1^ℓ$ follow $W^ℓ$.
		} \label{fig:cryptokarl}
	\end{figure}
	
	The double-stroke E-null function is bounded as below.
	Assume $-2/5≤t<0$.
	Let $s$ be $-t/(1+t)$; so $0<s≤2/3$ and $(1+s)(1+t)=1$.
	For any fixed $K$ and fixed $x_1^ℓ∈K$, the base of the $(1+t)$-th root
	in the definition of the double-stroke E-null function is
	\begin{align*}
		ℙ_{Y|𝕂}(y_1^ℓ｜K)
		&=	∑_{ξ_1^ℓ∈K}ℙ_{X|𝕂}(ξ_1^ℓ｜K)ℙ_{Y|X}(y_1^ℓ｜ξ_1^ℓ)
		=	∑_{ξ_1^ℓ∈K}q^jℙ_X(ξ_1^ℓ)ℙ_{Y|X}(y_1^ℓ｜ξ_1^ℓ) \\
		&=	∑_{ξ_1^ℓ∈K}q^jℙ_{XY}(ξ_1^ℓ,y_1^ℓ)
		=	q^jℙ_{XY}(x_1^ℓ,y_1^ℓ)+q^j∑_{x_1^ℓ≠ξ_1^ℓ∈K}ℙ_{XY}(ξ_1^ℓ,y_1^ℓ) \\
		&=	q^jℙ_{XY}(x_1^ℓ,y_1^ℓ)+q^jℙ_{XY}(K＼x_1^ℓ,y_1^ℓ).
	\end{align*}
	Here $ℙ_{XY}(K＼x_1^ℓ,y_1^ℓ)$ is a temporary shorthand for the summation
	of $ℙ_{XY}(ξ_1^ℓ,y_1^ℓ)$ over $ξ_1^ℓ∈K$ that excludes $x_1^ℓ$.
	Raise it to the power of $1/(1+t)=(1+s)$:
	\begin{align*}
		&{}	ℙ_{Y|𝕂}(y_1^ℓ｜K)^{1+s}
		=	ℙ_{Y|𝕂}(y_1^ℓ｜K)ℙ_{Y|𝕂}(y_1^ℓ｜K)^s
		=	∑_{x_1^ℓ∈K}q^jℙ_{XY}(x_1^ℓ,y_1^ℓ)ℙ_{Y|𝕂}(y_1^ℓ｜K)^s \\
		&=	∑_{x_1^ℓ∈K}q^jℙ_{XY}(x_1^ℓ,y_1^ℓ)
			（q^jℙ_{XY}(x_1^ℓ,y_1^ℓ)+q^jℙ_{XY}(K＼x_1^ℓ,y_1^ℓ)）^s \\
		&≤	∑_{x_1^ℓ∈K}q^jℙ_{XY}(x_1^ℓ,y_1^ℓ)
			（q^{js}ℙ_{XY}(x_1^ℓ,y_1^ℓ)^s+q^{js}ℙ_{XY}(K＼x_1^ℓ,y_1^ℓ)^s） \\
		&=	q^{j+js}∑_{x_1^ℓ∈K}ℙ_{XY}(x_1^ℓ,y_1^ℓ)ℙ_{XY}(x_1^ℓ,y_1^ℓ)^s
			+q^{j+js}∑_{x_1^ℓ∈K}ℙ_{XY}(x_1^ℓ,y_1^ℓ)ℙ_{XY}(K＼x_1^ℓ,y_1^ℓ)^s \\
		&=	q^{j+s}∑_{x_1^ℓ∈K}ℙ_{XY}(x_1^ℓ,y_1^ℓ)^{1+s}
			+q^{j+js}∑_{x_1^ℓ∈K}ℙ_{XY}(x_1^ℓ,y_1^ℓ)ℙ_{XY}(K＼x_1^ℓ,y_1^ℓ)^s.
	\end{align*}
	The inequality uses that the $s$-th power is sub-additive.
	Then the inner sum of the E-null function morphs as follows
	\begin{align*}
		&	∑_Kℙ_𝕂(K)ℙ_{Y|𝕂}(y_1^ℓ｜K)^{1+s} \\
		&≤	∑_Kℙ_𝕂(K)（q^{j+js}∑_{\!\!x_1^ℓ∈K\!\!}ℙ_{XY}(x_1^ℓ,y_1^ℓ)^{1+s}
			+q^{j+js}∑_{\!\!x_1^ℓ∈K\!\!}
				ℙ_{XY}(x_1^ℓ,y_1^ℓ)ℙ_{XY}(K＼x_1^ℓ,y_1^ℓ)^s） \\
		&=	q^{j+js}∑_Kℙ_𝕂(K)∑_{x_1^ℓ∈K}
			ℙ_{XY}(x_1^ℓ,y_1^ℓ)^{1+s}+{} \tag{major arc}\\
		&\qquad	+q^{j+js}∑_Kℙ_𝕂(K)∑_{x_1^ℓ∈K}
			ℙ_{XY}(x_1^ℓ,y_1^ℓ)ℙ_{XY}(K＼x_1^ℓ,y_1^ℓ)^s. \tag{minor arc}
	\end{align*}
	Divide and conquer---the inner sum of
	the double-stroke E-null function is split into two arcs as shown.
	The major arc is exactly
	\begin{multline*}
		q^{j+js}∑_Kℙ_𝕂(K)∑_{x_1^ℓ∈K}ℙ_{XY}(x_1^ℓ,y_1^ℓ)^{1+s}
		=	q^{j+js}q^{-j}∑_{x_1^ℓ∈𝔽_q^ℓ}ℙ_{XY}(x_1^ℓ,y_1^ℓ)^{1+s} \\
		=	q^{js}∑_{x_1^ℓ∈𝔽_q^ℓ}ℙ_X(x_1^ℓ)^{1+s}ℙ_{Y|X}(y_1^ℓ｜x_1^ℓ)^{1+s} 
		=	q^{js-ℓs}∑_{x_1^ℓ∈𝔽_q^ℓ}ℙ_X(x_1^ℓ)ℙ_{Y|X}(y_1^ℓ｜x_1^ℓ)^{1+s}.
	\end{multline*}
	The minor arc is loosen to
	\begin{align*}
		&{}	q^{j+js}∑_Kℙ_𝕂(K)∑_{x_1^ℓ∈K}
			ℙ_{XY}(x_1^ℓ,y_1^ℓ)ℙ_{XY}(K＼x_1^ℓ,y_1^ℓ)^s \\
		&=	q^{j+js}∑_{x_1^ℓ∈𝔽_q^ℓ}ℙ_{XY}(x_1^ℓ,y_1^ℓ)
			∑_{K∋x_1^ℓ}ℙ_𝕂(K)ℙ_{XY}(K＼x_1^ℓ,y_1^ℓ)^s \\
		&=	q^{j+js}∑_{x_1^ℓ∈𝔽_q^ℓ}ℙ_{XY}(x_1^ℓ,y_1^ℓ)
			q^{-j}∑_{K∋x_1^ℓ}ℙ_{𝕂|X}(K|x_1^ℓ)ℙ_{XY}(K＼x_1^ℓ,y_1^ℓ)^s \\
		&≤	q^{j+js}∑_{x_1^ℓ∈𝔽_q^ℓ}ℙ_{XY}(x_1^ℓ,y_1^ℓ)
			q^{-j}（∑_{K∋x_1^ℓ}ℙ_{𝕂|X}(K|x_1^ℓ)ℙ_{XY}(K＼x_1^ℓ,y_1^ℓ)）^s \\
		&=	q^{j+js}∑_{x_1^ℓ∈𝔽_q^ℓ}ℙ_{XY}(x_1^ℓ,y_1^ℓ)q^{-j}（∑_{K∋x_1^ℓ}
			ℙ_{𝕂|X}(K|x_1^ℓ)∑_{x_1^ℓ≠ξ_1^ℓ∈K}ℙ_{XY}(ξ_1^ℓ,y_1^ℓ)）^s \\
		&=	q^{j+js}∑_{x_1^ℓ∈𝔽_q^ℓ}ℙ_{XY}(x_1^ℓ,y_1^ℓ)q^{-j}（
			÷{q^{ℓ-j}-1}{q^ℓ-1}∑_{x_1^ℓ≠ξ_1^ℓ∈𝔽_q^ℓ}ℙ_{XY}(ξ_1^ℓ,y_1^ℓ)）^s \\
		&≤	∑_{\!x_1^ℓ∈𝔽_q^ℓ\!}ℙ_{XY}(x_1^ℓ,y_1^ℓ)
			（∑_{\!\!x_1^ℓ≠ξ_1^ℓ∈𝔽_q^ℓ\!\!}ℙ_{XY}(ξ_1^ℓ,y_1^ℓ)）^s
		≤	∑_{\!x_1^ℓ∈𝔽_q^ℓ\!}ℙ_{XY}(x_1^ℓ,y_1^ℓ)
			（∑_{\!ξ_1^ℓ∈𝔽_q^ℓ\!}ℙ_{XY}(ξ_1^ℓ,y_1^ℓ)）^s \\
		&=	∑_{x_1^ℓ∈𝔽_q^ℓ}ℙ_{XY}(x_1^ℓ,y_1^ℓ)ℙ_Y(y_1^ℓ)^s
		=	ℙ_Y(y_1^ℓ)ℙ_Y(y_1^ℓ)^s = ℙ_Y(y_1^ℓ)^{1+s}.
	\end{align*}
	Both major and minor arcs conquered,
	merge them and raise to the $(1+t)$-th power.
	The summand for any fixed $y_1^ℓ$ in
	the definition of the double-stroke E-null function is
	\begin{align*}
		&{}	（∑_Kℙ_𝕂(K)ℙ_{Y|𝕂}(y_1^ℓ｜K)^{1/(1+t)}）^{1+t}
		≤	(†minor†+†major†)^{1+t} ≤ †minor†^{1+t}+†major†^{1+t} \\
		&≤	（ℙ_Y(y_1^ℓ)^{1+s}）^{1+t}+（q^{js-ℓs}∑_{x_1^ℓ∈𝔽_q^ℓ}ℙ_X(x_1^ℓ)
			ℙ_{Y|X}(y_1^ℓ｜x_1^ℓ)^{1+s}）^{1+t} \\
		&=	ℙ_Y(y_1^ℓ)+q^{ℓt-jt}（∑_{x_1^ℓ∈𝔽_q^ℓ}ℙ_X(x_1^ℓ)
			ℙ_{Y|X}(y_1^ℓ｜x_1^ℓ)^{1+s}）^{1+t}
	\end{align*}
	We can finally bound the double-stroke E-null function per se:
	\begin{align*}
		\exp(-𝔼_0(t))
		&=	∑_{y_1^ℓ}（∑_Kℙ_𝕂(K)ℙ_{Y|𝕂}(y_1^ℓ｜K)^{1/(1+t)}）^{1+t} \\
		&≤	∑_{y_1^ℓ}ℙ_Y(y_1^ℓ)+q^{ℓt-jt}（∑_{x_1^ℓ∈𝔽_q^ℓ}
			ℙ_X(x_1^ℓ)ℙ_{Y|X}(y_1^ℓ｜x_1^ℓ)^{1+s}）^{1+t} \\
		&=	1+q^{ℓt-jt}∑_{y_1^ℓ}（∑_{x_1^ℓ∈𝔽_q^ℓ}
			ℙ_X(x_1^ℓ)ℙ_{Y|X}(y_1^ℓ｜x_1^ℓ)^{1+s}）^{1+t} \\
		&=	1+q^{ℓt-jt}\exp(-(\emph{the E-null function of }W^ℓ)(t)) \\
		&=	1+q^{ℓt-jt}\exp(-ℓE_0(t)).
	\end{align*}
	
	All efforts we spent on bounding $I_e(U_1^j；Y_1^ℓ)$ are for three creeds:
	First, we see Gallager's bound possessing innate elegance.
	Second, it fits the paradigm that solving the primary and the dual problems
	as a whole is easier than solving the primary problem alone.
	Third, the universal quadratic bound is
	waiting ahead for the E-null function.
	We infer that
	\begin{align*}
		𝔼I_e(U_1^j；Y_1^ℓ｜𝔾)
		&=	I_e(𝕂；Y_1^ℓ) = 𝔼_0'(0) ≤ ÷1t𝔼_0(t)
		=	÷1{-t}\log（\exp(-𝔼_0(t))） \\
		&≤	÷1{-t}\log（1+q^{ℓt-jt}\exp(-ℓE_0(t))）
		≤	÷1{-t}q^{ℓt-jt}\exp(-ℓE_0(t)) \\
		&=	\exp(-\log(-t)+(ℓ-j)t\log q-ℓE_0(t)).
	\end{align*}
	Recall the universal quadratic bound $E_0(t)≥I(W)t\log q-t^2\log(q)^2$
	as stated in \Cref{lem:quadratic} and used in the previous subsection.
	But this time $-2/5≤t<0$.
	We obtain that the exponent is
	\begin{align*}
		&{}	-\log(-t)+(ℓ-j)t\log q-ℓE_0(t) \\
		&=	-\log(-t)+(ℓ-H(W)ℓ+ℓ^{1/2+α})t\log q-ℓE_0(t) \\
		&≤	-\log(-t)+(ℓ-H(W)ℓ+ℓ^{1/2+α})t\log q
			-ℓ(I(W)t\log q-t^2\log(q)^2) \\
		&=	-\log(-t)+(ℓt\log q+ℓ^{1/2+α})t\log q\qquad
			†(redeem at $t=-ℓ^{-1/2+α}/2\log q$)† \\
		&↦	-\log(ℓ^{-1/2+α}/2\log q)
			-(-ℓℓ^{-1/2+α}/2+ℓ^{1/2+α})ℓ^{-1/2+α}/2 \\
		&=	\log(ℓ)/2-α\logℓ+\log2+\log\log q-ℓ^{2α}/4 \\
		&=	\log(ℓ)/2-\log\logℓ+\log2+\log\log q
			-ℓ^{2\log(\logℓ)/\logℓ}/4 \\
		&<	\log(ℓ)/2+\log\log q-\log(ℓ)^2/4.
	\end{align*}
	The first inequality uses $ℓ-j=ℓ-H(W)ℓ+ℓ^{1/2+α}$.
	The last inequality uses the assumption $ℓ≥e^2$.
	With the last line we conclude that $𝔼I_e(U_1^j；Y_1^ℓ｜𝔾)
		<\exp(\log(ℓ)/2+\log\log q-\log(ℓ)^2/4)=ℓ^{1/2-\log(ℓ)/4}\log q$.
	Switch back to the base-$q$ mutual information
	$𝔼I(U_1^j；Y_1^ℓ｜𝔾)<ℓ^{1/2-\log(ℓ)/4}$.
	We now reject kernels such that $I(U_1^j；Y_1^ℓ｜𝔾)≥ℓ^{1/2-\log(ℓ)/5}$.
	By Markov's inequality, the opposite direction ($<$) holds
	with probability $1-ℓ^{-\log(ℓ)/20}$ because $1/5+1/20=1/4$.
	Plug this upper bound into $h_α$.
	The left hand side of Inequality~\eqref{eq:eveblock} is less than
	\begin{multline*}
		jh_α（÷1jℓ^{1/2-\log(ℓ)/5}）=jj^{-α}ℓ^{α/2-α\log(ℓ)/5} \\
		<	ℓ^{1-α}ℓ^{α/2-α\log(ℓ)/5} = ℓ^{1-α/2-α\log(ℓ)/5}
		=	ℓ\log(ℓ)^{-1/2-\log(ℓ)/5}.
	\end{multline*}
	The inequality uses that the left hand side
	increases monotonically in $j$ and $j≔H(W)ℓ-ℓ^{1/2+α}<ℓ$.
	In any regard, the quantity at the end of the inequalities decays to $0$
	as $ℓ→∞$, so eventually it becomes less than $ℓ^{1/2+α}$,
	the right hand side of Inequality~\eqref{eq:eveblock}.
	This proves that Inequality~\eqref{eq:eve} holds with
	failing probability $ℓ^{-\log(ℓ)/20}$ as soon as $ℓ$ is large enough. % 20
	The lower bound on $ℓ$ in the statement of \Cref{lem:CLT} is large enough,
	hence the second half of \Cref{lem:CLT} settled.
	So is the whole lemma settled.

\subsection{Bibliographic remarks}

	Concerning the second moment bound:
	\cite[Lemma~1]{Arikan15v} has a looser bound comparing to \Cref{lem:second}.
	A similar bound for the third moment is \cite[Lemma~46]{PPV10},
	wherein Inequality~(468) looks dubious.
	In general, Gallager's E-null function is the cumulant generating function
	(the logarithm of the moment generating function),
	and bounding E-null is equivalent to bounding higher moments.
	Concerning the group symmetry:
	On both Bob and Eve's end, we use heavily
	the $2$-transitive nature of $\GL(ℓ,q)$'s action on $𝔽_q^ℓ$.
	Interestingly enough, $2$-transitivity is the main ingredient to prove
	that Reed--Muller codes achieve capacity over BECs \cite{KKMPSU17} as well.
	Concerning the secrecy bound:
	According to Hayashi \cite[Remark~4]{Hayashi06}, this technique of
	bounding secrecy exponent via the resolvability exponent and then
	the E-null function dated back to Oohama's conference paper \cite{Oohama02},
	although no formal proof was found there.
	See \cite{HM11,PTM17} for alternative
	descriptions and approaches on the same topic.
	
	For readers who took \Cref{lem:triforce} as granted
	or went through \Cref{pf:triforce} in advance,
	this is the last sentence of the proof of the main theorem%
	---polar codes' simplicity, random codes' durability.

\section{Conclusions} \label{sec:conclusions}

	Shannon introduced what we now understand as
	discrete memoryless channels seventy-two years ago.
	In the beginning, Shannon had no tool but developed their own theory
	of typical set, proved the noisy channel coding theory,
	and justified the notion of capacity.
	Gallager brought in error exponents.
	Capacities and error exponents quantify first and second order terms
	in the asymptotic performance of codes.
	Only in 2010 we are revealed the complete second order term.
	It was around the time that polar coding as a graceful instrument
	to explore the limits at low cost was discovered when Arıkan
	experimented with the channel transformation and with error exponents.
	Another ten years it took to grow
	variants and proof techniques of polar coding.
	Ultimately, it is feasible, and done by us coincidentally,
	to piece the puzzle together to show the mere possibility
	to achieve the second order limits at low cost.
	
	An overall comparison is integrated in \Cref{tab:3d}.
	Columns are classes of channels; from left to right:
	(BEC)		binary erasure channels;
	(BDMC)		binary-input discrete-output memoryless channels;
	($p$-ary)	channels of prime input size;
	($q$-ary)	channels of prime power input size;
	(finite)	channels of discrete input.
	Columns to the right are wider than columns to the left.
	The last column is exceptional;
	(asym.)		is about whether we can achieve the true Shannon capacity,
				instead of the symmetric capacity.
	Rows are goals; from top to bottom:
	(LLN)	to achieve (symmetric) capacity;
	(wLDP)	there exists $π>0$ such that $\P<\exp(-N^π)$;
	(wCLT)	there exists $ρ>0$ such that $R>I-N^{-ρ}$;
	(wMDP)	there exists $π,ρ>0$ such that
			$\P<\exp(-N^π)$ and $R>I-N^{-ρ}$ at once;
	(LDP)	the $π$ in (wLDP) can be arbitrarily close to $1$;
	(CLT)	the $ρ$ in (wCLT) can be arbitrarily close to $1/2$;
	(MDP)	the $(π,ρ)$-pair can be arbitrarily close to $π+2ρ=1$.
	Row (MDP) implies every other row; row (CLT) implies (wCLT);
	row (LDP) implies (wLDP); and every other row implies row (LLN).
	Rows (LDP) and (CLT) together almost imply (MDP)
	(need the partial distance profile).
	Cells represent how various goals are achieved over various channels.
	\def\CB#1#2{\!\colorbox{#1}{\!\smash{#2}\vphantom x\!}\!}%
	The \CB{xbec}{greenish background} means
	it is possible using Arıkan's kernel $\lol$.
	The \CB{xdecode}{purplish background}
	means it is possible using other kernels.
	The \CB{xflatten}{orangish background} means
	it is only possible using  dynamic kernels.
	
	\begin{table}
		\pgfplotstableread[col sep=ampersand]{
			\space&BEC       &BDMC      &$p$-ary  &$q$-ary &finite     &asym.  
			wLLN  &gArikan09 &gArikan09 &gSTA09x  &pSTA09x &pSTA09x    &gSRDR12
			wwLDP &gAT09     &gAT09     &gSTA09x  &pMT10c  &pSasoglu11 &gHY13  
			wwCLT &gKMTU10   &gHAU14    &gBGNRS18 &H       &H          &H      
			wwMDP &gGX13     &gGX13     &gBGS18   &H       &H          &H      
			wLDP  &pKSU10    &pKSU10    &pMT10c   &pMT10c  &H          &H      
			wCLT  &pFHMV17   &oGRY19    &H        &H       &H          &H      
			wMDP  &L         &H         &H        &H       &H          &H      
		}\nlab
		\def\arraystretch{1.44}
		\def\assigncolorcontent#1#2{\pgfkeyssetvalue
			{/pgfplots/table/@cell content}{\cellcolor{#1}#2}}
		\def\decoxdecodeorcontent#1#2\relax{
			     \if#1w\assigncolorcontent{white}{#2}
			\else\if#1g\assigncolorcontent{xbec}{\cite{#2}}
			\else\if#1p\assigncolorcontent{xdecode}{\cite{#2}}
			\else\if#1o\assigncolorcontent{xflatten}{\cite{#2}}
			\else\if#1L\assigncolorcontent{xflatten}{\cite{LargeDeviations18}}
			\else\if#1H\assigncolorcontent{xflatten}{Thm.~\ref{thm:hypotenuse}}
			\fi\fi\fi\fi\fi\fi
		}
		$$\pgfplotstabletypeset[
			every head row/.style={before row=\toprule,after row=\midrule},
			every last row/.style={after row=\bottomrule},
			assign cell content/.code={\decoxdecodeorcontent#1\relax}
		]\nlab$$
		\caption{
			Polar coding works arranged by their contribution
			in terms of targeted channels and targeted behaviors.
			See \Cref{sec:conclusions} for details.
		} \label{tab:3d}
	\end{table}
	
	The following works made critical progresses
	but our classification fails to include them:
	Rate-dependent result in LDP paradigm \cite{HMTU13}.
	Optimal relations among channel parameters \cite{MT14}.
	First family of $(π,ρ)$ pairs in MDP paradigm \cite{MHU16}.
	AWGNCs intersecting MDP \cite{FT17}.
	
	We did our best to excavate the archive but
	throughout the course of manuscript preparation we found ourselves
	underestimating early works multiple times so the record kept updating.
	We sincerely hope to hear about possible references to add to the table.
	
	Potential improvements include but are not limited to the following:
	(Tolls)	Tighten the explicit Hölder tolls.
		The current toll between any pair of parameters $H$, $\P$, $Z$, $\Z$,
		$T$, $S$, and $\S$ is roughly the sum of tolls collected when
		traveling through the spanning tree illustrated in \Cref{lem:imtoll}.
		Some improvements potentially tighten the bounds in \Cref{lem:triforce}.
	(FTPC)	Tighten the two fundamental theorems such that
		they degenerate to equalities over erasure channels.
		Once done,
		$\{𝘡_n\}$ is a supermartingale and \Cref{pf:super} is obsolete.
	(Symmetry)	generalize the arguments in
		\Cref{sec:ChangSahai,sec:Gallager,sec:Hayashi} to asymmetric channels.
		Once done, \Cref{sec:symmetrize} is obsolete.
		Note that the proof of the fundamental theorems
		applies to asymmetric channels.
	(Bijection)	Early works on polar coding over arbitrary alphabets
		introduced arbitrary bijections $\g W$.
		Generalize the two fundamental theorems
		to include arbitrary bijections.
	(Dynamic)	Achieve the main theorem with a large, but fixed, kernel.
		This does not immediately make the code practical.
		But the answer should shed light on our understanding of coding.
	(Alphabet)	Achieve the main theorem without
		the reduction to prime power alphabets.
		This is currently not an option because
		\emph{linear codes} are barely defined over non-fields.
		Plus the $S$-parameter---and thus FTPC$S$---would just break.
	(Dispersion)	Recall \Cref{pro:optimality}.
		Weird things happens when the channel dispersion vanishes $V=0$.
		Can we describe those channels better?
		One example of such channels is this:
		\[*\begin{bmatrix} 1/2&1/2&0 \\ 1/2&0&1/2 \\ 0&1/2&1/2 \end{bmatrix}.\]*
	
	We look forward to generalizations of the main theorem to
	non-identical channels (i.e., non-stationary) \cite{Mahdavifar17},
	non-independent channels (i.e., with memory) \cite{WHYLH15,ST16},
	deletion channels \cite{TPFV19,LT19},
	channels with restrictions on input distributions
		(e.g., due to energy constrain) \cite{FT16},
	wiretap channels \cite{SV13},
	rate-distortion problem \cite{HKU09},
	Wyner-Ziv problem \cite{HKU09},
	Slepian-Wolf problem \cite{Abbe15},
	broadcast channels \cite{GAG15,MHSU15}, and
	multiple access channels \cite{AT12,NT16}.
	We focus on noisy channel coding in this work
	for its historical significance.

\appendix

\section{Explicit Hölder Tolls (Proof of Lemma~\ref{lem:extoll})}
	\label{pf:extoll}

	As is promised in \Cref{sec:tolls}, we prove the explicit Hölder toll.
	Let $W$ be a $q$-ary channel.
	In the upcoming arguments, $H$, $\P$, $Z$, $\Z$, $S$, and $\S$ mean
	$H(W)$, $\P(W)$, $Z(W)$, $\Z(W)$, $S(W)$, and $\S(W)$, respectively.
	Also $q'$ means $q-1$, and $q''$ means $q-2$.
	Furthermore, $\lg$ means the base-$2$ logarithm;
	this is handy when we jump back and forth between nats, bits, and $q$-bits.
	
	First we show
	\[*\Z≤q√{H\log_4q}. \repeattag{eq:Z<H}\]*
	Start from $\Z$:
	By the definition $\Z≤q'Z$.
	Move on to $Z$:
	By \Cref{lem:pz}, $q'q^{-2}(√{1+q'Z}-√{1-Z})^2≤\P$
	so $√{1+q'Z}-√{1-Z}≤q√{\P/q'}$.
	Multiplying by the conjugate yields
	$(1+q'Z)-(1-Z)≤q√{\P/q'}(√{1+q'Z}+√{1-Z})$.
	The left hand side is $qZ$;
	in the right hand side $√{1+q'z}+√{1-z}$ has
	maximum $q/√{q'}$ at $z=q''/q'$ by calculus.
	So $Z≤√{\P/q'}(q/√{q'})=q√{\P}/q'$.
	Move on to $\P$:
	By \Cref{lem:ph} (the first lower bound),
	$2\P≤H\lg q$ or equivalently $\P≤H\log_4q$.
	Now we chain the inequalities $\Z≤q'Z≤q√{\P}≤q√{H\log_4q}$.
	This completes Inequality~\eqref{eq:Z<H}.
	That being proven, we use the weaker form
	\[*\Z≤q^3√H\]*
	in the calculus machinery for global MDP.
	
	Second we show
	\[*H≤√{eq'\Z/2}. \repeattag{eq:H<Z}\]*
	Start from $H$:
	By \Cref{lem:ph} (the upper bound, Fano's inequality),
	$H\lg q≤h_2(\P)+\P\lg q'$.
	By \Cref{fig:hqre},
	$h_2(\P)+\P\lg q'≤√{e\P}+\P\lg q'=√{\P}(√e+√{\P}\lg q')$.
	What is inside parentheses is less than $√e+√{q'/q}\lg q'$.
	Hence $H≤√{\P}(√e+√{q'/q}\lg q')/\lg q$.
	Focus on the scalar part---%
	$(√e+√{q'/q}\lg q')/\lg q$ has maximum $√e$ at $q=2$ (remember that $q≥2$).
	So $H≤√{e\P}$.
	Move on to $\P$:
	By \Cref{lem:pz}, $\P≤q'Z/2$.
	Move on to $Z$:
	By definition $Z≤\Z$.
	Now we chain the inequalities $H≤√{e\P}≤√{eq'Z/2}≤√{eq'\Z/2}$.
	This completes Inequality~\eqref{eq:H<Z}.
	That being proven, we use the weaker form
	\[*H≤q^3√{\Z}\]*
	in the calculus machinery for global MDP.
	
	Third we show (notice the logarithm is natural)
	\[*\S≤q'q√{(1-H)\log(q)/2}. \repeattag{eq:S<1-H}\]*
	Start from $\S$:
	By definition $\S≤q'S$.
	Move on to $S$:
	By \Cref{lem:ps}, $S≤q'q(q'/q-\P)√{1-÷q{q'}÷{q''}{q'}}$.
	The square root simplifies to $√{1/(q')^2}=1/q'$ as $qq''=(q')^2-1$.
	So $S≤q'-q\P$.
	Move on to $q'-q\P$:
	By \Cref{lem:ph} (the upper bound, Fano's inequality),
	$H\lg q≤h_2(\P)+\P\lg q'$.
	We claim that $h_2(\P)+\P\lg q'≤\lg q-2(q'/q-\P)^2/\log2$.
	To prove the claim, Taylor expand both sides at $\P=q'/q$.
	Verify that both evaluate to $\lg q$ at $\P=q'/q$;
	verify that both have derivative $0$ at $\P=q'/q$; and
	verify that the acceleration of the left hand side, $-1/(\P(1-\P)\log2)$,
	is more negative than the acceleration of the right hand side, $-4/\log2$.
	By Taylor's theorem, mean value theorem, or Euler method,
	the function with greater acceleration is greater;
	hence the claim.
	See also \cite[Fig.~1]{FM94};
	the $Φ$-curve seems parabolic at the upper right corner.
	Now we have $H\lg q≤\lg q-2(q'/q-\P)^2/\log2$, which is equivalent to
	$2(q'/q-\P)^2/\log q≤1-H$ and to $q'-q\P≤q√{(1-H)\log(q)/2}$.
	Now we chain the inequalities $\S≤q'S≤q'(q'-q\P)≤q'q√{(1-H)\log(q)/2}$.
	This completes Inequality~\eqref{eq:S<1-H}.
	That being proven, we use the weaker form
	\[*\S≤q^3√{1-H}\]*
	in the calculus machinery for global MDP.
	
	Fourth we show
	\[* 1-H≤q'\S/\log q. \repeattag{eq:1-H<S}\]*
	Start from $1-H$:
	By \Cref{lem:ph} (the second lower bound),
	$H\lg q≥q'q\lg(q/q')\*(\P-q''/q')+\lg q'$.
	The right hand side is $\lg q-q'\lg(q/q')(q'-q\P)$
	by matching the (rational) coefficients of
	$\P\lg q$, $\P\lg q'$, $\lg q$, and $\lg q'$, respectively.
	As $H\lg q≥\lg q-q'\lg(q/q')(q'-q\P)$ we bound
	$\lg(q/q')=\lg(1+1/q')≤1/q'$ by the tangent line at $1/q'=0$.
	So $H\lg q≥\lg q-(q'-q\P)$ and hence $1-H≤(q'-q\P)/\lg q$.
	Move on to $q'-q\P$:
	By \Cref{lem:ps}, $1-q\P/q'≤S$ so $q'-q\P≤q'S$.
	Move on to $S$:
	By definition $S≤\S$.
	Now we chain the inequalities
	$1-H≤(q'-q\P)/\lg q≤q'S/\lg q≤q'\S/\lg q$.
	This completes Inequality~\eqref{eq:1-H<S}.
	That being proven, we use the weaker form
	\[*1-H≤q^3√{\S}\]*
	in the calculus machinery for global MDP.
	
	This is end of the proof of \Cref{lem:extoll}.
	The proof of Lemma~\ref{lem:imtoll} follows the same logic, only shorter.

\section{Calculus Machinery for Global MDP\texorpdfstring\\{}
	(Proof of Lemma \ref{lem:triforce})} \label{pf:triforce}

	We are to prove that
	\[*𝘗\{𝘏_n<\exp(-ℓ^{πn}n)\}>1-𝘏_0-ℓ^{-ρn+o(n)} \repeattag{eq:finiteso}\]*
	given criteria (cb), (cm), (ct), and (cl), the local LDP behavior,
	the local CLT behavior, and that $π+2ρ≤1-8α$.
	The proof is split into several stepping stones.
	We will prove each of the following inequalities (including two equalities)
	in each of the upcoming subsections.
	This will be proven in \Cref{pf:eigen}:
	The \emph{eigen behavior} reads
	\[𝘌[h_α(𝘏_{n+1})｜ℱ_n]≤4ℓ^{-1/2+3α}h_α(𝘏_n). \label{eq:eigen}\]
	This will be proven in \Cref{pf:01limit}:
	As a lemma, $\{𝘏_n\}$ and $\{𝘡_n\}$
	converges to $0$ with probability $1-𝘏_0$, i.e.,
	\[𝘗\{𝘡_n→0\}=𝘗\{𝘏_n→0\}=1-𝘏_0. \label{eq:01limit}\]
	This will be proven in \Cref{pf:en23}:
	The \emph{en23 behavior} reads
	\[𝘗\{𝘡_n<\exp(-n^{2/3})\}>𝘏_0-ℓ^{(-1/2+4α)n+o(n)}. \label{eq:en23}\]
	This will be proven in \Cref{pf:super}:
	As a lemma, $\{\min(ℓ^{-2},√[4]{𝘡_n})\}$ is a supermartingale, i.e.,
	\[𝘌[\min(ℓ^{-2},√[4]{𝘡_{n+1}})｜ℱ_n]≤\min(ℓ^{-2},√[4]{𝘡_n}).
		\label{eq:super}\]
	This will be proven in \Cref{pf:cramer}:
	As a lemma, the following holds when $𝘡_0<ℓ^{-8}$:
	\[𝘡_{n+1}≤𝘡_n^{⌈𝘒_{n+1}^2/3ℓ⌉·3/4}\qquad†and†\qquad
		𝘌[(⌈𝘒_{n+1}^2/3ℓ⌉·3/4)^{-1/2}｜ℱ_n]<ℓ^{-1/2+2α}. \label{eq:cramer}\]
	This will be proven in \Cref{pf:een13}:
	The \emph{een13 behavior} reads
	\[𝘗\{𝘡_n<\exp(-e^{n^{1/3}})\}>1-𝘏_0-ℓ^{(-1/2+4α)n+o(n)}. \label{eq:een13}\]
	This will be proven in \Cref{pf:elpin}:
	The \emph{elpin behavior} reads
	for any constants $π,ρ>0$ such that $π+ρ≤1-8α$,
	\[𝘗\{𝘡_n<\exp(-ℓ^{πn}n^2)\}>1-𝘏_0-ℓ^{-ρn+o(n)}. \label{eq:elpin}\]
	The last inequality is a bi-Hölder toll away from 
	\[*𝘗\{𝘏_n<\exp(-ℓ^{πn}n)\}>1-𝘏_0-ℓ^{-ρn+o(n)}, \repeattag{eq:finiteso}\]*
	our destination.
	This finishes the proof of \Cref{lem:triforce}.
	
	The eigen, en23, een13, and elpin behaviors are intermediate checkpoints
	pinned in a way that moving from one to the next is easy
	while skipping any of them makes the next unreachable.
	Their entire purpose is to form a chain that connects
	the local LDP and CLT behaviors to the global MDP behavior and
	we do not specify if any of them falls inside the LDP, CLT or MDP paradigm.

\subsection{The eigen behavior} \label{pf:eigen}

	We want to prove Inequality~\eqref{eq:eigen},
	$𝘌[h_α(𝘏_{n+1})｜ℱ_n]≤4ℓ^{-1/2+3α}h_α(𝘏_n)$,
	given the local LDP behavior and the local CLT behavior.
	The idea is that, for $𝘏_n$ that is close to $1/2$, the local CLT behavior
	provides a measurement of the dichotomy/bifurcation behavior of $𝘏_{n+1}$.
	For $𝘏_n$ that are close to $0$, the $𝘡$-part of the local LDP behavior
	provides a measurement of the attraction toward $0$.
	For $𝘏_n$ that is close to $1$, the $𝘚$-part handles it dually.
	The formal proof is below.
	
	Inequality~\eqref{eq:eigen} is a local statement so we may assume $n=0$.
	To prove that $𝘌[h_α(𝘏_1)]≤4ℓ^{-1/2+3α}h_α(𝘏_0)$,
	we divide it into three cases per how
	$𝘏_0$ compares to $ℓ^{-2}$ and $1-ℓ^{-2}$.
	The mediocre case: if $ℓ^{-2}≤𝘏_0≤1-ℓ^{-2}$ then $h_α(𝘏_0)≥ℓ^{-2α}$.
	The local CLT behavior implies
	$𝘌[h_α(𝘏_1)]<4ℓ^{-1/2+α}=4ℓ^{-1/2+3α}ℓ^{-2α}≤4ℓ^{-1/2+3α}h_α(𝘏_0)$
	and we are done with this case.
	The noisy case: if $𝘏_0>1-ℓ^{-2}$, we replace $(𝘏,𝘡,𝘚)$ by $(1-𝘏,𝘚,𝘡)$
	to dual it to the reliable case dealt below and we are done with this case.
	(This is the only place in the proof where we ever mentioned $𝘚$ explicitly.
	Nevertheless, every statement concerning $𝘡$ concerns $𝘚$ by duality.)
	
	The last case---the reliable case: when $𝘏_0<ℓ^{-2}$, we further
	split it into two subcases per how $𝘒_1$ compares to $k≔ℓ^{1/2+5α/2}$.
	For the small $𝘒_1$ subcase, the martingale property fits.
	For the large $𝘒_1$ subcase, the local LDP behavior fits:
	\begin{align*}
		𝘌[h_α(𝘏_1)]
		&=	𝘌[h_α(𝘏_1)｜𝘒_1≤k]k/ℓ+𝘌[h_α(𝘏_1)｜𝘒_1>k](1-k/ℓ) \\
		&≤	h_α(𝘌[𝘏_1｜𝘒_1≤k])k/ℓ+h_α(𝘌[𝘏_1｜𝘒_1>k])(1-k/ℓ).
			\steplabel{eq:reliable}
	\end{align*}
	For the $𝘒_1≤k≔ℓ^{1/2+5α/2}$ subcase,
	the martingale property $𝘌[𝘏_1]=𝘏_0$ implies that $𝘌[𝘏_1｜𝘒_1≤k]≤𝘏_0ℓ/k$.
	Thus
	$h_α(𝘌[𝘏_1｜𝘒_1≤k])k/ℓ≤h_α(𝘏_0ℓ/k)k/ℓ=h_α(𝘏_0)ℓ^αk^{-α}kℓ^{-1}
		=h_α(𝘏_0)ℓ^αℓ^{-α/2-5α^2/2}ℓ^{1/2+5α/2}ℓ^{-1}≤ℓ^{-1/2+3α}h_α(𝘏_0)$.
	And the $𝘒_1≤k$ subcase is closed.
	For the $𝘒_1>k≔ℓ^{1/2+5α/2}$ subcase,
	pay the explicit Hölder toll: $𝘡_0≤q^3√{𝘏_0}<q^3/ℓ<1$.
	Invoke the local LDP behavior:
	$𝘌[𝘡_1｜𝘒_1>k]≤𝘌[ℓ\exp(q𝘡_0ℓ)(q𝘡_0)^{⌈𝘒_1^2/3ℓ⌉}｜𝘒_1>k]
		≤ℓ\exp(q𝘡_0ℓ)(q𝘡_0)^{k^2/3ℓ}≤ℓ\exp(q^4)(q^4√{𝘏_0})^{k^2/3ℓ}
		=ℓ\exp(q^4)(q^8𝘏_0)^{\log(ℓ)^5/6}$.
	Pay the return-trip toll:
	$𝘏_1≤q^3√{𝘡_1}≤q^3ℓ^{1/2}\exp(q^4/2)(q^8𝘏_0)^{\log(ℓ)^5/12}$.
	Now we claim and prove that the following quantity is less than $1$:
	(there is nothing to show if $h_α(𝘏_0)=0$)
	\begin{align*}
		&{}	(h_α(𝘏_1)/ℓ^{-1/2+3α}h_α(𝘏_0))^{12/α}
		=	𝘏_1^{12}ℓ^{6/α-36}𝘏_0^{-12} < 𝘏_1^{12}ℓ^{6\logℓ-30}𝘏_0^{-12} \\
		&≤	q^{36}ℓ^6\exp(6q^4)(q^8𝘏_0)^{\log(ℓ)^5}ℓ^{6\logℓ-30}𝘏_0^{-12}
		=	q^{36+8\log(ℓ)^5}e^{6q^4}ℓ^{6\logℓ-24}𝘏_0^{\log(ℓ)^5-12} \\
		&<	q^{36+8\log(ℓ)^5}e^{6q^4}ℓ^{6\logℓ-24}ℓ^{-2\log(ℓ)^5+24}
		=	q^{36+8\log(ℓ)^5}e^{6q^4}ℓ^{6\logℓ-1.6\log(ℓ)^5-0.4\log(ℓ)^5} \\
		&≤	q^{36+8\log(ℓ)^5}e^{6q^4}ℓ^{6\logℓ-8\log(q)\log(ℓ)^4-0.4\log(ℓ)^5}
		=	q^{36}e^{6q^4}ℓ^{6\logℓ-0.4\log(ℓ)^5} \\
		&=	q^{36}e^{6q^4}ℓ^{6\logℓ-0.1\log(ℓ)^5}e^{0.3\log(ℓ)^6}
		<	q^{36}e^{6q^4}ℓ^{6\logℓ-0.1\log(ℓ)^5}e^{0.3\log(41)^2(q\log3)^4} \\
		&<	q^{36}e^{6q^4}ℓ^{6\logℓ-0.1\log(ℓ)^5}e^{6.02q^4}
		<	q^{36}ℓ^{6\logℓ-0.1\log(ℓ)^5} \\
		&=	\exp(36\log q+6\log(ℓ)^2-\log(ℓ)^6/30-\log(ℓ)^6/15) \\
		&=	\exp(36\log q+6\log(ℓ)^2
			-5\log(q)\log(22)^5/30-\log(22)^4\log(ℓ)^2/15) \\
		&<	\exp(0) ≤ 1.
	\end{align*}
	The inequality involving $1.6$ uses $ℓ≥q^5$.
	The inequality involving $0.3$ uses $ℓ≥\max(41,3^q)$.
	The inequality involving $15$ uses $ℓ≥\max(22,q^5)$.
	We just showed that $h_α(𝘏_1)/ℓ^{-1/2+3α}h_α(𝘏_0)$ is less than $1$,
	with and hence without the $12/α$-th power.
	Thus $𝘌[h_α(𝘏_1)｜𝘒_1>k]≤ℓ^{-1/2+3α}h_α(𝘏_0)$.
	And the $𝘒_1>k≔ℓ^{1/2+5α/2}$ subcase is closed.
	To sum up the reliable case:
	We bound separately the two terms in Formula~\eqref{eq:reliable}.
	They are both at most $ℓ^{-1/2+3α}h_α(𝘏_1)$,
	hence their sum is at most $2ℓ^{-1/2+3α}h_α(𝘏_1)$.
	Since Inequality~\eqref{eq:eigen} wants $4$ instead of 2,
	The reliable case is closed.
	And the proof of the eigen behavior, Inequality~\eqref{eq:eigen}, is sound
	when combining the three cases.
	
	Bibliographic remarks:
	\cite[Theorem~7]{FHMV17} also cut the cases at $ℓ^{-2}$ and $1-ℓ^{-2}$.
	In contrast, \cite[Theorem~5.1]{GRY19} cut at $ℓ^{-4}$ and $1-ℓ^{-4}-ε$.
	A potential improvement is, when $ℓ^{-2}≤𝘏_0<ℓ^{-1}$,
	Inequality~\eqref{eq:eve} will simply evaporate.
	Similarly, Inequality~\eqref{eq:bob} evaporates
	when $1-ℓ^{-1}<𝘏_0≤1-ℓ^{-2}$.
	They tighten the right hand side of Inequality~\eqref{eq:CLT}.
	The lesson here is that the hard transition between
	local LDP and CLT behaviors weakens the bounds.

\subsection{Polarization in mean} \label{pf:01limit}

	We want to prove Equality~\eqref{eq:01limit}, $𝘗\{𝘡_0→0\}=𝘗\{𝘏_0→0\}=1-𝘏_0$,
	given the martingale property and the eigen behavior.
	The idea is that the eigen behavior expels $𝘏_n$ from being close to $1/2$,
	so the only reasonable limits are $0$ and $1$.
	The formal proof is below.
	
	As a bounded martingale $\{𝘏_n\}$ converges to an r.v.%
	---which we call $𝘏_∞$---a.s.\ (almost surely).
	This is Doob's martingale convergence theorem
	\cite[Theorem~4.2.11]{Durrett19}.
	Owing to $h_α$'s continuity, $h_α(𝘏_n)→h_α(𝘏_∞)$ a.s.
	Point-wise convergence and (uniform) boundedness imply convergence in $L^1$,
	i.e., $𝘌[h_α(𝘏_n)]→𝘌[h_α(𝘏_∞)]$ as $n→∞$.
	This is Lebesgue's dominated convergence theorem
	\cite[Theorem~1.6.7]{Durrett19}.
	By the eigen behavior, $𝘌[h_α(𝘏_n)]$ decays toward $0$ by a constant factor
	every time $n$ increases, thus $𝘌[h_α(𝘏_∞)]$ is $0$.
	This forces $h_a(𝘏_∞)=0$ a.s.\ and hence $𝘏_∞∈\{0,1\}$ a.s.
	Since $𝘏_∞$ is Bernoulli
	 $𝘗\{𝘏_∞=0\}=𝘌[𝕀\{𝘏_∞=0\}]=𝘌[1-𝘏_∞]←𝘌[1-𝘏_n]=1-𝘏_0$.
	So $𝘗\{𝘏_n→0\}=1-𝘏_0$.
	By the implicit bi-Hölder toll, $𝘏_n→0$ if and only if $𝘡_n→0$,
	thus the latter has the same probability measure.
	And the proof of Equality~\eqref{eq:01limit} is sound.
	
	Bibliographic remarks:
	The statement $𝘏_n→𝘏_∞∈\{0,1\}$ is usually referred to
	as \emph{channel polarization} in spite of that it does not guarantee
	the corresponding codes to be capacity-achieving.
	See also \cite[Proposition~10]{Arikan09} \cite[Definition~3]{MT14}
	\cite[Lemma~3.8]{Sasoglu11}.
	This lemma should have been bestowed upon the fundamental theorem
	but it is not mandatory if some sort of CLT behavior is present.
	See \cite[Lemma~1]{MHU16} \cite[Lemma~4]{FHMV17} \cite[Lemma~9.5]{GRY19}.
	Recently, Reed--Muller codes' channels are shown to polarize \cite{AY19};
	Reed--Muller codes achieving capacity is not a consequence,
	but a different story.

\subsection{The en23 behavior} \label{pf:en23}

	We want to prove $𝘗\{𝘡_n<\exp(-n^{2/3})\}<1-𝘏_0-ℓ^{(-1/2+4α)n+o(n)}$,
	namely Inequality~\eqref{eq:en23},
	given the eigen behavior and the polarization in mean.
	The idea is to read off the behavior of $\{H_n\}$
	from the behavior of $\{h_α(𝘏_n)\}$ in the eigen behavior.
	The formal proof is below.
	
	$𝘌[h_α(𝘏_{n+1})｜ℱ_n]≤ℓ^{-1/2+4α}h_α(𝘏_n)$ by $ℓ≥e^4$.
	This simplifies the eigenvalue.
	Without loss of generality, we rescale $h_α$ such that $h_α(𝘏_0)=1$.
	Let $ε_n$ be $\exp(-n^{3/4})$;
	note that $ε_n≤𝘏_0≤1-ε_n$ for $n$ large enough.
	Owing to $h_α$'s concavity, that $h_α(0)=h(1)=0$, and that $h_α(𝘏_0)=1$,
	we deduce that $h_α(z)≥ε_n$ whenever $ε_n≤z≤1-ε_n$.
	Consider these three events as a partition:
	let $𝘈_n$ be $\{𝘏_n<ε_n\}$;
	let $𝘉_n$ be $\{ε_n≤𝘏_n≤1-ε_n\}$;
	let $𝘊_n$ be $\{1-ε_n<𝘏_n\}$.
	Note that $𝘉_n$ implies $h_α(𝘏_n)≥ε_n$.
	
	Next we show $𝘗(𝘉_n)<ℓ^{(-1/2+4α)n+o(n)}$:
	Telescoping leads to $𝔼[h_α(𝘏_n)]≤h_a(𝘏_0)ℓ^{(-1/2+4α)n}=ℓ^{(-1/2+4α)n}$.
	Markov's inequality leads to $𝘗\{h(𝘏_n)≥ε_n\}≤𝔼[h(𝘏_n)]/ε_n
		≤ℓ^{(-1/2+4α)n}/ε_n=ℓ^{-(-1/2+4α)n+O(n^{3/4})}<ℓ^{(-1/2+4α)n+o(n)}$.
	Therefore $𝘗(𝘉_n)≤𝘗\{h(𝘏_n)≥ε_n\}<ℓ^{(-1/2+4α)n+o(n)}$, as desired.
	Moreover, summing the geometric series leads to
	$∑_{m≥n}𝘗(𝘉_m)<ℓ^{(-1/2+4α)n+o(n)}$.
	
	Next we show $1-𝘏_0-𝘗(𝘈_n)<ℓ^{(-1/2+4α)n+o(n)}$:
	The left hand side is at most
	the probability measure of $\{𝘏_∞=0\}、𝘈_n$.
	That is the probability that $𝘏_n$ was not small (not in $𝘈_n$)
	but $𝘏_{n+1},𝘏_{n+1},\dotsc$ will end up converging to $0$.
	Note that being a martingale causes $1-𝘏_{n+1}≤ℓ(1-𝘏_n)$,
	which forbids $𝘏_n$ jumping from $𝘊_n$ directly into $𝘈_{n+1}$---%
	it must pass by $𝘉_m$ for some $m≥n$ before ever landing in $𝘈_{m+1}$.
	From the summation of $𝘗(𝘉_m)$ over $m≥n$ we know that
	very few descendants of $𝘏_n$ can do that;
	the probability measure of $\{𝘏_∞=0\}、𝘈_n$
	is less than $ℓ^{(-1/2+4α)n+o(n)}$.
	Therefore $1-𝘏_0-𝘗(𝘈_n)<ℓ^{(-1/2+4α)n+o(n)}$
	and hence $𝘗\{𝘏_n<\exp(-n^{3/4})\}=𝘗(𝘈_n)>1-𝘏_0-ℓ^{(-1/2+4α)n+o(n)}$.
	Pay the bi-Hölder toll $𝘗\{𝘡_n<\exp(-n^{2/3})\}>1-𝘏_0-ℓ^{(-1/2+4α)n+o(n)}$.
	And the proof of the en23 behavior, Inequality~\eqref{eq:en23}, is sound.

\section{The een13 and elpin Behaviors} \label{pf:highrule}

	In this section, we continue proving \Cref{lem:triforce}.
	The previous section covers \eqref{eq:eigen} to~\eqref{eq:en23}.
	We are left with \eqref{eq:super} to~\eqref{eq:elpin}.

\subsection{A Supermartingale} \label{pf:super}

	We want to show that a certain monotonic function in $Z_n$
	is a supermartingale so we can control how frequently does $𝘡_n$
	stay in the turf where the local LDP behavior dominates.
	Making it a supermartingale, we are able to cite Doob's optional stopping
	theorem \cite[Theorem~4.8.4 and Exercise~4.8.2]{Durrett19} later.
	The formal proof is below.
	
	Inequality~\eqref{eq:super} is a local statement so we may assume $n=0$.
	To prove that $𝘌[\min(ℓ^{-2},√[4]{𝘡_1})]≤√[4]{𝘡_0}$,
	we may assume $𝘡_0<ℓ^{-8}$ or the inequality becomes trivial.
	Invoke the local LDP behavior
	$𝘡_1≤ℓ\exp(q𝘡_0ℓ)(q𝘡_0)^{⌈𝘒_1^2/3ℓ⌉}≤2ℓ(q𝘡_0)^{⌈𝘒_1^2/3ℓ⌉}$.
	The last inequality uses $q≤ℓ$.
	When $𝘒_1≤√{3ℓ}$, we do nothing but apply the last-resort exponent $1$:
	\[*𝘌[\min(ℓ^{-2},√[4]{𝘡_1})｜𝘒_1≤√{3ℓ}]≤√[4]{ℓeq𝘡_0}.\]*
	When $𝘒_1>√{3ℓ}$, the stronger exponent applies:
	\begin{multline*}
		𝘌[\min(ℓ^{-2},√[4]{𝘡_1})｜𝘒_1>√{3ℓ}]
		≤	𝘌[√[4]{ℓe(q𝘡_0)^{⌈𝘒_1^2/3ℓ⌉}}｜𝘒_1>√{3ℓ}] \\
		≤	√[4]{ℓe(q𝘡_0)^2} = √[4]{ℓeq^2𝘡_0𝘡_0}
		≤	√[4]{ℓeq^2𝘡_0/ℓ^8} ≤ √[4]{eq^2𝘡_0/ℓ^7}.
	\end{multline*}
	Combining the two cases that are cut per how $𝘒_1$ compares to $√{3ℓ}$,
	we infer that
	\begin{multline*}
		𝘌[\min(ℓ^{-2},√[4]{𝘡_1})]
		=	𝘌[√[4]{𝘡_1}｜𝘒_1≤√{3ℓ}]·÷{√{3ℓ}}{ℓ}
			+𝘌[√[4]{𝘡_1}｜𝘒_1>√{3ℓ}]·÷{ℓ-√{3ℓ}}{ℓ} \\
		≤	√[4]{ℓeq𝘡_0}·√{3ℓ}/ℓ+√[4]{eq^2𝘡_0/ℓ^3}·ℓ/ℓ
		=	(√[4]{9eq/ℓ}+√[4]{eq^2/ℓ^7})√[4]{𝘡_0} ≤ √[4]{𝘡_0}.
	\end{multline*}
	The last inequality uses $ℓ≥\max(50,q^5)$.
	And the proof of Inequality~\eqref{eq:super} is sound.
	
	Bibliographic remarks:
	This lemma is inspired by \cite[Proposition~9]{Arikan09}.
	In \cite[Lemma~22]{HAU14} Arıkan's lemma is overlooked and
	another is reinvented that serves the same purpose.
	The latter lemma also served in \cite[Theorem~3]{MHU16}.
	We generalized the idea to non-binary cases
	in \cite[Lemma~1]{LargeDeviations18}.
	The quartic root here is an aesthetic choice;
	$\min(ℓ^{-2},√[2+ε]{𝘡_n})$ is also a supermartingale
	but only for astronomic $ℓ$ (depending on $q$).
	For any non-random kernel, a small enough power works
	provided that the kernel polarizes channels in the first place.

\subsection{A Cramér--Chernoff gadget} \label{pf:cramer}

	Let $𝘋_{n+1}$ be $⌈𝘒_{n+1}^2/3ℓ⌉·3/4$.
	We want to prove inequalities in \eqref{eq:cramer}, that $𝘡_n<ℓ^{-8}$
	implies $𝘡_{n+1}≤𝘡_n^{𝘋_{n+1}}$ and $𝘌[𝘋_{n+1}^{-1/2}｜ℱ_n]<ℓ^{-1/2+2α}$,
	given the local LDP behavior.
	The motivation is to reformat the local inequalities
	so that it is easy to telescope for future reference.
	The formal proof is below.
	
	They are both local statements so we may assume $n=0$.
	When $𝘡_0<ℓ^{-8}$, invoke the local LDP behavior
	$𝘡_1≤ℓ\exp(q𝘡_0ℓ)(q𝘡_0)^{⌈𝘒_1^2/3ℓ⌉}≤ℓe(q𝘡_0)^{⌈𝘒_1^2/3ℓ⌉}
		≤ℓe(q^4𝘡_0)^{⌈𝘒_1^2/3ℓ⌉/4}𝘡_0^{⌈𝘒_1^2/3ℓ⌉·3/4}
		≤ℓe(ℓ^{-7})^{⌈𝘒_1^2/3ℓ⌉/4}𝘡_0^{⌈𝘒_1^2/3ℓ⌉·3/4}≤𝘡_0^{⌈𝘒_1^2/3ℓ⌉·3/4}$.
	The fourth inequality uses $ℓ≥q^4$.
	That validates the first inequality in \eqref{eq:cramer}.
	For the second inequality,
	\begin{align*}
		𝘌[𝘋_1^{-1/2}]
		&=	÷1{ℓ}∑_{k=1}^ℓ（\Bigl\lceil÷{k^2}{3ℓ}\Bigr\rceil·÷34）^{-1/2}
		<	÷1{ℓ}∑_{k=1}^{√{3ℓ}}（÷34）^{-1/2}
			+÷1{ℓ}∑_{k=√{3ℓ}+1}^ℓ（÷{k^2}{4ℓ}）^{-1/2} \\
		&<	÷1{ℓ}√{3ℓ}÷2{√3}+÷1{ℓ}√{4ℓ}∫_{√{3ℓ}}^ℓ÷{dk}k
		=	2ℓ^{-1/2}+2ℓ^{-1/2}\log k\Bigr\rvert_{√{3ℓ}}^ℓ \\
		&<	2ℓ^{-1/2}+2ℓ^{-1/2}\log ℓ = 2ℓ^{-1/2}+2ℓ^{-1/2+α}
		<	4ℓ^{-1/2+α} < ℓ^{-1/2+2α}.
	\end{align*}
	The last inequality uses $ℓ≥e^4$.
	This validates the second inequality in \eqref{eq:cramer}.
	And the proof of inequalities in \eqref{eq:cramer} is sound.

\subsection{The een13 Behavior} \label{pf:een13}

	We want to prove $𝘗\{𝘡_n<\exp(-e^{n^{1/3}})\}>1-𝘏_0-ℓ^{(-1/2+4α)n+o(n)}$,
	namely Inequality~\eqref{eq:een13}, given the en23 behavior,
	the supermartingale property, and the Cramér--Chernoff gadget.
	The idea is to apply the gadget consecutively to show that
	$𝘡_n$ becomes smaller and smaller as $n$ increases.
	To reach the goal $\exp(-e^{n^{1/3}})$,
	we apply $√n$ times to avoid losing too much code rate.
	
	(Define events.)
	Let $𝘌_0^0$ be the empty event.
	For every $m=√n,2√n…n-√n$, we define five series of events $𝘈_m$,
	$𝘉_m$, $𝘊_m$, $𝘌_m$, and $𝘌_0^m$ inductively as below:
	Let $𝘈_m$ be $\{𝘡_m<\exp(-m^{2/3})\}、𝘌_0^{m-√n}$.
	Let $𝘉_m$ be a subevent of $𝘈_m$ where $𝘡_k≥ℓ^{-8}$ for some $k≥m$.
	Let $𝘊_m$ a subevent of $𝘈_m$ where
	\[𝘋_{m+1}𝘋_{m+2}\dotsm𝘋_{m+√n}≤ℓ^{2α√n}. \label{eq:mgf2an}\]
	Let $𝘌_m$ be $𝘈_m、(𝘉_m∪𝘊_m)$.
	Let $𝘌_0^m$ be $𝘌_0^{m-√n}∪𝘌_m$.
	Let $𝘢_m$, $𝘣_m$, $𝘤_m$, $𝘦_m$, and $𝘦_0^m$
	be the probability measures of the corresponding capital letter events.
	Moreover, let $𝘨_m$ be $1-𝘏_0-𝘦_0^m$.
	
	(Bound $𝘣_m/𝘢_m$ from above.)
	Conditioning on $𝘈_m$, we want to estimate
	the probability that $𝘡_k≥ℓ^{-8}$ for some $k≥m$, which is equal to
	the probability that $\min(ℓ^{-2},√[4]{𝘡_k})≥ℓ^{-2}$ for some $k≥m$.
	Recall that $\min(ℓ^{-2},√[4]{𝘡_k})$ was made a supermartingale.
	Hence by Doob's optional stopping theorem \cite[Exercise~4.8.2]{Durrett19},
	$𝘗\{\min(ℓ^{-2},√[4]{𝘡_k})≥ℓ^{-2}† for some †k≥m｜𝘈_m\}
		≤\min(ℓ^{-2},√[4]{𝘡_m})ℓ^2<\exp(-m^{2/3}/4)ℓ^2$.
	This is an upper bound on $𝘣_m/𝘢_m$
	and will be summoned in Formula~(\ref{eq:sumeen13}).
	
	(Bound $𝘤_m/𝘢_m$ from above.)
	We want to estimate how often does Inequality~\eqref{eq:mgf2an} happen.
	It is the probability of
	$(𝘋_{m+1}𝘋_{m+2}\dotsm𝘋_{m+√n})^{-1/2}≥ℓ^{-α√n}$.
	This probability does not exceed
	$𝘌[(𝘋_{m+1}𝘋_{m+2}\dotsm𝘋_{m+√n})^{-1/2}]ℓ^{α√n}
		=𝘌[𝘋_1^{-1/2}]^{√n}ℓ^{α√n}=(𝘌[𝘋_1^{-1/2}]ℓ^α)^{√n}≤ℓ^{(-1/2+3α)√n}$
	by Markov's inequality.
	This is an upper bound on $𝘤_m/𝘢_m$
	and will be summoned in Formula~(\ref{eq:sumeen13}).
	
	(Bound $(𝘨_{m-√n}-𝘢_m)^+$ from above.)
	By definition, $𝘨_{m-√n}-𝘢_m=1-𝘏_0-(𝘦_0^{m-√n}+𝘢_m)$.
	The definition of $𝘈_m$ forces it to be disjoint from $𝘌_0^{m-√n}$,
	therefore $𝘦_0^{m-√n}+𝘢_m$ is the probability measure of $𝘌_0^{m-√n}∪𝘈_m$.
	This union event must contain the event $\{𝘡_m<\exp(-m^{2/3})\}$
	by how $𝘈_m$ was defined.
	From the en23 behavior
	$𝘗\{𝘡_m<\exp(-m^{2/3})\}>1-𝘏_0-ℓ^{(-1/2+4α)m}.$
	Chaining all inequalities together,
	we deduce that $𝘨_{m-√n}-𝘢_m<ℓ^{(-1/2+4α)m+o(m)}$.
	Let $(𝘨_{m-√n}-𝘢_m)^+$ be $\max(0,𝘨_{m-√n}-𝘢_m)$
	so we can write $(𝘨_{m-√n}-𝘢_m)^+<ℓ^{(-1/2+4α)m+o(m)}$.
	This upper bound will be summoned in Formula~(\ref{eq:sumeen13}).
	
	(Bound $𝘦_0^n$ from below.)
	We start rewriting $𝘨_m$ with $𝘨_m^+$ being $\max(0,𝘨_m)$:
	\begin{align*}
		&{}	𝘨_m = 1-𝘏_0-𝘦_0^m = 1-𝘏_0-(𝘦_0^{m-√n}+𝘦_m) = 𝘨_{m-√n}-𝘦_m \\
		&=	𝘨_{m-√n}（1-÷{𝘦_m}{𝘢_m}）+÷{𝘦_m}{𝘢_m}(𝘨_{m-√n}-𝘢_m)
		≤	𝘨_{m-√n}^+（1-÷{𝘦_m}{𝘢_m}）+÷{𝘦_m}{𝘢_m}(𝘨_{m-√n}-𝘢_m)^+\! \\
		&≤	𝘨_{m-√n}^+（1-÷{𝘦_m}{𝘢_m}）+(𝘨_{m-√n}-𝘢_m)^+
		≤	𝘨_{m-√n}^+（÷{𝘣_m}{𝘢_m}+÷{𝘤_m}{𝘢_m}）+(𝘨_{m-√n}-𝘢_m)^+ \\
		&<	𝘨_{m-√n}^+（\exp(-m^{2/3}/4)ℓ^2+ℓ^{(-1/2+3α)√n}）
			+ℓ^{(-1/2+4α)m+o(m)}
			\steplabel{eq:sumeen13}
	\end{align*}
	The first three equalities are by the definitions of $𝘨_m$ and $𝘌_0^m$.
	The next equality is simple algebra.
	The next two inequalities are by $0≤𝘦_m/𝘢_m≤1$.
	The next inequality is by the definition of $𝘌_m$.
	The last inequality summons upper bounds derived in the last few paragraphs.
	The last line contains two terms in the big parentheses.
	Between them $ℓ^{(-1/2+3α)√n}$ dominates $\exp(-m^{2/3}/4)ℓ^2$
	once $m$ is greater than $O(n^{3/4})$.
	Subsequently, we obtain this recurrence relation
	\[*\begin{cases}
		𝘨_{O(n^{3/4})}≤1; \\
		𝘨_m≤2𝘨_{m-√n}^+ℓ^{(-1/2+4α)√n}+ℓ^{(-1/2+4α)m+o(m)}.
	\end{cases}\]*
	Solve it; we get $𝘨_{n-√n}<ℓ^{(-1/2+4α)n+o(n)}$.
	By the relation between $𝘦_{n-√n}$ and $𝘨_{n-√n}$,
	we immediately get $𝘦_0^{n-√n}>1-𝘏_0-ℓ^{(-1/2+4α)n+o(n)}$.
	
	(Analyze $𝘌_0^{n-√n}$.)
	We want to estimate $𝘏_n$ when $𝘌_0^{n-√n}$ happens.
	More precisely, we attempt to bound $𝘡_{m+√n}$
	when $𝘌_m$ happens for each $m=√n,2√n…n-√n$.
	When $𝘌_m$ happens, its superevent $𝘈_m$ happens,
	so we know that $𝘡_m<\exp(-m^{2/3})$.
	But $𝘉_m$ does not happen, so $𝘡_k<ℓ^{-8}$ for all $k≥m$.
	This implies that $𝘡_{k+1}≤𝘡_k^{𝘋_{k+1}}$ for those $k$.
	Telescope;
	$𝘡_{m+√n}$ is less than $𝘡_m$ raised to
	the power of $𝘋_{m+1}𝘋_{m+2}\dotsm𝘋_{m+√n}$.
	But $𝘊_m$ does not happen, so the product is greater than $ℓ^{2α√n}$.
	Jointly we have $𝘡_{m+√n}≤𝘡_m^{ℓ^{2α√n}}<\exp(-m^{2/3}ℓ^{2α√n})$.
	Recall that $𝘡_{k+1}≤ℓeq𝘡_k$ for all $k≥m+√n$ so long as
	$𝘡_k$ stays below $ℓ^{-8}$, which it does because $𝘉_m$ is excluded.
	Then telescope again;
	$𝘡_n≤(ℓeq)^{n-m-√n}𝘡_{m+√n}<(ℓeq)^n\exp(-m^{2/3}ℓ^{2α√n})
		<\exp(-e^{n^{1/3}})$ provided that $n$ is sufficiently large.
	In other words, $𝘌_0^{n-√n}$ implies $𝘡_n<\exp(-e^{n^{1/3}})$.
	
	(Summary.)
	Now we conclude that
	$𝘗\{𝘡_n<\exp(-e^{n^{1/3}})\}≥𝘗(𝘌_0^{n-√n})=𝘦_0^n>1-𝘏_0-ℓ^{(-1/2+4α)n+o(n)}$.
	And hence the proof of the een13 behavior,
	Inequality~\eqref{eq:een13}, is sound.
	
	This subsection is parallel to \cite[Section~V]{LargeDeviations18}.
	Do not confuse this subsection with the next.
	The subtlety is explained in \cite[Section~III]{LargeDeviations18}.

\subsection{The elpin behavior} \label{pf:elpin}

	Recall $π,ρ>0$ is such that $π+2ρ≤1-8α$.
	We want to prove $𝘗\{𝘡_n<\exp(-ℓ^{πn}n^2)\}>1-𝘏_0-ℓ^{-ρn+o(n)}$,
	namely Inequality~\eqref{eq:elpin}, given the een13 behavior,
	the supermartingale property, and the Cramér--Chernoff gadget.
	The idea is to apply the gadget consecutively to show that
	$𝘡_n$ becomes smaller and smaller as $n$ increases.
	To reach the goal $\exp(-ℓ^{πn})$,
	we apply as many times as possible before we run out of depth $n$.
	
	(Define events.)
	Let $𝘈_0^0$ and $𝘌_0^0$ be the empty event.
	For every $m=√n,2√n…n-√n$, we define six series of events
	$𝘈_m$, $𝘈_0^m$, $𝘉_m$, $𝘊_m$, $𝘌_m$, and $𝘌_0^m$
	inductively as follows:
	Let $𝘈_m$ be $\{𝘡_m<\exp(-e^{m^{1/3}})\}、𝘈_0^{m-√n}$.
	Let $𝘈_0^m$ be $𝘈_0^{m-√n}∪𝘈_m$.
	Let $𝘉_m$ be a subevent of $𝘈_m$ where $𝘡_k≥ℓ^{-8}$ for some $k≥m$.
	Let $𝘊_m$ a subevent of $𝘈_m$ where
	\[𝘋_{m+1}𝘋_{m+2}\dotsm𝘋_n≤ℓ^{πn}. \label{eq:mgfpin}\]
	Let $𝘌_m$ be $𝘈_m、(𝘉_m∪𝘊_m)$.
	Let $𝘌_0^m$ be $𝘌_0^{m-√n}∪𝘌_m$.
	Let $𝘢_m$, $𝘢_0^m$, $𝘣_m$, $𝘤_m$, $𝘦_m$, and $𝘦_0^m$
	be the probability measures of the corresponding capital letter events.
	Moreover, let $𝘧_m$ be $1-𝘏_0-𝘢_0^m$ and let $𝘨_m$ be $1-𝘏_0-𝘦_0^m$.
	
	(Bound $𝘣_m/𝘢_m$ from above.)
	Conditioning on $𝘈_m$, we want to estimate
	the probability that $𝘡_k≥ℓ^{-8}$ for some $k≥m$, which is equal to
	the probability that $\min(ℓ^{-2},√[4]{𝘡_k})≥ℓ^{-2}$ for some $k≥m$.
	Recall that $\min(ℓ^{-2},√[4]{𝘡_k})$ was made a supermartingale.
	Hence by Doob's optional stopping theorem \cite[Exercise~4.8.2]{Durrett19},
	$𝘗\{\min(ℓ^{-2},√[4]{𝘡_k})≥ℓ^{-2}† for some †k≥m｜𝘈_m\}
		≤\min(ℓ^{-2},√[4]{𝘡_m})ℓ^2<\exp(-e^{m^{1/3}}/4)ℓ^2$.
	This is an upper bound on $𝘣_m/𝘢_m$
	and will be summoned in Formula~(\ref{eq:sumelpin}).
	
	(Bound $𝘤_m/𝘢_m$ from above.)
	We want to estimate how often does Inequality~\eqref{eq:mgfpin} happen.
	It is the probability of
	$(𝘋_{m+1}𝘋_{m+2}\dotsm𝘋_n)^{-1/2}≥ℓ^{-πn/2}$.
	By Markov's inequality, this probability is at most
	$𝘌[𝘋_1^{-1/2}]^{n-m}ℓ^{πn/2}<ℓ^{(-1/2+2α)(n-m)}ℓ^{πn/2}
	=ℓ^{(1/2-2α)m-(1/2-2α-π/2)n}≤ℓ^{(1/2-2α)m-(ρ+2α)n}$.
	The last inequality uses Inequality~\eqref{eq:finiteif}, $π+2ρ≤1-8α$.
	This is an upper bound on $𝘤_m/𝘢_m$
	and will be summoned in Formula~(\ref{eq:sumelpin}).
	
	(Bound $𝘧_m^+$ from above.)
	The definition of $𝘧_m$ reads $1-𝘏_0-𝘢_0^m$.
	Here $𝘢_0^m$ is the probability measure of $𝘈_0^m$,
	and $𝘈_0^m$ is a superevent of $𝘈_m$ by how the former is defined.
	Event $𝘈_0^m$ must contain $\{𝘡_m<\exp(-e^{m^{1/3}})\}$
	by how $𝘈_m$ was defined.
	By the een13 behavior,
	$𝘗\{𝘡_m<\exp(-e^{m^{1/3}})\}>1-𝘏_0-ℓ^{(-1/2+4α)m+o(m)}.$
	Chaining all inequalities together, we infer that $𝘧_m<ℓ^{(-1/2+4α)m+o(m)}$.
	Let $𝘧_m^+$ be $\max(0,𝘧_{m+√n})$ so
	we can write $𝘧_m^+<ℓ^{(-1/2+4α)m+o(m)}$.
	This upper bound will be summoned in Formula~(\ref{eq:sumelpin}).
	
	(Bound $𝘦_0^n$ from below.)
	We start rewriting $𝘨_m-𝘧_m^+$ with
	$(𝘧_{m-√n}-𝘢_m)^+$ being $\max(0,𝘧_{m-√n}-𝘢_m)$:
	\begin{align*}
		&{} 𝘨_m-𝘧_m^+
		=	1-𝘏_0-𝘦_m-(1-𝘏_0-𝘢_m)^+
		=	𝘨_{m-√n}-𝘦_m-(𝘧_{m-√n}-𝘢_m)^+ \\
		&≤	𝘨_{m-√n}-𝘦_m-÷{𝘦_m}{𝘢_m}(𝘧_{m-√n}-𝘢_m)^+
		≤	𝘨_{m-√n}-𝘦_m-÷{𝘦_m}{𝘢_m}(𝘧_{m-√n}^+-𝘢_m) \\
		&=	𝘨_{m-√n}-𝘧_{m-√n}^++𝘧_{m-√n}^+（1-÷{𝘦_m}{𝘢_m}） 
		≤	𝘨_{m-√n}-𝘧_{m-√n}^++𝘧_{m-√n}^+（÷{𝘣_m}{𝘢_m}+÷{𝘤_m}{𝘢_m}） \\
		&<	𝘨_{m-√n}-𝘧_{m-√n}^++{} \\
		&\qquad	+ℓ^{(-1/2+4α)(m-√n)+o(m-√n)}（\exp(-e^{m^{1/3}}/4)ℓ^2
			+ℓ^{(1/2-2α)m-(ρ+2α)n}） \steplabel{eq:sumelpin}
	\end{align*}
	The first two equalities are by the definitions of $𝘨_m$ and $𝘧_m$.
	The next inequality is by $0≤𝘦_m/𝘢_m≤1$.
	The next inequality is by $\max(0,f-a)=\max(a,f)-a≥\max(0,f)-a$.
	The next equality is simple algebra.
	The next inequality is by the definition of $𝘌_m$.
	The last inequality summons upper bounds derived in the last few paragraphs.
	Now the last line contains two terms in the big parentheses.
	Between them, $ℓ^{(1/2-2α)m-(ρ+2α)n}$ dominates
	$\exp(-e^{m^{1/3}}/4)ℓ^2$ once $n→∞$.
	Subsequently, we obtain this recurrence relation
	\[*\begin{cases}
		𝘨_0-𝘧_0^+=0; \\
		𝘨_m-𝘧_m^+≤𝘨_{m-√n}-𝘧_{m-√n}^++2ℓ^{-ρn+o(n)}.
	\end{cases}\]*
	Solve it; we get $𝘨_{n-√n}-𝘧_{n-√n}^+<ℓ^{-ρn+o(n)}$.
	Once again we summon $𝘧_{n-√n}^+<ℓ^{(-1/2+4α)(n-√n)+o(n)}<ℓ^{-ρn+o(n)}$;
	therefore $𝘨_{n-√n}<ℓ^{-ρn+o(n)}$.
	By the relation between $𝘦_{n-√n}$ and $𝘨_{n-√n}$
	we immediately get $𝘦_0^{n-√n}>(1-𝘏_0)-ℓ^{-ρn+o(n)}$.
	
	(Analyze $𝘌_0^{n-√n}$.)
	We want to estimate $𝘡_n$ when $𝘌_0^{n-√n}$ happens.
	More precisely, we attempt to bound $𝘡_n$
	when $𝘌_m$ happens for each $m=√n,2√n…n-√n$.
	When $𝘌_m$ happens, its superevent $𝘈_m$ happens,
	so we know that $𝘡_m<\exp(-e^{m^{1/3}})$.
	But $𝘉_m$ does not happen, so $𝘡_k<ℓ^{-8}$ for all $k≥m$.
	This implies $𝘡_{k+1}≤𝘡_k^{𝘋_{k+1}}$ for those $k$.
	Telescope;
	$𝘡_n$ is less than $𝘡_m$ raised to the power of $𝘋_{m+1}𝘋_{m+2}\dotsm𝘋_n$.
	But $𝘊_m$ does not happen, so the product is greater than $ℓ^{πn}$.
	Jointly we have
	$𝘡_n≤𝘡_m^{ℓ^{πn}}<\exp(-e^{m^{1/3}}ℓ^{πn})<\exp(-ℓ^{πn}n^2)$.
	In other words, $𝘌_0^{n-√n}$ implies $𝘡_n<\exp(-ℓ^{πn}n^2)$.
	
	(Summary.)
	Now we conclude that
	$𝘗\{𝘡_n<\exp(-ℓ^{πn}n^2)\}≥𝘗(𝘌_0^{n-√n})=𝘦_0^n>1-𝘏_0-ℓ^{-ρn+o(n)}$.
	And hence the proof of the elpin behavior,
	Inequality~\eqref{eq:elpin}, is sound.
	
	This subsection is parallel to \cite[Section~VI]{LargeDeviations18}.
	Do not confuse this subsection with the previous.
	The subtlety is explained in \cite[Section~III]{LargeDeviations18}.
	
	As we finish proving Inequalities \eqref{eq:eigen} to~\eqref{eq:elpin},
	we finish the proof of \Cref{lem:triforce}.
	\Cref{lem:triforce,lem:LDP,lem:CLT} are all finished.
	This is the last sentence of the proof of the main theorem.

\section{Constants Dependence Summary}

	Given a discrete memoryless channel $W$.
	The sender chooses the message alphabet size $ς≥2$.
	Depending on the factorization of $ς$,
	we choose $q$ to be a certain prime power or alternate
	between $q_2,q_3,q_5,\dotsc$ (a finite list depending on $ς$).
	Fix a $q$.
	Given $π,ρ>0$ such that $π+2ρ<1$;
	fix them.
	Choose $ℓ$;
	this also determines $α≔\log(\logℓ)/\logℓ$.
	The choice of $ℓ$ is such that $π+2ρ≤1-8α$ and such that
	the failing probabilities in \Cref{lem:LDP,lem:CLT} do not sum to one.
	It depends on $q,π,ρ$.
	Once $ℓ$ is fixed, the complexity is a function in $n$ (or in $N=ℓ^n$).
	The asymptotic complexity $O(N\log N)$ hides
	the scalar term that is determined by $q$ and $ℓ$.
	The decaying gap $ℓ^{-ρn+o(n)}$ in
	\Cref{cla:trichotomy,lem:triforce} hides two things:
	A scalar term in front of $ℓ$ determined by $q$ and $ℓ$ alongside with
	a $O(n^{1-ε})$ term determined by the choice of en23 and een13 checkpoints.
	This $ε$ is fixed throughout the paper and is irrespective of $ς,π,ρ,q,ℓ$.

\makeatletter
\hbadness879
\g@addto@macro\sloppy{
	\advance\baselineskip\glueexpr0ptplus1ptminus1pt/32
	\advance\parskip     \glueexpr0ptplus1ptminus1pt/16}
\bibliographystyle{alpha}
\bibliography{Hypotenuse-2}

\end{document}